\documentclass[11pt,letterpaper]{article}
\usepackage{cmap}
\usepackage[utf8x]{inputenc}
\usepackage[T1]{fontenc}
\usepackage{amsthm}
\usepackage{listings}
\usepackage{multirow}
\usepackage{amsfonts}
\usepackage{parskip}
\usepackage{enumitem}

\usepackage{caption}
\usepackage{subcaption}

\usepackage[pdftex]{graphicx} %
\usepackage[english]{babel} %
\usepackage[colorlinks,citecolor=blue,linkcolor=red]{hyperref}
\hypersetup{pdftitle={Random Quantum Circuits Beyond Moment Matching},pdfauthor={Shih-Han Hung}}
\usepackage{calc} %
\usepackage{enumitem} %
\setlist{nosep}

\usepackage[margin=0.94in]{geometry}

\usepackage[all]{nowidow} %
\usepackage[protrusion=true,expansion=true]{microtype} %

\usepackage{svg}
\usepackage{amsmath}
\usepackage{pdfpages}

\usepackage{tikz}
\usetikzlibrary{quantikz2}

\newcommand{\eq}[1]{\hyperref[eq:#1]{(\ref*{eq:#1})}}
\renewcommand{\sec}[1]{\hyperref[sec:#1]{Section~\ref*{sec:#1}}}
\newcommand{\thm}[1]{\hyperref[thm:#1]{Theorem~\ref*{thm:#1}}}
\newcommand{\lem}[1]{\hyperref[lem:#1]{Lemma~\ref*{lem:#1}}}
\newcommand{\cor}[1]{\hyperref[cor:#1]{Corollary~\ref*{cor:#1}}}
\newcommand{\dfn}[1]{\hyperref[dfn:#1]{Definition~\ref*{dfn:#1}}}
\newcommand{\rmk}[1]{\hyperref[rmk:#1]{Remark~\ref*{rmk:#1}}}
\newcommand{\prop}[1]{\hyperref[prop:#1]{Proposition~\ref*{prop:#1}}}

\theoremstyle{definition}
\newtheorem{theorem}{Theorem}[section]
\newtheorem{lemma}[theorem]{Lemma}
\newtheorem{definition}[theorem]{Definition}
\newtheorem{proposition}[theorem]{Proposition}
\newtheorem{corollary}[theorem]{Corollary}

\newtheorem{remark}[theorem]{Remark}

\renewcommand{\ket}[1]{\vert #1\rangle}
\renewcommand{\bra}[1]{\langle #1\vert}
\newcommand{\Id}[0]{I}
\DeclareMathOperator{\op}{op}
\renewcommand{\i}[0]{\mathrm{i}}
\renewcommand{\epsilon}[0]{\varepsilon}
\newcommand{\M}[0]{\mathcal{M}}

\DeclareMathOperator{\diag}{diag}

\title{Random Quantum Circuits Beyond Moment Matching}
\author{Shih-Han Hung\footnote{Department of Electrical Engineering and Center for Quantum Science and Engineering, National Taiwan University. Email: shihhanh@ntu.edu.tw.}}

\date{}

\newcommand{\bit}[0]{\{0,1\}}
\DeclareMathOperator*{\Exp}{\mathbb{E}}
\DeclareMathOperator*{\Var}{\mathrm{Var}}
\DeclareMathOperator*{\Sd}{\mathrm{Sd}}
\renewcommand{\d}[0]{\mathrm{d}}
\DeclareMathOperator{\Beta}{Beta}
\DeclareMathOperator{\Dir}{Dir}
\DeclareMathOperator{\Uniform}{Uniform}
\newcommand{\poly}[0]{\mathrm{poly}}
\newcommand{\Haar}[0]{\mathrm{Haar}}
\newcommand{\R}[0]{\mathbb{R}}

\newcommand{\e}[0]{e}
\newcommand{\dd}[0]{\mathrm{d}}

\begin{document} %

\pagenumbering{roman} %
\pagestyle{empty}
\maketitle
\thispagestyle{empty}

\begin{abstract}
Random quantum circuits aim to efficiently reproduce the statistical properties of ideal random quantum evolution. One approach is to construct approximate unitary designs, which match the moments of Haar-random unitaries up to a prescribed order with controlled error. We establish quantitative guarantees for how accurately these designs reproduce the distribution of each output probability. We show that, for every strong $\varepsilon$-approximate unitary $k$-design on $n$ qubits, the distribution of each output probability is within $O(k^{-1/2}(1+k/2^n)^{-1/2}+\varepsilon)$ in Kolmogorov distance of the Porter--Thomas distribution, a beta distribution with parameters $1$ and $2^n-1$. This bound is optimal up to constant factors, implying that a substantially better uniform bound requires additional structure. We then show that a mild form of local invariance suffices. We prove that any such design that is invariant under unitary translations acting on a single qubit approximates Porter--Thomas to total variation distance $O(k^{-1/2}+\varepsilon)$. Every strong $\varepsilon$-approximate unitary $k$-design invariant on $\log k+O(1)$ qubits achieves error $2^{-\Omega(k)}+O(\varepsilon)$ in Kolmogorov distance and $2^{-\Omega(k)}+O(\varepsilon\log(2/\varepsilon))$ in total variation distance. These results identify the distributional accuracy guaranteed by moment matching and the structure that improves it. As an application, we show that the Shannon entropy of the output distribution of a strong $\varepsilon$-approximate $k$-design is concentrated around its Haar value $n-\Theta(1)$: its mean differs from the Haar mean by $O(k^{-1/4}+\varepsilon^{1/2})$ and its standard deviation is $O(2^{-n/2}+k^{-1/8}+\varepsilon^{1/4})$. Consequently, a trusted device sampling such circuits generates $n-O(1)$ bits of entropy per sample, which makes them useful for random number generation. 
\end{abstract}

\newpage
{
  \hypersetup{linkcolor=black}
  \tableofcontents
}
\newpage
\pagestyle{plain}
\pagenumbering{arabic} %

\section{Introduction}\label{sec:introduction}

After more than three decades of theoretical and experimental research in quantum information science, we are now in a new era in which practical quantum hardware can solve certain computational problems that are widely believed to be intractable for classical computers \cite{aruteQuantumSupremacyUsing2019,zhongQuantumComputationalAdvantage2020,wuStrongQuantumComputational2021,madsenQuantumComputationalAdvantage2022,aaronsonCertifiedRandomnessQuantum2023,decrossComputationalPowerRandom2025a,liuCertifiedRandomnessUsing2025a}. \ 

Among the approaches to leveraging the power of existing quantum hardware, a random circuit is an essential building block useful in near-term feasible quantum experiments \cite{aruteQuantumSupremacyUsing2019,wuStrongQuantumComputational2021,aaronsonCertifiedRandomnessQuantum2023,decrossComputationalPowerRandom2025a,liuCertifiedRandomnessUsing2025a,emersonScalableNoiseEstimation2005,knillRandomizedBenchmarkingQuantum2008,elbenMixedStateEntanglementLocal2020,elbenRandomizedMeasurementToolbox2023,brydgesProbingRenyiEntanglement2019,huangPredictingManyProperties2020,zhaoFermionicPartialTomography2021,movassaghHardnessRandomQuantum2023,boulandComplexityVerificationQuantum2019,morvanPhaseTransitionsRandom2024,fisherRandomQuantumCircuits2023,nahumOperatorSpreadingRandom2018,cotlerFluctuationsSubsystemEntropies2022,choiPreparingRandomStates2023,cotlerEmergentQuantumState2023}. \ 
Despite the fact that a practical random circuit must be statistically different from a truly random unitary---one distributed according to the Haar measure on the unitary group---researchers often draw intuitions from the latter to get a sense of its statistical properties in the most idealized manner. \ 
The study of random circuits is hence deeply connected to random matrix theory, whose mathematical framework provides powerful tools for their analysis. \ 

The demand for using a random circuit in near-term quantum experiments \cite{webbCliffordGroupForms2016,zhuMultiqubitCliffordGroups2017,aruteQuantumSupremacyUsing2019,huangPredictingManyProperties2020,zhaoFermionicPartialTomography2021,aaronsonCertifiedRandomnessQuantum2023,schusterRandomUnitariesExtremely2025,cuiUnitaryDesignsNearly2025} and the mathematical foundations in random matrix theory \cite{Fyodorov2005} have guided a line of research on constructing practical random circuits. 
However, a simple counting argument shows that implementing a Haar random unitary is prohibitively expensive in terms of both gate count and circuit depth. \ 
Hence for these experiments to work on a near-term device, it is essential to build a practical random circuit that uses as few quantum resources as possible, and that is provably close to certain statistical properties of a Haar random unitary. \ 

A key concept enabling this approximation is the $\epsilon$-approximate unitary $k$-design, which is provably sufficient for many applications \cite{brandaoEfficientQuantumPseudorandomness2016,brandaoLocalRandomQuantum2016,brandaoModelsQuantumComplexity2021,haferkampRandomQuantumCircuits2022a,schusterRandomUnitariesExtremely2025,cuiUnitaryDesignsNearly2025}. \ 
Let $\tilde\Phi$ be the quantum channel formed by $k$ parallel applications of the quantum circuit and $\Phi$ be the corresponding quantum channel with a Haar random unitary. \
A random circuit forms an $\epsilon$-approximate unitary $k$-design if the channels satisfy the following closeness condition:
\begin{align}\label{eq:unitary-design}
    (1-\epsilon)\Phi\leq \tilde\Phi \leq (1+\epsilon)\Phi,
\end{align}
where the ordering $\Phi\leq\Psi$ between two linear maps means that $\Psi-\Phi$ is a completely positive map.
This notion is also called a unitary design with relative error $\epsilon$, in contrast to one defined with respect to additive error. We use the former notion as in \eq{unitary-design} throughout this work.
More precisely, we work with the \emph{strong} variant of this notion \cite{schusterStrongRandomUnitaries2025}, in which the channel also acts with $k$ copies of the complex conjugate $U^*$, so that it controls algorithms that query both $U$ and $U^\dag$; see \dfn{strong-design} and \rmk{design}.
In tasks such as quantum state tomography and the demonstration of quantum computational advantage, a small value of $k$ already suffices, suggesting that these experiments can be implemented with shallow quantum circuits. \ 
Although applications of unitary $k$-designs with $k>3$ remain comparatively sparse \cite{cuiUnitaryDesignsNearly2025}, the study of improving the efficiency for asymptotically increasing $k$ remains active \cite{harrowRandomQuantumCircuits2009,brandaoLocalRandomQuantum2016,harrowApproximateUnitaryTDesigns2023,haferkampRandomQuantumCircuits2022a,chenIncompressibilitySpectralGaps2025,schusterRandomUnitariesExtremely2025,schusterStrongRandomUnitaries2025,parelladilmeStrongUnitaryDesigns2026}. \ 

A unitary $k$-design reproduces the moments of a Haar random unitary that are of degree $k$ in the matrix entries and of degree $k$ in their complex conjugates. \
The moments of a single entry, under the Haar measure, are integrals against the marginal distribution of that entry. \
While previous works have shown that practical constructions can reproduce these moments with small resource overhead, very little is known about the marginal distribution of these matrix elements themselves. \ 

This question becomes particularly relevant in the context of demonstrating quantum computational advantage. \ 
The linear cross-entropy benchmark (XEB), a standard benchmarking protocol for random circuit sampling experiments, presupposes that the output probabilities follow the Porter--Thomas distribution. \
Indeed, it has been a folklore conjecture, supported by empirical observations, that sampling from a random circuit reproduces outcome probabilities distributed according to the Porter--Thomas distribution, a beta distribution with parameters $(1,d-1)$ for a $d$-dimensional Haar random unitary \cite{boixoCharacterizingQuantumSupremacy2018}.
To see why it is beta distributed, recall how we sample a Haar random vector in a $d$-dimensional Hilbert space: \
one draws $d$ independent and identically distributed samples from the standard complex normal distribution to form a random vector, then normalizes it to unit norm. \ 
This random vector $\ket{\psi}$ forms a probability vector $p$ whose entries $p(x)=|\langle x|\psi\rangle|^2$ follow a beta distribution with shape parameters $(1,d-1)$, which approaches an exponential distribution of mean $1/d$ as $d$ grows large \cite{donatimartinRandomTruncationsHaar2013,aaronsonCertifiedRandomnessQuantum2023}. \ 

To make the question precise, recall that the convergence of probability measures can have different meanings in modern probability theory by considering different metrics between probability measures. 
In particular, a sequence of random variables converges to its limit in distribution if the distribution function converges to its limit pointwise \cite[Chapter~14.1, Definition~1]{fristedtModernApproachProbability1997d}. That is, at every continuity point $x\in\mathbb R$ of the limiting distribution function, the sequence of distribution functions at $x$ converges to its limit. 
This notion of convergence already guarantees convergence of expectations for every fixed bounded continuous test function \cite[Chapter~14.4, Proposition~7]{fristedtModernApproachProbability1997d}. 
Despite the guarantee on ``well-behaved'' test functions, the convergence rate can be very slow at particular $x$ in the support. 

A quantitative notion is uniform convergence of the distribution functions, which states that the Kolmogorov distance between the $k$th random variable in the sequence and its limit converges to zero as $k\to\infty$. 
Here, the Kolmogorov distance between two probability distributions $\mu,\nu$ is 
\begin{align}
    d_\mathrm{K}(\mu,\nu) = \sup_{z\in\mathbb R}\Big|\Pr_{X\sim\mu}[X\leq z]-\Pr_{Y\sim\nu}[Y\leq z]\Big|, 
\end{align}
i.e., the largest difference between the distribution functions.
This notion implies convergence in distribution and provides a rate that is uniform in $x\in\mathbb R$; for a continuous limit such as the Porter--Thomas distribution, pointwise convergence of the distribution functions already implies uniform convergence (P\'olya's theorem).
Hence, the two notions differ only quantitatively: what a Kolmogorov bound adds is an explicit rate of convergence that holds uniformly in $x$. 

One of the strongest notions regarding convergence of probability measures is the convergence in total variation, which implies the previous two notions we discussed. 
In particular, for two probability distributions $\mu,\nu$, the total variation distance is defined as 
\begin{align}
    d_{\mathrm{TV}}(\mu,\nu) = \sup_A \Big|\Pr_{X\sim\mu}[X\in A] - \Pr_{Y\sim\nu}[Y\in A]\Big|
\end{align}
where the supremum is over measurable sets.
A sequence converges to its limit in total variation if the total variation distance between the $k$th random variable and its limit converges to zero as $k\to\infty$. 
It also guarantees that the expectation of any bounded measurable function converges to its limit, uniformly over functions bounded by a fixed constant, as opposed to only bounded continuous functions under convergence in distribution. 
Conversely, convergence in Kolmogorov distance does not imply convergence in total variation, since one can construct a sequence of discrete measures that converges to a continuous distribution in Kolmogorov distance, while their total variation distance from it does not decrease.

Given that a unitary $k$-design matches the first $k$ moments of a Haar random unitary, a natural question arises: \ 
does a unitary $k$-design converge to Porter--Thomas as $k\to\infty$? \ 
This is a special form of a fundamental question: when are two distributions that agree on their first $k$ moments necessarily close to each other? \ 
This question is closely related to the moment problem, which studies when a probability distribution is uniquely determined by its moments. \ 

Indeed, in general, for a continuous distribution, one can always construct another discrete distribution that matches the first few moments through the framework of Gaussian quadrature and the theory of orthogonal polynomials. 
The two distribution functions must look different at the jumps. \ 
This sets a fundamental obstacle toward proving the closeness between a continuous distribution and its moment-matching counterpart. 
The expectation of functions that may not be approximated by a low-degree polynomial can be drastically different. \ 

Taking a step back, to our knowledge, no quantitative guarantee was known on how well the output probabilities of an efficient random circuit approximate their Haar counterparts in a sense beyond moment matching and anti-concentration. \ 
In this work, we address the fundamental problem, and show that strong unitary designs do approximate every single entry marginally to its Haar counterpart in Kolmogorov distance and, under a mild local invariance property, also in total variation distance. 
As an application of our new understanding of approximate designs, we show that the expectation and standard deviation of the Shannon entropy of the output distribution converge at a dimension-independent rate, which supports the use of these random quantum circuits for random number generation with a trusted device.

\subsection{Our results}

In this work, we provide a quantitative analysis of how fast an approximate design converges to Porter--Thomas, both in Kolmogorov distance and in total variation distance. 
First, we identify the closeness between the distribution functions in Kolmogorov distance for every strong $\epsilon$-approximate unitary $k$-design.
Throughout, $\mu(k,\epsilon)$ denotes the distribution of a strong $\epsilon$-approximate unitary $k$-design (a distribution on the unitary group), and for every $x,y\in\bit^n$, $\nu(k,\epsilon)$ denotes the distribution of the single entry $|\langle x|U|y\rangle|^2$ for $U\sim\mu(k,\epsilon)$; $\nu_{\Haar(2^n)}=\Beta(1,2^n-1)$ is its Haar counterpart. 
We also write $(k,k)$-design for a strong $k$-design, since the design condition involves $k$ copies of $U$ and $k$ copies of $U^*$.

\begin{theorem}[Informal statement of \thm{approximate-design-upper-bound}]\label{thm:ks-generic}
    For any strong $\epsilon$-approximate $(k,k)$-design on $n$ qubits and every fixed pair $x,y$, $d_{\mathrm K}(\nu(k,\epsilon), \nu_{\Haar(2^n)}) 
        = O\left(k^{-1/2}(1+k/2^n)^{-1/2}+\epsilon\right)$.
\end{theorem}

Since this result applies to any approximate design, continuous or discrete, it cannot provide any guarantee on the closeness in total variation distance. 

The rate in \thm{ks-generic} cannot be improved in general.
The proof of the upper bound compares the distribution functions at the nodes of a Gauss quadrature rule for the Porter--Thomas distribution; by a separate construction, which embeds such a quadrature measure into a unitary ensemble, we show that a Kolmogorov distance proportional to the largest quadrature weight is attained by some design.

\begin{theorem}[Informal statement of \thm{approximate-design-lower-bound}]\label{thm:ks-generic-lower}
    For every $n\geq 1$, $k\geq 1$, and $\epsilon\in[0,1]$, there exists a strong $\epsilon$-approximate $(k,k)$-design on $n$ qubits for which the Kolmogorov distance of a fixed output-probability marginal is $\Omega\left(k^{-1/2}(1+k/2^n)^{-1/2}+\epsilon\right)$.
    In particular, the bound in \thm{ks-generic} is optimal up to constant factors in the worst case.
\end{theorem}

Consequently, in all regimes of $k$ relative to the dimension $2^n$ of the unitary matrix, no guarantee that holds uniformly over all strong approximate designs can improve on the rate $1/k$ up to a constant, ruling out a superpolynomial convergence rate without further assumptions (individual designs, such as the Haar measure itself, may of course do better). 

We then show that mild forms of local invariance suffice for stronger convergence.
A design is $m$-locally invariant if applying any fixed unitary $V$ on a fixed set of $m$ qubits after the circuit does not change the distribution of the unitary ensemble. 
This invariance introduces a strong smoothing effect to $\nu(k,\epsilon)$: without it the distribution $\nu(k,\epsilon)$ can come from a discrete design whose distribution function is discontinuous at its atoms. 

\begin{theorem}[{Informal statement of \thm{one-local-invariance}}]\label{thm:informal-one-local-invariance}
    For every $1$-locally invariant strong $\epsilon$-approximate $(k,k)$-design on $n$ qubits,
    $d_{\mathrm{TV}}(\nu(k,\epsilon),\nu_{\Haar(2^n)}) \leq O(k^{-1/2}+\epsilon)$.
\end{theorem}

That is, local invariance on a single qubit already lifts the convergence to the total variation distance. 

Closeness in total variation controls every event concerning a single output probability: for every measurable $A\subseteq[0,1]$, the probability that the output probability lies in $A$ differs from its Haar value by at most $O(k^{-1/2}+\epsilon)$.
Likewise, the expectation of every bounded measurable function of a single output probability is reproduced up to this error times its supremum norm.

Local invariance on a single qubit is a property of the circuit ensemble rather than of the design notion, and it holds for every random circuit whose output layer consists of Haar random gates.
Indeed, consider an ensemble in which the last gate acting on a fixed qubit is drawn from the Haar measure on a block of qubits containing it, independently of the remaining gates.
By the left invariance of the Haar measure, applying any fixed unitary on that qubit after the circuit does not change the distribution of the ensemble, so the ensemble is $1$-locally invariant. 
For example, every brickwork circuit whose final layer consists of independent Haar random two-qubit gates is $1$-locally invariant on every qubit acted on by that final layer.

Many existing constructions of approximate unitary designs use Haar random two-qubit gates.
Local random circuits with Haar random two-qubit gates, in one dimension and on general architectures, form approximate designs \cite{harrowRandomQuantumCircuits2009,brandaoLocalRandomQuantum2016,hunterjonesUnitaryDesignsStatistical2019,brandaoModelsQuantumComplexity2021,haferkampRandomQuantumCircuits2022a,harrowApproximateUnitaryTDesigns2023,mittalLocalRandomQuantum2023,belkinApproximateTDesignsGeneric2024,chenIncompressibilitySpectralGaps2025}, as do the low-depth random circuits of \cite{schusterRandomUnitariesExtremely2025}, and recent constructions of strong unitary designs also use random circuits of Haar random two-qubit gates \cite{schusterStrongRandomUnitaries2025,parelladilmeStrongUnitaryDesigns2026}.
The layered constructions among them, such as brickwork circuits, end with a layer of Haar random two-qubit gates and are therefore $1$-locally invariant.
Sequential constructions, which apply one Haar random gate to a random pair of qubits in each step, need not be exactly $1$-locally invariant, but any approximate design can be made so by appending a Haar random gate on a single qubit, which preserves the strong design property.
We stress that most of these ensembles have only been shown to be designs in the standard sense, possibly with additive error, whereas \thm{informal-one-local-invariance} requires a strong design with relative error. 

Combining \thm{informal-one-local-invariance} with the strong design result of Schuster, Ma, Lombardi, Brand\~ao, and Huang \cite[Lemma~3]{schusterStrongRandomUnitaries2025}, we show that one-dimensional brickwork circuits of Haar random two-qubit gates approximate the Porter--Thomas distribution in total variation distance.
By their result, a brickwork circuit of depth $O(\log^7 k\,(nk+\log(1/\epsilon)))$ on $n$ qubits is a strong $\epsilon$-approximate $(k,k)$-design whenever $k\geq 2$ and $\epsilon\geq 8k^2/2^n$ (their bound is stated in terms of the total order $p+q$, which equals $2k$ for $(k,k)$-designs).
Since the circuit is $1$-locally invariant, the distribution of each output probability is within total variation distance $O(k^{-1/2}+\epsilon)$ of Porter--Thomas (\cor{brickwork}).
In particular, linear depth $O(n)$ suffices for any constant total variation distance, and, for every fixed $c$, depth $O(n^{2c+1}\log^7 n)$ suffices for total variation distance $n^{-c}$.

A natural question raised by \thm{informal-one-local-invariance} is whether local invariance on more qubits yields faster convergence.
We answer this question in the affirmative: strong designs that are locally invariant on $\log k+O(1)$ qubits converge to Porter--Thomas at a rate exponential in $k$. 
The proof is based on a Fourier analysis.
Local invariance allows us to insert an independent Haar random unitary on the invariant qubits without changing the ensemble, which multiplies the probability of a block of outcomes by an independent beta random variable.
This multiplication removes the high-frequency components of the distribution function and thereby smooths out the discontinuities. 
This yields the following bound in Kolmogorov distance.

\begin{theorem}[Informal statement of \thm{convergence-distribution} and \cor{locally-invariant-entry}]\label{thm:ks-locally-invariant}
Let $k$ be a positive integer, and let $m$ be such that $a=2^m\geq 2\pi k$ and $N=2^n>2a$.
    For an $m$-locally invariant strong $\epsilon$-approximate $(k,k)$-design on $n$ qubits, 
    $d_{\mathrm{K}}(\nu(k,\epsilon),\nu_{\Haar(2^n)}) \leq 2^{-\Omega(k)} + O(\epsilon)$.
\end{theorem}

Local invariance on $m$ qubits is a mild assumption when $m$ is small enough to allow for efficient sampling: as above, a circuit whose last layer contains a Haar random gate on $m$ qubits is automatically $m$-locally invariant on those qubits, and any approximate design can be made $m$-locally invariant by appending an $m$-qubit Haar random unitary. 
As analyzed in \sec{locally-invariant-properties}, a $d$-dimensional Haar random unitary can be replaced by an efficiently sampled circuit of $d^2(\log d + \log(1/\epsilon))^{O(1)}$ gates at the cost of an additional $O(\epsilon)$ in Kolmogorov distance for the squared magnitude of any fixed matrix element; since the resulting ensemble is discrete, this replacement preserves the Kolmogorov guarantee but not the total variation guarantee below.

By \thm{ks-locally-invariant} and the aforementioned smoothing effects from a locally invariant design, we can further lift the closeness to the stronger notion of convergence in total variation. 

\begin{theorem}[Informal statement of \thm{convergence-total-variation} and \cor{locally-invariant-entry}]\label{thm:tv-locally-invariant}
    Let $k$ be a positive integer, and let $m$ be such that $a=2^m\geq 2\pi k$ and $N=2^n>2a$.
    For an $m$-locally invariant strong $\epsilon$-approximate $(k,k)$-design on $n$ qubits, 
    $d_{\mathrm{TV}}(\nu(k,\epsilon),\nu_{\Haar(2^n)}) \leq 2^{-\Omega(k)} + O(\epsilon\log(2/\epsilon))$.
\end{theorem}

At $\epsilon=0$, the second term is interpreted as zero. 
The dependence on $\epsilon$ is slightly weaker than in \thm{ks-locally-invariant}; the logarithmic factor arises from the truncation used to pass from pointwise closeness of the densities to closeness in $L^1$ norm.

As an application of the aforementioned results, we prove that strong unitary designs can be used to generate random numbers. 
Fix an input $y\in\bit^n$, and for every $x\in\bit^n$ let $P_x(U) = |\bra{x} U \ket{y}|^2$ be the probability of sampling $x$ given $U$.
The Shannon entropy of the output distribution is defined as 
    $S(U) = -\sum_{x\in\bit^n} P_x(U) \log P_x(U)$.
\sec{random-number} relates the mean and the standard deviation of the entropy to the closeness of marginal distributions. 

\begin{theorem}[Informal statement of \thm{shannon-entropy-mean} and \thm{shannon-entropy-variance}]\label{thm:informal-entropy}
    For $U\sim\mu(k,\epsilon)$, if the marginal distribution of every output probability $P_x(U)$ is $\eta$-close to its Haar counterpart in Kolmogorov distance, then the expected Shannon entropy differs from its Haar value by $O(\eta^{1/2})$. 
    In addition, if the law of every nonnegative linear combination $sP_x+tP_{x'}$ of two output probabilities is $\delta$-close to its Haar counterpart in Kolmogorov distance, then the standard deviation of the Shannon entropy differs from its Haar value by $O(\delta^{1/4})$.
    For $k\geq 2$, one can take $\delta=O(k^{-1/2}+\epsilon)$.
\end{theorem}

The Haar value of the expected Shannon entropy is $n-\Theta(1)$, more precisely $n-(1-\gamma)/\ln 2+o(1)$ with $\gamma$ the Euler--Mascheroni constant \cite{wolpertEstimatingFunctionsProbability1995,boixoCharacterizingQuantumSupremacy2018}, so \thm{informal-entropy} shows that the output distribution of a strong approximate design carries $n-O(1)$ bits of entropy in expectation as soon as $\eta=O(1)$.
While the Shannon entropy is a global quantity that depends on $2^n$ entries of $U$, our proof bounds the difference using only entrywise or pairwise closeness in Kolmogorov distance. 

The bound on the standard deviation implies that the entropy of a single circuit drawn from the design is concentrated.
The Haar standard deviation is $\Theta(2^{-n/2})$ \cite{wolpertEstimatingFunctionsProbability1995,nemenmanEntropyInferenceRevisited2002}, so the standard deviation of the Shannon entropy of the design is $O(2^{-n/2}+\delta^{1/4})$.
By Chebyshev's inequality, the Shannon entropy of a circuit drawn from the design deviates from its Haar value by more than $t+O(\eta^{1/2})$ with probability at most $O\big((2^{-n}+\delta^{1/2})/t^2\big)$.
Although the bound $O(2^{-n/2}+\delta^{1/4})$ on the standard deviation exceeds its Haar value by far unless $k$ is exponentially large and $\epsilon$ exponentially small in $n$, it tends to zero as $k\to\infty$ and $\epsilon\to 0$.
Thus not only the expected entropy but also the entropy of a typical circuit from the design is nearly maximal.
In particular, a circuit that is sampled once and then reused also produces nearly maximal entropy.

We now combine \thm{informal-entropy} with \thm{ks-generic} and \thm{ks-locally-invariant}.
For a generic strong $\epsilon$-approximate $k$-design, the expected Shannon entropy is within $O(k^{-1/4}+\epsilon^{1/2})$ of its Haar value $n-\Theta(1)$.
Moreover, a single circuit drawn from the design has Shannon entropy $n-O(1)$ with probability $1-O(2^{-n}+k^{-1/4}+\epsilon^{1/2})$; for instance, this probability is $1-O(1/n)$ when $k=\Omega(n^4)$ and $\epsilon=O(1/n^2)$.
For a design that is locally invariant on $\log k+O(1)$ qubits, the expected Shannon entropy is within $2^{-\Omega(k)}+O(\epsilon^{1/2})$ of its Haar value. 

Such designs can be implemented efficiently.
As discussed after \thm{informal-one-local-invariance}, one-dimensional brickwork circuits of depth $O(\log^7 k\,(nk+\log(1/\epsilon)))$, using $O(n\log^7 k\,(nk+\log(1/\epsilon)))$ gates, are strong $\epsilon$-approximate $k$-designs with relative error for $\epsilon\geq 8k^2/2^n$ \cite[Lemma~3]{schusterStrongRandomUnitaries2025}.
For instance, a brickwork circuit of depth $O(n^5\log^7 n)$ has Shannon entropy $n-O(1)$ with probability $1-O(1/n)$.

To extract randomness from such random circuits, we sample a fresh circuit in each round and measure its output.
Since the rounds are independent, the asymptotic equipartition property \cite{tomamichelFullyQuantumAsymptotic2009}
already shows that $M$ rounds accumulate $M(h-o(1))$ bits of smooth min-entropy, where $h$ is the expected Shannon entropy per round.
A randomness extractor then turns the outcomes into nearly uniform bits. 
However, when verifiability is a concern, whether this random quantum circuit can be resistant to an adversary and robustly generate random numbers requires further study.

\subsection{Technical overview}

We provide a technical overview of the proofs of our main results. 

\subsubsection{Bounding distribution functions by expected polynomials}\label{sec:overview-upper-bound}

Let $N=2^n$. 
Recall from \eq{unitary-design} and \dfn{strong-design} that a strong $\epsilon$-approximate unitary $k$-design reproduces, up to a relative error $\epsilon$, the expectation of $|f(U)|^2$ for every polynomial $f$ of degree $k$ in the entries of $U$ and degree $k$ in their complex conjugates. 
Our goal is to compare an individual output probability $P=|U_{xy}|^2$ with its Haar counterpart.
Our starting point is therefore the relative error control of the expectation of $|p(P)|^2$ for every polynomial $p$ of degree at most $k$ (\sec{block-encoding}). 
To establish closeness in Kolmogorov distance, we must turn this control into a uniform bound on the difference of the two distribution functions.
This is the main challenge we tackle in our analysis.
Recall that the distribution function at $x\in\mathbb R$ is the expectation of the indicator function $I_x(t)=1_{t\leq x}$. 
This function is discontinuous and thus cannot be uniformly approximated by polynomials. 

We address this challenge with polynomial majorants and minorants of the indicator function, together with a discrete measure supported on a grid of $k+1$ points that reproduces the Haar expectations of these polynomials.
Since the Haar marginal is a beta measure, a natural choice of the grid is the set of roots of an orthogonal polynomial with respect to this beta measure. 
These orthogonal polynomials are shifted Jacobi polynomials, so the grid consists of the nodes of a shifted Jacobi--Gauss quadrature rule.
The quadrature rule with $k+1$ nodes is exact for every polynomial of degree at most $2k+1$.
Hence the quadrature measure matches the expectations of these polynomials exactly under the Haar measure and, by the design condition, up to a relative error $\epsilon$ under the design. 

At each node, we construct a polynomial majorant and a polynomial minorant of the corresponding step function, both of degree at most $2k$. 
The majorant is nonnegative on the real line, so it is a sum of squares of polynomials of degree at most $k$.
The strong design condition therefore transfers its expectation with relative error $\epsilon$. 
The minorant differs from the majorant by the square of a Lagrange interpolation polynomial, whose Haar expectation equals the weight of that node.
These facts bound the difference between the two distribution functions at the quadrature nodes by $O(w_{\max}+\epsilon)$, where $w_{\max}$ is the largest quadrature weight.
By the monotonicity of distribution functions, the same bound holds at every $x\in\mathbb R$.

It remains to bound $w_{\max}$, which we do through the Christoffel function of the beta measure.
The Chebyshev--Markov--Stieltjes inequality \cite{akhiezerClassicalMomentProblem1965,lindsayMomentsDetermineTail2000,gilboaDifferentialVersionChebyshev2015} states that, for two measures whose moments agree up to order $2k$, the difference of the distribution functions at $x\in\mathbb R$ is at most the Christoffel function $\lambda(x)$ of degree $k$ of either measure, and our node-based argument is an approximate version of this inequality.
The Christoffel function admits a variational form: $\lambda(x)$ is the minimum of the expectation of $p^2$ over all polynomials $p$ with $\deg(p)\leq k$ and $p(x)=1$. 
At a quadrature node $x_i$, $\lambda(x_i)$ equals the quadrature weight $w_i$.
The problem thus reduces to bounding the Christoffel function at the nodes.

By the variational form, if the beta measure is dominated by a constant multiple of another measure, then its Christoffel function is dominated by the same multiple of the Christoffel function of that measure.
We compare with two measures \cite{szegoOrthogonalPolynomials1975}: the exponential measure, which is the limit of the beta measure rescaled by $N$ as $n\to\infty$, and the arcsine (Chebyshev) measure, which captures the behavior of the quadrature weights when $N$ is small compared with $k$.
The former gives $w_{\max}=O(1/\sqrt k)$ and the latter gives $w_{\max}=O(\sqrt N/k)$. 
Combining the two bounds covers all regimes of $k$ relative to $2^n$ and yields $w_{\max}=O(\min\{\sqrt N/k,1/\sqrt k\})=O(k^{-1/2}(1+k/N)^{-1/2})$.
Together with the node bound $O(w_{\max}+\epsilon)$, this proves the Kolmogorov distance bound in \thm{ks-generic}.

\subsubsection{Embedding quadrature measures into unitary designs}

The same quadrature that gives the upper bound as discussed in \sec{overview-upper-bound} also provides a distribution that saturates the bound up to a constant factor.
A discrete distribution with an atom of weight $w$ has Kolmogorov distance at least $w/2$ from any continuous distribution, since across the jump, a continuous distribution function cannot be close to both one-sided values. 

A remaining issue is that a distribution matching the moments of one output probability need not automatically arise from a unitary design. 
We resolve this issue by explicitly embedding a scalar random variable $X\in[0,1]$ into a unitary ensemble with $|U_{00}|^2=X$. 
The same argument extends to $|U_{xy}|^2$ for any fixed $x,y\in\bit^n$ since both the Haar measure and the design condition are invariant under permutations of the computational basis. 
The construction randomizes the phase of the distinguished entry $(x,y)$ and uses independent Haar random unitaries on the remaining $N-1$ dimensions.
When $X$ has the beta distribution, the resulting unitary is Haar distributed, and when $X$ is sampled from the quadrature measure, the resulting unitary is a design. 

Knowing that every moment-matching distribution can be embedded into a unitary design, it suffices to lower bound the Kolmogorov distance between $\Beta(1,N-1)$ and its Jacobi--Gauss quadrature with $k+1$ nodes. 
Let $\kappa:=k^{1/2}(1+k/N)^{1/2}$, so that the target bound is $\Omega(1/\kappa)$.
We show that the interval $[0,1/N]$ contains at most $O(\kappa)$ nodes. 
To this end, we compare the nodes with the roots of the Legendre polynomial of the same degree, which is the Jacobi polynomial with parameters $(0,0)$.
By Bruns' inequality \cite[Theorem~6.21.2]{szegoOrthogonalPolynomials1975}, the $j$th smallest root of the shifted Legendre polynomial on $[0,1]$ is $\Omega(j^2/k^2)$. 
Our nodes are the roots of a shifted Jacobi polynomial with parameters $(0,N-2)$. 
To pass from $(0,0)$ to $(0,N-2)$, we apply a monotonicity result of Dimitrov and Rafaeli \cite{dimitrovMonotonicityZerosJacobi2007},
which bounds how far the roots move toward zero as the second parameter increases from $0$ to $N-2$, and obtain $x_j=\Omega(j^2/(k(k+N)))$. 
Hence $x_j\leq 1/N$ forces $j=O(\kappa)$. 
On the other hand, $\Pr[X\leq 1/N]\geq 1/2$ for $X\sim\Beta(1,N-1)$. 
Therefore, either some node in $[0,1/N]$ carries weight $\Omega(1/\kappa)$, or the distribution function of the quadrature is $\Omega(1)$-far from that of $\Beta(1,N-1)$ at $x=1/N$. 
In either case, the Kolmogorov distance is $\Omega(1/\kappa)$. 

To establish optimality in $\epsilon$ as well, we perturb the quadrature weights multiplicatively by factors in $[1-\epsilon,1+\epsilon]$ while preserving their total mass. 
This preserves the relative error design condition while allowing an $\Omega(\epsilon)$ change in the cumulative probability. 
Together, the quadrature construction and the embedding into a unitary ensemble establish the lower bound in \thm{ks-generic-lower}, and hence the optimality of the convergence rate in \thm{ks-generic}.

\subsubsection{Smoothing with local invariance}\label{sec:overview-smoothing}

Let $a=2^m$ be the dimension of the subsystem on which the ensemble is invariant under local unitary translations.
Local invariance allows us to apply an independent Haar random unitary on this subsystem without changing the distribution of the ensemble.
Let $X_k$ be the total output probability of the corresponding block of $a$ basis states.
Conditioned on $X_k$, a single output probability $P$ in the block has the distribution of $X_kY$, where $Y\sim\Beta(1,a-1)$ is independent of $X_k$.
For a Haar random unitary, the block probability $X$ satisfies $X\sim\Beta(a,N-a)$, and the product $XY$ has precisely the Haar marginal distribution $\Beta(1,N-1)$.
This factorization separates the part controlled by the design from the independent smoothing supplied by the local Haar random unitary.
Multiplying by $Y$ therefore smooths the distribution of $X_k$, which may be discrete.
The two results below exploit this smoothing in two different ways.

\textbf{One invariant qubit.}
For $a=2$, the factor $Y$ is uniform on $[0,1]$.
Multiplication by a uniform variable acts explicitly on densities: the density of $X_kY$ at $x$ is $\Exp[X_k^{-1}1_{X_k>x}]$, whether or not $X_k$ has a density.
Integrating by parts, the difference between the densities of $X_kY$ and $XY$ is expressed through the difference $\Delta(t)=F_{X_k}(t)-F_X(t)$ of the distribution functions.
This bounds the total variation distance by the mass of both laws below a threshold $b>0$, to be chosen later, plus a weighted $L^1$ norm of $\Delta$
\begin{align}\label{eq:informal-tv-one-invariance}
d_{\mathrm{TV}}\leq F_{X_k}(b)+F_X(b)+\int_{b}^\infty|\Delta(t)|\,t^{-1}\d t.
\end{align} 

By \eq{informal-tv-one-invariance}, it suffices to show that the integral of $|\Delta(t)|/t$ over $[b,\infty)$ is small, together with the two mass terms.
To this end, we prove an approximate version of the Chebyshev--Markov--Stieltjes inequality that accommodates the design approximation error $\epsilon$: 
\begin{align}\label{eq:informal-approximate-christoffel-bound}
    |\Delta(t)|\leq\epsilon\min\{F_X(t),1-F_X(t)\}+(1+\epsilon)\lambda(t),
\end{align}
where $\lambda$ is the Christoffel function of $\Beta(2,N-2)$ for polynomials of degree at most $k$.
Its proof uses polynomial majorants and minorants of the indicator function $I_t(z)=1_{z\leq t}$, as in \sec{overview-upper-bound}, with nodes at $t$ and at the roots of the polynomial that attains the variational form of $\lambda(t)$.
The $\epsilon$-term contributes $O(\epsilon)$ to \eq{informal-tv-one-invariance}, since $\int_0^\infty\min\{F_X(t),1-F_X(t)\}\,t^{-1}\d t=O(1)$.

Then we prove a pointwise bound on $\lambda$.
The Christoffel function is invariant under rescaling, and so is the total variation distance, so we may work in the scaled variable $s=Nt$, in which $X$ has mean $2$.
We show that 
\begin{align}\label{eq:informal-pointwise-bound}
    \lambda(s/N) \leq 
    \begin{cases}
        Ck^{-1/2}(s+1/k) & 0\leq s\leq 1, \\
        Ck^{-1/2}s^{-2} & s\geq 1.
    \end{cases}
\end{align}
This bound makes the integral in \eq{informal-tv-one-invariance} small.
On $[1,\infty)$, the integrand decays like $k^{-1/2}s^{-3}$.
On $[b,1]$, the factor $s$ in the bound cancels the weight $1/s$, so that $\int_b^1\lambda(s/N)\,s^{-1}\d s\leq Ck^{-1/2}(1+\ln(1/b)/k)$.
Choosing $b=\Theta(1/k)$ in the scaled variable, the mass terms are $O(k^{-3/2})$, and we obtain the bound $O(k^{-1/2}+\epsilon)$.
A uniform bound $\lambda=O(k^{-1/2})$ would instead give $k^{-1/2}\ln(1/b)$, losing a factor $\log k$.

It remains to explain how we obtain \eq{informal-pointwise-bound}.
The proof follows that of \thm{ks-generic} closely.
For the tail $s\geq 1$, we rely on the fact that the beta measure is dominated by an exponential measure. 
The Christoffel function of the exponential measure is only bounded uniformly by $O(k^{-1/2})$.
To gain the factor $s^{-2}$, we use the test polynomial $p(x)=(x/s)\,r(x)$, where $r$ attains the Christoffel function of the exponential measure with the degree lowered by one.
The extra polynomial factor is absorbed by passing to an exponential measure with a slower decay rate, which does not change the order of the Christoffel function.
For $0\leq s\leq 1$, we compare the beta measure with an arcsine measure on $[0,R]$, where $R=\Theta(k)$, as in the proof of \thm{ks-generic}.
The test polynomial is the square of the polynomial attaining the Christoffel function of the arcsine measure at $s$.
We prove that this polynomial is localized within distance $v=O(\sqrt{s/k}+1/k)$ of $s$, and $v=O(k^{-1/2})$ for $s\leq 1$.
The density of the scaled variable $NX$ at $u$ is at most $u$ for $u$ in the window $[s-v,s+v]$, so its integral over the window is at most $2v(s+v)=O\big(k^{-1/2}(s+\sqrt{s/k}+1/k)\big)$.
By the inequality of arithmetic and geometric means, this is $O(k^{-1/2}(s+1/k))$.

\textbf{Invariance on $\log k+O(1)$ qubits.}
Larger blocks make the block probability $X\sim\Beta(a,N-a)$ concentrate, and we use this fact to prove a convergence at exponential rates.
Since $X$ has mean $a/N$, we introduce the rescaled random variable $Z=NX/a-1$, which has mean zero and variance $O(1/a)$.
We derive a recurrence relation for the central moments of $X$ and use it to show that the moments of $Z$ of order up to $k$ are exponentially small when $a$ is a sufficiently large constant multiple of $k$. 
The strong design condition transfers these moment bounds to $X_k$.

The concentration allows us to discard the exponentially rare event that $X$ or $X_k$ is below $\frac{a}{2N}$.
We then work in logarithmic coordinates, where the product $XY$ becomes a sum of independent random variables $\log X+\log Y$.
Consequently, the characteristic function factorizes.
The moment bounds on $Z$ allow us to expand the characteristic function of $\log X$ to degree $k$ with an exponentially small remainder, and the design condition transfers the expansion to $\log X_k$.
Hence the characteristic functions of $\log X_k$ and $\log X$ are close in a neighborhood of the origin, and only their high-frequency components remain uncontrolled.
The characteristic function of $\log Y$ has a bell shape, whose decaying tails suppress these high-frequency components.
These facts, together with the Berry--Esseen smoothing inequality, give the Kolmogorov bound $2^{-\Omega(k)}+O(\epsilon)$ of \thm{ks-locally-invariant}.

The same argument controls the difference of the densities in logarithmic coordinates through the Fourier inversion formula, which bounds the $L^\infty$ distance of the densities by the $L^1$ distance of the characteristic functions.
To pass from the $L^\infty$ norm to the total variation distance, which is the $L^1$ distance of the densities, we choose a window $I$ of length $2R$ centered at the logarithm of the mean.
On $I$, the $L^\infty$ bound costs a factor $2R$.
Outside $I$, we use the tail bound of the Haar law, which decays like $e^{-R}$ and which the design inherits up to the Kolmogorov distance of \thm{ks-locally-invariant}. 
Optimizing over $R$ gives the bound in \thm{tv-locally-invariant}; the optimal $R$ is logarithmic in the target accuracy, which is the source of the factor $\log(2/\epsilon)$.

\subsubsection{Bounding the mean and the fluctuations of the entropy}\label{sec:overview-entropy}

\textbf{Entropy mean and fluctuations.}
The entropy analysis in \sec{random-number} must account for all $N$ output probabilities without incurring a factor of $N$ in the error bound.
For the expectation, we subtract the distribution-independent term $\log(N/B)$ and split the remaining contribution at the threshold probability $B/N$, for a cutoff parameter $B$ to be optimized later.
Below the threshold, each summand is a function with interval superlevel sets and magnitude $O(B/N)$.
The marginal Kolmogorov bound $\eta$ therefore controls each summand up to $O(\eta B/N)$, and the $N$ summands together up to $O(\eta B)$.
Above the threshold, the contribution is bounded by a sum of second moments of the output probabilities, which the strong design matches up to relative error; this part contributes $O(1/B)$.
Choosing $B=\Theta(\eta^{-1/2})$ yields the dimension-independent bound $O(\eta^{1/2})$ of \thm{informal-entropy}.

The standard deviation involves pairs of output probabilities.
We use the same decomposition with smooth truncations, $S=\log(N/B)+G_B-R_B$, where $G_B$ collects the contributions $x\ln(1+B/(Nx))$ and $R_B$ the tail contributions $x\ln(1+Nx/B)$, up to the factor $1/\ln 2$.
The tail part has standard deviation $O(1/B)$ under both measures, by the Cauchy--Schwarz inequality and the fourth moments of the output probabilities, which the design matches up to relative error.
For the smooth part, we write each summand as a Laplace transform, $x\ln(1+B/(Nx))=B/N-N^{-1}\int_0^\infty w_B(s)e^{-sNx}\d s$, with an explicit weight $w_B\geq 0$ of total mass $B$.
The second moment of $G_B$ then becomes a double integral, over pairs $x,x'$ and over $s,t\geq 0$, of expectations of $e^{-N(sP_x+tP_{x'})}$.
Each of these is the expectation of a bounded monotone function of the linear combination $sP_x+tP_{x'}$, so it is controlled by the Kolmogorov distance between the laws of this linear combination under the design and under the Haar measure.
The supremum $\delta$ of these distances over all pairs and all nonnegative coefficients is the pair distance in \thm{informal-entropy}.
This gives $|\Var_{\mu(k,\epsilon)}[G_B]-\Var_\mu[G_B]|=O(\delta B^2)$, and optimizing $B=\Theta(\delta^{-1/4})$ against the tail part bounds the difference of the standard deviations by $O(\delta^{1/4})$.

It remains to bound $\delta$ for a strong design.
A polynomial in $sP_x+tP_{x'}$ is a polynomial in the two output probabilities, so the design condition controls the squared polynomial moments of $sP_x+tP_{x'}$ up to degree $k$ with relative error $\epsilon$.
Then we use \eq{informal-approximate-christoffel-bound}, with the Haar law of $sP_x+tP_{x'}$ in place of $\Beta(2,N-2)$, to bound the Kolmogorov distance by $\epsilon/2+(1+\epsilon)\sup_t\lambda(t)$, where $\lambda$ is now the Christoffel function of this law.
This Haar law has a density dominated by a constant multiple of the density of $\Beta(2,N-2)$ or of $\Beta(1,N-1)$, depending on the ratio $t/s$, and the Christoffel functions of both beta laws are $O(k^{-1/2})$ uniformly.
Hence $\delta=O(k^{-1/2}+\epsilon)$.
Finally, Chebyshev's inequality turns the bounds on the mean and on the standard deviation into the concentration statement of \thm{shannon-entropy-concentration}: a single circuit drawn from the design has entropy within $t+O(\eta^{1/2})$ of its Haar value, except with probability $O\big((2^{-n}+\delta^{1/2})/t^2\big)$.

\textbf{Randomness generation.}
Combining these estimates with the distributional bounds from \sec{convergence-upper-bound} and \sec{local-translation-invariant} gives quantitative guarantees on the entropy of measurements from approximate designs. The following paragraphs explain how this connects to randomness generation.

To explain why an expected entropy bound suffices for randomness generation, consider a trusted implementation that independently samples a fresh circuit $U_j$ in each round, prepares $U_j|0^n\rangle$, and measures in the computational basis to obtain $X_j$. The circuit descriptions may be public. The conditional Shannon entropy per round is
$h:=H(X_j\mid U_j)=\mathbb{E}[S(U_j)]$.
Thus our expectation bound directly lower-bounds the entropy generated per round, even when the circuit choices are known.

For $M$ independent rounds, the conditional asymptotic equipartition property gives
\begin{align}
H_{\min}^{\xi}(X_1,\ldots,X_M\mid U_1,\ldots,U_M)
\ge M(h-o(1))
\end{align}
for any fixed smoothing parameter $\xi>0$ \cite{tomamichelFullyQuantumAsymptotic2009,tomamichelLeftoverHashingQuantum2011}. A strong seeded randomness extractor then converts the measurement outcomes into nearly uniform bits that are approximately independent of both the circuit descriptions and the extractor seed. Consequently, the asymptotically achievable extraction rate approaches $h$, which our result lower-bounds by the Haar expected entropy minus $O(\sqrt{\eta})$.

This argument does not require a separate concentration inequality for the entropy of an individually sampled circuit. The required typicality concerns the accumulated conditional information across independent rounds and is already supplied by the asymptotic equipartition property. The concentration bound of \thm{shannon-entropy-concentration} addresses the different setting in which a circuit is sampled once and reused, since it guarantees that most individual circuits have high output entropy. Our randomness-generation application assumes fresh independent sampling in each round and trusted execution of the prescribed circuits.

\subsection{Related work}\label{sec:related-work}

\emph{Unitary designs and their constructions.}
The notion of a unitary design was formalized in \cite{grossEvenlyDistributedUnitaries2007,dankertExactApproximateUnitary2009}, and the Clifford group is an exact unitary $3$-design in the standard sense of \dfn{design} \cite{webbCliffordGroupForms2016,zhuMultiqubitCliffordGroups2017}.
Local random circuits were shown to form approximate $2$-designs by Harrow and Low \cite{harrowRandomQuantumCircuits2009} and approximate $k$-designs for every $k$ by Brand\~ao, Harrow, and Horodecki \cite{brandaoLocalRandomQuantum2016}; the required depth was subsequently improved in \cite{hunterjonesUnitaryDesignsStatistical2019,haferkampRandomQuantumCircuits2022a,harrowApproximateUnitaryTDesigns2023,chenIncompressibilitySpectralGaps2025}, the results were extended to arbitrary circuit architectures \cite{mittalLocalRandomQuantum2023,belkinApproximateTDesignsGeneric2024}, and, by gluing together small random unitaries, the depth was reduced to logarithmic in $n$ \cite{schusterRandomUnitariesExtremely2025,cuiUnitaryDesignsNearly2025}.
Other efficient constructions of designs and of pseudorandom unitaries are based on random Pauli rotations, random permutations, and related structured ensembles \cite{haahEfficientApproximateUnitary2025,chenEfficientUnitaryDesigns2024,metgerSimpleConstructionsLineardepth2024,maHowConstructRandom2025}.
The strong notion of design used in this paper, which also controls algorithms that query $U^\dag$, $U^*$, and $U^\top$, was introduced by Schuster, Ma, Lombardi, Brand\~ao, and Huang \cite{schusterStrongRandomUnitaries2025}; strong designs with measurable error have been constructed in optimal logarithmic depth without ancillas \cite{parelladilmeStrongUnitaryDesigns2026}.
All of these works bound the error of the twirling channel, that is, the quality of moment matching, in one of several error notions.
Our results assume the strong relative-error condition of \dfn{strong-design} and apply to any construction satisfying it as a black box; the question of what this condition implies for the law of a single output probability is independent of the construction.

\emph{Statistics of output probabilities.}
The Porter--Thomas law of the output probabilities of random circuits was conjectured and numerically supported by Boixo et al.\ \cite{boixoCharacterizingQuantumSupremacy2018}; it underlies the linear cross-entropy benchmark and enters the complexity-theoretic arguments for random circuit sampling \cite{aaronsonComplexityTheoreticFoundationsQuantum2017,aaronsonClassicalHardnessSpoofing2020,boulandComplexityVerificationQuantum2019,movassaghHardnessRandomQuantum2023}.
Rigorous results have so far concerned moments.
Anti-concentration, a bound on the collision probability of the output distribution, follows from the $2$-design property \cite{hangleiterAnticoncentrationTheoremsSchemes2018} and holds for local random circuits already in logarithmic depth \cite{dalzellRandomQuantumCircuits2022}; the behavior of the benchmark for shallow or noisy circuits has been analyzed in \cite{barakSpoofingLinearCrossEntropy2021,gaoLimitationsLinearCrossEntropy2024,wareSharpPhaseTransition2023,morvanPhaseTransitionsRandom2024}.
Recently, Hunter-Jones and Haferkamp \cite{hunterjonesIdealRandomQuantum2026} used higher moments, viewing brickwork circuits of near-quadratic depth as unitary $n$-designs, to show that noiseless random circuits pass the benchmark with high probability and that the tails of their output probabilities then resemble Porter--Thomas.
A different kind of anti-concentration, namely small-ball estimates for polynomials in the entries of a Haar random unitary, was obtained in \cite{feffermanAntiConcentrationUnitaryHaar2026}.
In contrast, we control the entire distribution function of a single output probability, for an arbitrary strong $k$-design rather than a specific circuit architecture, with rates that are tight in the worst case, and we identify the additional structure, local invariance, that upgrades the rate to exponential and yields guarantees in total variation.

\emph{Haar random unitaries and projected ensembles.}
For a Haar random unitary, the law of the squared magnitude of a single entry is a beta distribution \cite{donatimartinRandomTruncationsHaar2013}, the moments of the entries are computed by the Weingarten calculus \cite{collinsIntegrationRespectHaar2006}, and the entries and truncations of Haar random matrices have been studied in detail in random matrix theory \cite{zyczkowskiTruncationsRandomUnitary2000,jiangHowManyEntries2006}.
In contrast, the design condition alone does not determine the law of a single output probability: \sec{convergence-lower-bound} constructs strong designs for which the squared magnitude of an entry is supported on only $k+1$ points and attains the worst-case Kolmogorov separation.
A complementary line of work studies the ensemble of states obtained by measuring part of a chaotic many-body state and shows that it forms an approximate state design \cite{cotlerEmergentQuantumState2023,choiPreparingRandomStates2023,ippolitiSolvableModelDeep2022,markMaximumEntropyPrinciple2024}; there, too, closeness to Haar randomness is measured through moments, and distribution-level guarantees of the kind established here are a natural question in that setting.

\emph{Randomness generation.}
Certified randomness protocols have been proposed based on random circuit sampling \cite{aaronsonCertifiedRandomnessQuantum2023} and, under different assumptions, on Fourier sampling \cite{bassirianCertifiedRandomnessFourier2026}; the former has been demonstrated experimentally \cite{liuCertifiedRandomnessUsing2025a,liuCertifiedRandomnessAmplification2025}.
In the protocol by Aaronson and Hung \cite{aaronsonCertifiedRandomnessQuantum2023}, the behavior of an honest device is analyzed in the Haar random (Porter--Thomas) model, and the accumulation of entropy over many rounds relies on entropy accumulation theorems \cite{metgerGeneralisedEntropyAccumulation2024}.
\sec{random-number} supplies the ingredient needed to replace the Haar random unitary by an efficiently implementable strong design in the trusted-device setting, namely that the expected output entropy of such a design matches its Haar value up to a dimension-independent error.

\emph{Techniques.}
Bounding a distribution function from finitely many moments is the subject of the classical moment problem, and the Chebyshev--Markov--Stieltjes inequalities together with Gauss quadrature \cite{szegoOrthogonalPolynomials1975,akhiezerClassicalMomentProblem1965,lindsayMomentsDetermineTail2000,gilboaDifferentialVersionChebyshev2015} are the tools behind \sec{convergence-upper-bound} and \sec{convergence-lower-bound}. 
To our knowledge, the embedding of the extremal measures of the moment problem into unitary designs is new.
In \sec{local-translation-invariant}, the exponential rates under local invariance are obtained by combining the moment information with the Berry--Esseen smoothing inequality \cite{Feller1971} in logarithmic coordinates.
This combination is new as well.

\emph{Concurrent work.}
In concurrent and independent work, Sen and Hunter-Jones \cite{senLocalRandomQuantum2026} prove that the output probability of a fixed bitstring in a one-dimensional brickwork circuit converges to Porter--Thomas in total variation distance $O(\sqrt{\log k/k})$ once the circuit forms an additive $2^{-nk}$-approximate $k$-design, i.e., $O(n^{-c})$ at depth $\Theta(n^{2c+1}\log n)$ for fixed $c$.
Their result is specific to brickwork circuits, whereas \thm{informal-one-local-invariance} applies to every strong design whose last layer contains a Haar random gate on a single qubit, with the rate $O(k^{-1/2}+\epsilon)$, slightly improving the bound in \cite{senLocalRandomQuantum2026} by a factor $O(\log^{1/2}k)$.
By \thm{ks-generic-lower}, no better rate holds without local structure.
They note that their techniques do not reach inverse-exponential rates. 
\thm{tv-locally-invariant} achieves the rate $2^{-\Omega(k)}$ once the last layer contains a Haar random gate on $\log k+O(1)$ qubits, while the question for plain brickwork circuits remains open.
The Kolmogorov bounds for arbitrary strong designs, the matching lower bounds, and the entropy results are specific to our work.

\subsection{Discussion}

We established quantitative guarantees for the distributional accuracy of strong approximate unitary designs. Our Kolmogorov bounds characterize, up to constant factors and in the worst case over designs, how closely individual output probabilities follow their Haar counterparts. Additional local invariance substantially improves the dependence on the design order and also yields guarantees in total variation distance. As an application, we showed that marginal distributional closeness controls the expected Shannon entropy, supporting randomness generation with fresh, independently sampled circuits in a trusted implementation.

A natural next question is whether these guarantees extend to the joint distribution of multiple output probabilities. We conjecture that, under the same strong-design assumption, the Kolmogorov bound extends to joint distribution functions with at most a linear loss in the number of entries. Here, the comparison is with the corresponding joint Haar distribution, and closeness is measured uniformly over coordinatewise thresholds. Establishing such a result would quantify how accurately designs reproduce correlations among output probabilities. It would also be interesting to determine whether local invariance yields comparable improvements for joint distributions.

Another direction is to identify further applications that benefit from distributional guarantees beyond moment matching. Can these guarantees improve the analysis of tasks whose performance depends on non-polynomial functions of output probabilities or on joint probability events? In particular, it remains to determine when the stronger convergence enabled by local invariance translates into lower design orders or more efficient circuit constructions for such tasks.

\section{Preliminaries}\label{sec:preliminaries}

In this paper, all logarithms are base two unless indicated otherwise. 
Throughout the paper, we use the Kolmogorov distance $d_{\mathrm K}$ and the total variation distance $d_{\mathrm{TV}}$ between probability distributions defined in \sec{introduction}, and for a distribution $\rho$ on $\mathbb R$ we write $F_\rho$ for its distribution function.
If $\mu$ and $\nu$ have densities $f_\mu$ and $f_\nu$ with respect to a common $\sigma$-finite measure $\lambda$ (for instance the Lebesgue measure, the counting measure, or their sum when the distributions have both an absolutely continuous part and atoms), then the total variation distance is half the $L^1$ distance between the densities,
\begin{align}\label{eq:tv-l1}
    d_{\mathrm{TV}}(\mu,\nu)=\frac12\int|f_\mu-f_\nu|\d\lambda,
\end{align}
which is attained by the set $A=\{f_\mu>f_\nu\}$ in the definition of $d_{\mathrm{TV}}$ (Scheff\'e's identity); since $\lambda=\mu+\nu$ dominates both distributions, such densities always exist.
For random variables $X$ and $Y$, we write $F_X$ for the distribution function of $X$ and, by a slight abuse of notation, $d_{\mathrm K}(X,Y)$ and $d_{\mathrm{TV}}(X,Y)$ for the corresponding distances between their laws, so that $d_{\mathrm K}(X,Y)=\sup_{x\in\mathbb R}|F_X(x)-F_Y(x)|$.

\subsection{Approximate unitary designs}\label{sec:design}

We will analyze the convergence of a unitary design to Porter--Thomas. 
A standard notion of an approximate unitary design was defined by Brand\~ao, Harrow, and Horodecki as follows.
\begin{definition}[{$\epsilon$-approximate unitary $k$-design~\cite[Definition~2]{brandaoLocalRandomQuantum2016}}]\label{dfn:design}
    Let $\nu$ be a distribution on $\mathbb U(N)$ and 
    \begin{align}
        \Delta_\nu(X) := \Exp_{U\sim\nu}\left[U^{\otimes k} X (U^\dag)^{\otimes k} \right].
    \end{align}
    Then $\nu$ is an $\epsilon$-approximate unitary $k$-design if 
    \begin{align}
        (1-\epsilon)\Delta_{\Haar}
        \leq 
        \Delta_\nu \leq (1+\epsilon)\Delta_{\Haar},
    \end{align}
    where $\Haar$ is the Haar measure on $\mathbb U(N)$, and $\mathcal A\leq\mathcal B$ denotes that $\mathcal B-\mathcal A$ is completely positive.
\end{definition}

In our analysis to show statistical convergence, we require the approximation of expected polynomials of degree at most $2k$, and therefore we require a stronger notion. 
\begin{definition}[{Strong $\epsilon$-approximate unitary $(p,q)$-design}]\label{dfn:strong-design}
    Let $\nu$ be a distribution on $\mathbb U(N)$ and 
    \begin{align}
        \Delta_\nu^{(p,q)}(X) := \Exp_{U\sim\nu}\left[
            U^{\otimes p}\otimes (U^*)^{\otimes q} 
            X
            (U^\dag)^{\otimes p}\otimes (U^\top)^{\otimes q}
        \right],
    \end{align}
    where $U^*,U^\top$ are the elementwise complex conjugate of $U$ and the transpose of $U$ respectively. 
    Then $\nu$ is a strong $\epsilon$-approximate unitary $(p,q)$-design if 
    \begin{align}
        (1-\epsilon)\Delta_{\Haar}^{(p,q)}
        \leq 
        \Delta_\nu^{(p,q)} \leq (1+\epsilon)\Delta_{\Haar}^{(p,q)},
    \end{align}
    where $\Haar$ is the Haar measure on $\mathbb U(N)$, and $\mathcal A\leq\mathcal B$ denotes that $\mathcal B-\mathcal A$ is completely positive.
\end{definition}

The unitary design definition in \dfn{design} is also called a unitary design with relative error. 
Nevertheless, we can also define a unitary design with an additive error to mean a distribution $\nu$ satisfying $\|\Delta_{\Haar}^{(p,q)}-\Delta_{\nu}^{(p,q)}\|_{\diamond}\leq\epsilon$.
In this paper, we prove statistical convergence for a strong $\epsilon$-approximate $(k,k)$-design with relative error.

\begin{remark}\label{rmk:design}
Throughout this paper, when we refer to an $\epsilon$-approximate unitary $k$-design in our own statements, we consistently mean a strong $\epsilon$-approximate $(k,k)$-design with relative error for integer $k$ (see \dfn{strong-design} for the definition).
The only exception is our discussion of results from the literature, in particular in \sec{related-work}, where a design is understood in the sense of the cited work: often the standard notion of \dfn{design}, or a notion with additive (also called measurable) error, neither of which implies the relative-error condition assumed in our theorems.
We will also use the notation $\mu(k,\epsilon)$ to refer to the distribution of a strong $\epsilon$-approximate $(k,k)$-design, and, when the two are compared, $\mu$ for the Haar measure on unitaries of the same dimension.
Note that a strong $\epsilon$-approximate $(k,k)$-design is also a strong $\epsilon$-approximate $(k',k')$-design for every integer $0\leq k'\leq k$: for density matrices $\rho,\sigma$, the map $\Delta_\nu^{(k',k')}(X)$ is obtained from $\Delta_\nu^{(k,k)}(X\otimes\rho^{\otimes(k-k')}\otimes\sigma^{\otimes(k-k')})$ (with the tensor factors arranged so that $\rho$ occupies the last $k-k'$ copies of $U$ and $\sigma$ the last $k-k'$ copies of $U^*$) by tracing out those copies, and pre- and post-composition with completely positive maps preserves the ordering in \dfn{strong-design}.
\end{remark}

\subsection{Berry--Esseen smoothing lemma}\label{sec:berry-esseen}

We will be using the well-known Berry--Esseen smoothing lemma for our analysis in \sec{local-translation-invariant}.
In particular, this allows us to bound the closeness between two random variables in Kolmogorov distance by bounding the pointwise closeness between their characteristic functions. 
We state the lemma here. 

\begin{lemma}[{Berry--Esseen smoothing lemma \cite[Chapter~XVI]{Feller1971}}]\label{lem:berry-esseen}
Let $X$ and $Y$ be two real-valued random variables with finite first moments and with distribution functions $F_X$ and $F_Y$, respectively. 
If the probability density function of $Y$ is bounded by $M$, then for $T>0$,
\begin{align}
\sup_{x\in\R}|F_X(x)-F_Y(x)|
\leq \frac{1}{\pi}\int_{-T}^{T} \left|\frac{\varphi_{X}(t)-\varphi_Y(t)}{t}\right| \d t + \frac{24M}{\pi T}
\end{align}
where $\varphi_X$ and $\varphi_Y$ are the characteristic functions of $X$ and $Y$, respectively.
\end{lemma}

The first part grows when $T$ becomes larger, so to get the best achievable bound using \lem{berry-esseen}, we can optimize over the choices of $T$. 

When the probability density functions of $X$ and $Y$ as in \lem{berry-esseen} exist and the characteristic functions are integrable, we can use the fact that the probability density function is the inverse Fourier transform of the characteristic function (we drop the factor $1/(2\pi)$ of the inversion formula; keeping it would only sharpen the bound) to show that the maximum pointwise difference is at most 
\begin{align}
    \max_{x\in\mathbb R}|f_X(x)-f_Y(x)| \leq \int_{-\infty}^\infty \d t |\varphi_X(t)-\varphi_Y(t)|.
\end{align}
We will be using this closeness notion to reason about the closeness in total variation distance in \sec{locally-invariant-tv}. 

\subsection{Beta and Dirichlet random variables}\label{sec:preliminaries-beta}

For parameters $\alpha=(\alpha_1,\ldots,\alpha_d)$ with $\alpha_i>0$, the Dirichlet distribution $\Dir(\alpha)$ is the distribution on the simplex $\{x\in[0,1]^d: x_1+\cdots+x_d=1\}$ with probability density function
\begin{align}
    f_{\alpha}(x_1,\ldots,x_d) = \frac{1}{B(\alpha)} x_1^{\alpha_1-1}\cdots x_d^{\alpha_d-1}, \qquad B(\alpha)=\frac{\Gamma(\alpha_1)\cdots\Gamma(\alpha_d)}{\Gamma(\alpha_1+\cdots+\alpha_d)},
\end{align}
with respect to the Lebesgue measure on the simplex (in the coordinates $x_1,\ldots,x_{d-1}$, with $x_d=1-x_1-\cdots-x_{d-1}$); $B(\alpha)$ is the multivariate beta function.
We use the following two standard facts; see, e.g., \cite[Chapter~49]{kotzContinuousMultivariateDistributions2000}.
\begin{enumerate}
    \item \emph{Aggregation.} If $(X_1,\ldots,X_d)\sim\Dir(\alpha)$ and $I_1,\ldots,I_r$ is a partition of $[d]$, then
    \begin{align}
        \Big(\sum_{i\in I_1}X_i,\ldots,\sum_{i\in I_r}X_i\Big)\sim\Dir\Big(\sum_{i\in I_1}\alpha_i,\ldots,\sum_{i\in I_r}\alpha_i\Big).
    \end{align}
    \item \emph{Moments.} For nonnegative integers $m=(m_1,\ldots,m_d)$, since $x_1^{m_1}\cdots x_d^{m_d}f_\alpha\propto f_{\alpha+m}$, we have $\Exp\big[X_1^{m_1}\cdots X_d^{m_d}\big]=B(\alpha+m)/B(\alpha)$ for $(X_1,\ldots,X_d)\sim\Dir(\alpha)$.
\end{enumerate}

The beta distribution $\Beta(\alpha,\beta)$ is the law of the first coordinate of $\Dir(\alpha,\beta)$, i.e., the distribution on $[0,1]$ with density
\begin{align}
    f_{\alpha,\beta}(x) = \frac{1}{B(\alpha,\beta)} x^{\alpha-1} (1-x)^{\beta-1},
\end{align}
where $B(\alpha,\beta)=\frac{\Gamma(\alpha)\Gamma(\beta)}{\Gamma(\alpha+\beta)}$ is the beta function, and its moments are
\begin{align}
    \Exp_{X\sim\Beta(\alpha,\beta)}[X^m] = \frac{B(\alpha+m,\beta)}{B(\alpha,\beta)}.
\end{align}
By aggregation, if $(X_1,\ldots,X_d)\sim\Dir(\alpha)$ and $I\subseteq[d]$ is nonempty and proper, then $\sum_{i\in I}X_i\sim\Beta\big(\sum_{i\in I}\alpha_i,\sum_{i\notin I}\alpha_i\big)$; in particular, each coordinate $X_i$ has the law $\Beta\big(\alpha_i,\sum_{j\neq i}\alpha_j\big)$.
We also use the following product rule: if $X\sim\Beta(\alpha+\beta,\gamma)$ and $Y\sim\Beta(\alpha,\beta)$ are independent, then $XY\sim\Beta(\alpha,\beta+\gamma)$.
Indeed, by independence and the moment formula, for every integer $m\geq 0$,
\begin{align}\label{eq:beta-product}
    \Exp[(XY)^m]=\frac{B(\alpha+\beta+m,\gamma)}{B(\alpha+\beta,\gamma)}\cdot\frac{B(\alpha+m,\beta)}{B(\alpha,\beta)}
    =\frac{\Gamma(\alpha+m)\,\Gamma(\alpha+\beta+\gamma)}{\Gamma(\alpha)\,\Gamma(\alpha+\beta+\gamma+m)}
    =\frac{B(\alpha+m,\beta+\gamma)}{B(\alpha,\beta+\gamma)},
\end{align}
and a distribution on the bounded interval $[0,1]$ is determined by its moments.

These distributions enter through the following observation.
A column of a $d$-dimensional Haar random unitary $U$ is a uniformly random unit vector in $\mathbb C^d$, which can be realized as $g/\|g\|$ for a vector $g$ of independent standard complex Gaussian entries.
Since $|g_1|^2,\ldots,|g_d|^2$ are independent exponential random variables, the gamma representation of the Dirichlet distribution \cite[Chapter~49]{kotzContinuousMultivariateDistributions2000} shows that the vector $(|U_{1y}|^2,\ldots,|U_{dy}|^2)$ of squared magnitudes of any column $y$ has the law $\Dir(1^d)=\Dir(1,\ldots,1)$, the uniform distribution on the simplex.
Consequently, by aggregation, the squared magnitude of a single entry has the Porter--Thomas law $\Beta(1,d-1)$, the total probability of any set of $a$ entries of a column, $1\leq a\leq d-1$, has the law $\Beta(a,d-a)$, and, for $d\geq 3$ and two distinct entries $x\neq x'$ of a column, the triple $(|U_{xy}|^2,|U_{x'y}|^2,1-|U_{xy}|^2-|U_{x'y}|^2)$ has the law $\Dir(1,1,d-2)$, so that $(|U_{xy}|^2,|U_{x'y}|^2)$ has the density $(d-1)(d-2)(1-u-v)^{d-3}$ on $\{u,v\geq 0,\ u+v\leq 1\}$.
We use these facts in \sec{local-translation-invariant} and \sec{random-number}.

\subsection{Orthogonal polynomials and Jacobi polynomials}\label{sec:jacobi}

Let $\nu$ be a probability measure on $\mathbb R$ with finite moments whose support is an infinite set.
Applying the Gram--Schmidt procedure to $1,x,x^2,\ldots$ in $L^2(\nu)$ yields the orthonormal polynomials $q_0,q_1,\ldots$ of $\nu$, which are unique up to sign and satisfy $\deg(q_n)=n$ and $\int q_nq_m\d\nu=\delta_{n,m}$, where $\delta_{n,m}$ is the Kronecker delta.
The zeros of $q_n$ are real and simple and lie in the interior of the convex hull of the support of $\nu$, and the zeros of $q_n$ and $q_{n+1}$ interlace \cite[Theorems~3.3.1 and 3.3.2]{szegoOrthogonalPolynomials1975}.
We use these facts to construct Gauss quadrature rules in \sec{jacobi-gauss}.

\paragraph{Jacobi polynomials.}
For parameters $\alpha,\beta>-1$, the Jacobi polynomials $P_n^{(\alpha,\beta)}$, $n\geq 0$, are the orthogonal polynomials on $[-1,1]$ with respect to the weight
\begin{align}
    w_{\alpha,\beta}(x) = (1-x)^\alpha (1+x)^\beta,
\end{align}
normalized as in \cite[Chapter~IV]{szegoOrthogonalPolynomials1975}; the condition $\alpha,\beta>-1$ ensures that the weight is integrable.
The polynomial $P_n^{(\alpha,\beta)}$ has degree $n$, and
\begin{align}
    \int_{-1}^1 P_{n}^{(\alpha,\beta)}(x) P_{m}^{(\alpha,\beta)}(x) w_{\alpha,\beta}(x)\d x = h_n\delta_{n,m},
\end{align}
where $h_n>0$ is a normalization constant.
The Jacobi polynomials satisfy the symmetry relation \cite[Section~4.1]{szegoOrthogonalPolynomials1975}
\begin{align}\label{eq:jacobi-symmetry}
    P_n^{(\alpha,\beta)}(-x) = (-1)^n P_n^{(\beta,\alpha)}(x).
\end{align}

The orthonormal polynomials of a beta distribution are shifted Jacobi polynomials.
Indeed, let $\mu_{a,D}$ be the law of $\Beta(a,D)$ for $a,D>0$ (\sec{preliminaries-beta}); the change of variables $t=2x-1$ gives, for all polynomials $p$ and $q$,
\begin{align}
    \int_0^1 p(x)q(x)\frac{x^{a-1}(1-x)^{D-1}}{B(a,D)}\d x = \frac{1}{2^{a+D-1}B(a,D)}\int_{-1}^1 p\Big(\frac{1+t}{2}\Big)q\Big(\frac{1+t}{2}\Big)(1-t)^{D-1}(1+t)^{a-1}\d t,
\end{align}
so that the orthonormal polynomials $\{q_n\}$ of $\mu_{a,D}$ are, up to normalization,
\begin{align}\label{eq:beta-jacobi}
    q_n(x)\propto P_n^{(D-1,a-1)}(2x-1).
\end{align}
The following two facts about the zeros of Jacobi polynomials are used in \sec{convergence-lower-bound}.

\begin{proposition}[Zeros of Jacobi polynomials]\label{prop:jacobi-zeros}
    For $\ell\geq 1$, write the zeros of $P_\ell^{(\alpha,\beta)}$ as $z_j^{(\alpha,\beta)}=\cos\theta_j^{(\alpha,\beta)}$ with $0<\theta_1^{(\alpha,\beta)}<\cdots<\theta_\ell^{(\alpha,\beta)}<\pi$, so that $z_1^{(\alpha,\beta)}>\cdots>z_\ell^{(\alpha,\beta)}$.
    \begin{enumerate}
        \item (Bruns' inequality \cite[Theorem~6.21.2]{szegoOrthogonalPolynomials1975}) For the Legendre polynomials, i.e., $\alpha=\beta=0$, we have $\theta_j^{(0,0)}\geq\frac{(2j-1)\pi}{2\ell+1}$ for every $j\in[\ell]$.
        \item (Monotonicity in the parameters \cite{dimitrovMonotonicityZerosJacobi2007}) Let $A_\ell(\alpha,\beta):=2\ell^2+2\ell(\alpha+\beta+1)+(\alpha+1)(\beta+1)$. For fixed $\ell$, $\alpha$, and $j$, the quantity $A_\ell(\alpha,\beta)\big(1-z_j^{(\alpha,\beta)}\big)$ is an increasing function of $\beta$.
    \end{enumerate}
\end{proposition}

\paragraph{Chebyshev polynomials and the arcsine measure.}
For $\alpha=\beta=-1/2$, the Jacobi polynomials are, up to normalization, the Chebyshev polynomials of the first kind $T_n$, which are characterized by $T_n(\cos\theta)=\cos(n\theta)$.
For $z\geq 1$,
\begin{align}\label{eq:chebyshev-growth}
    T_n(z)=\frac{(z+\sqrt{z^2-1})^n+(z-\sqrt{z^2-1})^n}{2}, \qquad\text{and hence}\qquad |T_n(z)|\leq(2z)^n.
\end{align}
For $R>0$, the arcsine probability measure $\rho_R$ on $[0,R]$ is
\begin{equation}
 \d\rho_R(t)=\frac{\d t}{\pi\sqrt{t(R-t)}};
 \label{eq:arcsine}
\end{equation}
it is the image of the normalized Chebyshev weight $\frac{\d x}{\pi\sqrt{1-x^2}}$ on $[-1,1]$ under $x\mapsto t=R(1+x)/2$, and its orthonormal polynomials are
\begin{align}\label{eq:arcsine-basis}
    q_0(t)=1,\qquad q_j(t)=\sqrt2\,T_j(2t/R-1)\quad(j\geq 1).
\end{align}
In the angular variable $t=R\sin^2(\phi/2)$ with $\phi\in[0,\pi]$, we have $2t/R-1=-\cos\phi$, $\d\rho_R(t)=\d\phi/\pi$, and $q_j(t)=\sqrt2(-1)^j\cos(j\phi)$.
The arcsine measure serves as a comparison measure for Christoffel functions in \sec{jacobi-gauss} and \sec{one-local-invariance}.

\subsection{Inequalities for the factorials}

\begin{proposition}[{Stirling's formula \cite{Robbins1955}}]\label{prop:stirling}
    For integer $n\geq 1$, 
    \begin{align}
    e^{\frac{1}{12n+1}}\sqrt{2\pi n}(n/e)^{n}\leq n! \leq e^{\frac{1}{12n}}\sqrt{2\pi n}(n/e)^{n}.
    \end{align}
\end{proposition}

We will need the bound for double factorials, which can be derived from Stirling's formula.
\begin{proposition}\label{prop:double-factorial}
    For integer $n\geq 1$, 
    \begin{align}
        (2n-1)!! \leq 2 \left(\frac{2n}{e}\right)^{n}
    \end{align}
\end{proposition}
\begin{proof}
    For integer $n\geq 1$, by \prop{stirling} we have 
    \begin{align}
        (2n-1)!! 
        &= \frac{(2n)!}{2^n n!}
        \leq \frac{e^{\frac{1}{24n}}\sqrt{4\pi n}(2n/e)^{2n}}{2^n \sqrt{2\pi n}(n/e)^{n}} 
        = \sqrt{2} \cdot e^{\frac{1}{24n}} \left(\frac{2n}{e}\right)^{n}.
    \end{align}
    The factor $e^{1/24}<1.05<\sqrt 2$ for $n\geq 1$.
\end{proof}

\prop{double-factorial-ratio} below shows that $\frac{k!!}{(k-1)!!}=\Theta(k^{1/2})$ for integer $k\geq 1$.
The following propositions adopt the convention that $0!!=(-1)!!=1$.
\begin{proposition}[{Wallis cosine formula \cite[Section~3.621]{GradshteynRyzhik2015}}]\label{prop:wallis}
    For integer $k\geq 0$, 
    \begin{align}
        \int_0^{\pi/2} \cos^k x \d x = 
        \begin{cases}
            \frac{(k-1)!!}{k!!}\cdot \frac{\pi}{2} & \text{if $k$ is even}\\
            \frac{(k-1)!!}{k!!} & \text{if $k$ is odd}
        \end{cases}
    \end{align}
    where the convention $0!!=(-1)!!=1$ is adopted.
\end{proposition}

\begin{proposition}\label{prop:double-factorial-ratio}
    For integer $k\geq 1$, $\sqrt{\frac{2k}{\pi}} \leq \frac{k!!}{(k-1)!!}\leq \sqrt{\frac{\pi(k+1)}{2}}$, with the convention that $0!!=1$.
\end{proposition}
\begin{proof}
    By \prop{wallis}, since $\cos x\leq 1$, for even $k$, we have 
    \begin{align}
        \frac{k!!}{(k+1)!!} \leq \frac{\pi}{2} \cdot \frac{(k-1)!!}{k!!} \leq \frac{(k-2)!!}{(k-1)!!},
    \end{align}
    with the convention that $0!!=(-1)!!=1$.
    Since $k!!=k\cdot (k-2)!!$, we have 
    \begin{align}\label{eq:wallis-even}
        \sqrt{\frac{\pi k}{2}} \leq \frac{k!!}{(k-1)!!} \leq \sqrt{\frac{\pi(k+1)}{2}}.
    \end{align}
    For odd $k$, we can simply substitute the factor $\pi/2$ in \eq{wallis-even} by its inverse, and the bound becomes
    \begin{align}
        \sqrt{\frac{2k}{\pi}} \leq \frac{k!!}{(k-1)!!} \leq \sqrt{\frac{2(k+1)}{\pi}}.
    \end{align}
    For our purpose, we combine the two bounds by choosing the weaker one on each side to establish a unified bound for all $k\geq 1$.
\end{proof}

\section{Closeness of polynomial expectations}\label{sec:block-encoding}

In this section, we show that for a polynomial $p$ of degree at most $k$, the expectation of $|p(X)|^2$, where $X$ is the squared magnitude of an entry of $U$ (or, more generally, the total probability of a block of entries in a column) and $U$ is sampled from an $\epsilon$-approximate $k$-design, approximates that from the Haar measure to relative error $\epsilon$ (see \rmk{design} for the definition); we also derive a first-moment version of this statement.

\begin{lemma}\label{lem:strong-design-polynomial}
    Let $\nu$ be a probability distribution on $\mathrm U(D)$, and $p,q$ be nonnegative integers.
    Then $\nu$ is a strong $\epsilon$-approximate $(p,q)$-design with relative error if and only if for every polynomial $f$ whose every monomial contains $p$ entries of $U$ and $q$ entries of $U^*$, it holds that 
    \begin{align}
        (1-\epsilon) \Exp_{U\sim \Haar(D)}\left[|f(U)|^2\right]
        \leq 
        \Exp_{U\sim \nu}\left[|f(U)|^2\right]
        \leq 
        (1+\epsilon) \Exp_{U\sim \Haar(D)}\left[|f(U)|^2\right].
    \end{align}
\end{lemma}
\begin{proof}
    Let $R_{p,q}(U) := U^{\otimes p}\otimes (U^*)^{\otimes q}$ and 
    \begin{align}
        \Delta_\mu^{(p,q)}(X) := \Exp_{U\sim\mu}\left[
            R_{p,q}(U)
            X
            R_{p,q}(U)^\dag
        \right].
    \end{align}
    Let $\ket{v(U)}$ be the vectorization of $R_{p,q}(U)$.
    Then for every polynomial $f$ whose every monomial contains $p$ entries of $U$ and $q$ entries of $U^*$, there is $\ket{\psi}\in\mathbb C^{D^{p+q}}\otimes\mathbb C^{D^{p+q}}$ such that
    \begin{align}
        f(U) = \bra{\psi} v(U)\rangle.
    \end{align}
    Define the Choi matrix of $\Delta_\mu^{(p,q)}$ as
    \begin{align}
        M_\mu := \Exp_{U\sim\mu}[\proj{v(U)}].
    \end{align}
    By the Choi criterion for complete positivity, $\Delta_\mu^{(p,q)}$ is completely positive if and only if its Choi matrix $M_\mu = \Exp_{U\sim\mu}[\proj{v(U)}]$ is positive semidefinite. 
    Then we have
    \begin{align}
        \Exp_{U\sim\mu}[|f(U)|^2] = \bra{\psi} M_\mu \ket{\psi}.
    \end{align}
    Therefore $\nu$ is a strong design if and only if \begin{align}
        (1-\epsilon) M_{\Haar(D)} \leq M_\nu \leq (1+\epsilon) M_{\Haar(D)},
    \end{align}
    i.e., if and only if $(1-\epsilon)\bra{\psi}M_{\Haar(D)}\ket{\psi}\leq\bra{\psi}M_\nu\ket{\psi}\leq(1+\epsilon)\bra{\psi}M_{\Haar(D)}\ket{\psi}$ for every vector $\ket{\psi}$.
    Since every $\ket{\psi}$ arises from a polynomial $f(U)=\langle\psi|v(U)\rangle$ of the stated form (each entry of $\ket{v(U)}$ is a product of $p$ entries of $U$ and $q$ entries of $U^*$), these scalar inequalities are exactly the inequalities in the lemma.
\end{proof}

\begin{theorem}\label{thm:polynomial-relative-error}
    For an integer $N\geq 2$, let $\mu(k,\epsilon)$ be the distribution of an $\epsilon$-approximate $k$-design on $\mathrm U(N)$ (see \rmk{design} for the clarification on the parameters).
    Fix a column index $y$ and a set $S$ of row indices, and let
    \begin{align}
        X = X(U) = \sum_{r\in S} |\bra{r}U\ket{y}|^2
    \end{align}
    be the total probability of the block $S$; for $S=\{x\}$ this is the squared magnitude $|\bra{x}U\ket{y}|^2$ of a single entry.
    Let $p(x)$ be a polynomial of degree at most $k$ with complex coefficients.
    Then
    \begin{align}\label{eq:relative-error-p}
        (1-\epsilon)\Exp_{U\sim\Haar(N)}[|p(X)|^2]
        \leq
        \Exp_{U\sim\mu(k,\epsilon)}[|p(X)|^2]
        \leq
        (1+\epsilon)\Exp_{U\sim\Haar(N)}[|p(X)|^2].
    \end{align}
    Furthermore, for every two polynomials $p,q$ of degree at most $k$ with complex coefficients,
    \begin{align}\label{eq:relative-error-pq}
        \left|\Exp_{U\sim\mu(k,\epsilon)}\big[\overline{p(X)}q(X)\big] - \Exp_{U\sim\Haar(N)}\big[\overline{p(X)}q(X)\big]\right|
        \leq \epsilon \sqrt{\Exp_{U\sim\Haar(N)}[|p(X)|^2]\cdot\Exp_{U\sim\Haar(N)}[|q(X)|^2]}.
    \end{align}
    In particular, taking $p\equiv 1$, $\big|\Exp_{U\sim\mu(k,\epsilon)}[q(X)]-\Exp_{U\sim\Haar(N)}[q(X)]\big|\leq\epsilon\sqrt{\Exp_{U\sim\Haar(N)}[|q(X)|^2]}$.
\end{theorem}
\begin{proof}
    For a polynomial $p(x)=\sum_{j=0}^k a_j x^j$, define
    \begin{align}
        f_p(U) = \sum_{j=0}^k a_j\Big(\sum_{r\in S}|U_{ry}|^{2}\Big)^{j} \left(\sum_{r=0}^{N-1} |U_{ry}|^2\right)^{k-j}.
    \end{align}
    Then $f_p(U)$ is a polynomial whose every monomial contains $k$ entries of $U$ and $k$ entries of $U^*$.
    Since $U$ is a unitary matrix, we have
    \begin{align}
        \sum_{r=0}^{N-1} |U_{ry}|^2 = 1,
    \end{align}
    which implies that $f_p(U) = p(X(U))$.
    By \lem{strong-design-polynomial}, we conclude that \eq{relative-error-p} holds.

    For \eq{relative-error-pq}, we use the notation of the proof of \lem{strong-design-polynomial}: there are vectors $\ket{\psi_p},\ket{\psi_q}$ such that $f_p(U)=\langle\psi_p|v(U)\rangle$ and $f_q(U)=\langle\psi_q|v(U)\rangle$, so that $\Exp_{U\sim\mu}[\overline{f_p(U)}f_q(U)]=\bra{\psi_q}M_\mu\ket{\psi_p}$ for every distribution $\mu$, where $M_\mu=\Exp_{U\sim\mu}[\proj{v(U)}]$.
    Let $M=M_{\Haar(N)}$ and $\Delta=M_{\mu(k,\epsilon)}-M$.
    By the proof of \lem{strong-design-polynomial}, the design condition is equivalent to $-\epsilon M\leq\Delta\leq\epsilon M$, i.e., the matrices $A_\pm:=\epsilon M\pm\Delta$ are positive semidefinite.
    For a positive semidefinite matrix $A$ and vectors $u,w$, the Cauchy--Schwarz inequality gives $|\bra{u}A\ket{w}|\leq\sqrt{\bra{u}A\ket{u}\bra{w}A\ket{w}}$.
    Since $\Delta=(A_+-A_-)/2$, applying this inequality to $A_+$ and $A_-$ and then the Cauchy--Schwarz inequality in $\mathbb R^2$, we obtain
    \begin{align}
        |\bra{u}\Delta\ket{w}|
        &\leq \frac{1}{2}\sum_{\pm}\sqrt{\bra{u}A_\pm\ket{u}\bra{w}A_\pm\ket{w}} \\
        &\leq \frac{1}{2}\sqrt{\Big(\sum_\pm\bra{u}A_\pm\ket{u}\Big)\Big(\sum_\pm\bra{w}A_\pm\ket{w}\Big)}
        = \epsilon\sqrt{\bra{u}M\ket{u}\bra{w}M\ket{w}}.\nonumber
    \end{align}
    Taking $u=\psi_q$ and $w=\psi_p$ yields \eq{relative-error-pq}, since $\bra{\psi_q}M\ket{\psi_q}=\Exp_{\Haar(N)}[|q(X)|^2]$ and $\bra{\psi_p}M\ket{\psi_p}=\Exp_{\Haar(N)}[|p(X)|^2]$.
\end{proof}

\begin{remark}\label{rmk:several-blocks}
    The proof of \thm{polynomial-relative-error} applies without change to polynomials in several block probabilities of the same column: if $X_1,\ldots,X_r$ are the probabilities of blocks $S_1,\ldots,S_r$ (not necessarily disjoint) in the column $y$ and $p(x_1,\ldots,x_r)$ is a polynomial of total degree at most $k$ with complex coefficients, then \eq{relative-error-p} and \eq{relative-error-pq} hold with $p(X)$ replaced by $p(X_1,\ldots,X_r)$.
    Indeed, multiplying each monomial of total degree $j$ by $\big(\sum_{z=0}^{N-1}|U_{zy}|^2\big)^{k-j}$ produces a polynomial $f_p(U)$ in which every monomial contains $k$ entries of $U$ and $k$ entries of $U^*$ and which coincides with $p(X_1,\ldots,X_r)$ on unitary matrices, so \lem{strong-design-polynomial} applies.
    We use this extension in \sec{variance-shannon} for linear combinations $sP_x+tP_{x'}$ of two output probabilities.
\end{remark}

We will be using \thm{polynomial-relative-error} together with polynomial approximations of step functions and exponential functions to establish our statistical convergence results in \sec{convergence-upper-bound} and \sec{local-translation-invariant}.

We will be using the following theorem to show that our construction of a unitary distribution must have a single entry far from its Haar counterpart in Kolmogorov distance in \sec{convergence-lower-bound}.

\begin{theorem}\label{thm:polynomial-to-design}
    For integer $D\geq 1$, let $\nu$ be a unitary distribution on $\mathrm U(D)$.
    If for every polynomial $f$ whose every monomial contains at most $p$ entries of $U$ and at most $q$ entries of $U^*$, it holds that 
    \begin{align}\label{eq:moment-matching-condition}
        (1-\epsilon) \Exp_{U\sim \Haar(D)}\left[|f(U)|^2\right]
        \leq 
        \Exp_{U\sim \nu}\left[|f(U)|^2\right]
        \leq 
        (1+\epsilon) \Exp_{U\sim \Haar(D)}\left[|f(U)|^2\right],
    \end{align} 
    then $\nu$ is a strong $\epsilon$-approximate $(p,q)$-design. 
\end{theorem}
\begin{proof}
    By the assumption, we have that for every polynomial $f$ whose every monomial contains exactly $p$ entries of $U$ and $q$ entries of $U^*$, \eq{moment-matching-condition} holds. 
    Then by \lem{strong-design-polynomial}, we conclude that $\nu$ is a strong $\epsilon$-approximate $(p,q)$-design.
\end{proof}

\section{Kolmogorov convergence to Porter--Thomas}\label{sec:convergence-upper-bound}

In this section, we show that every $\epsilon$-approximate unitary $k$-design (see \rmk{design} for the definition) approximates Porter--Thomas to error 
\begin{align}\label{eq:ks-upper-bound}
O\left(\sqrt{\frac{N}{k(k+N)}}+\epsilon\right) = O\left(\frac{1}{\sqrt k}\cdot \min\left\{1, \sqrt{\frac{N}{k}}\right\}+\epsilon\right)
\end{align}
in Kolmogorov distance,
where $N$ is the dimension of the unitary matrix.
We first provide the proof idea for showing the bound in \eq{ks-upper-bound}; the result is formally stated in \thm{approximate-design-upper-bound}. 

The first step is to establish a discrete measure whose distribution function is close to that of the beta distribution $\Beta(1,N-1)$ in Kolmogorov distance. 
To this end, we invoke the framework of Jacobi--Gauss quadrature to construct the discrete measure. 
In particular, the orthogonal polynomials with respect to $\Beta(1,N-1)$ can be obtained from the Jacobi polynomials by an affine change of variable, and the distribution of the roots of Jacobi polynomials is well understood. 

To show closeness in Kolmogorov distance, we view the distribution function as the expectation of a step function. 
The function is not continuous, and thus cannot be directly related to the closeness of expected polynomials we have shown in \thm{polynomial-relative-error} in \sec{block-encoding}. 
To establish a connection, we consider a degree-$(2k)$ polynomial majorant $U_i$ and a polynomial minorant $L_i$ of the step function. 
Recall that a function $f$ is a majorant of a function $g$ if $f(x)\geq g(x)$ for all $x$, and a function $f$ is a minorant of $g$ if $f(x)\leq g(x)$ for all $x$. 

We choose the quadrature nodes to be the roots of the orthogonal polynomial of degree $k+1$ with respect to the distribution $\Beta(1,N-1)$.
This implies that the expectation of any polynomial of degree at most $2k$ can be calculated using the weights of the quadrature. 
Following this idea, the expected minorant can be represented in terms of weights as well.
Since the difference of the majorant and the minorant is the square of a Lagrange interpolation polynomial, we can show that with respect to $\Beta(1,N-1)$, the distribution function at each quadrature node $x_i$ is located in an interval of length no more than the weight $w_i$. 
By proving that $w_i=O\left(\sqrt{\frac{N}{k(k+N)}}\right)$, we know that at these nodes, the distribution is restricted to a small interval.
For points other than these quadrature nodes, we can reason about their closeness using the monotonicity of the distribution functions. 

Now since the degrees of both the majorant and the minorant are at most $2k$, by \thm{polynomial-relative-error}, we can bound the closeness of the expected majorant and minorant to relative error $O(\epsilon)$ with respect to the distributions induced by the design and the Haar random unitary. 
Thus the distribution induced by the design is bounded in the same interval up to some additive error $O(\epsilon)$ at these nodes as well.
Since both distribution functions are confined to these intervals at the nodes, their difference can be bounded by $O(\max_i w_i+\epsilon)$ accordingly. 

In summary, the entire proof in \sec{convergence-upper-bound} manages to show that 
\begin{align*}
    \text{squared magnitude of Haar's entry} &\approx \text{measure induced by shifted Jacobi--Gauss quadrature} \\
    &\approx \text{squared magnitude of design's entry},
\end{align*}
where the approximation denotes the closeness in Kolmogorov distance.

\subsection{Shifted Jacobi--Gauss quadrature}\label{sec:jacobi-gauss}

    For integers $a\geq 1$ and $D\geq 1$, let $\mu_{a,D}$ be the distribution of $\Beta(a,D)$, and let $\{q_n\}$ be the orthonormal polynomials with respect to $\mu_{a,D}$, i.e., $\deg(q_n)=n$ and
    \begin{align}
        \int_{0}^1 q_n(x) q_m (x) \frac{x^{a-1}(1-x)^{D-1}}{B(a,D)} \d x = \delta_{n,m}.
    \end{align}
    By \eq{beta-jacobi}, $q_n(x)\propto P_n^{(D-1,a-1)}(2x-1)$ is, up to normalization, a shifted Jacobi polynomial of degree $n$ with parameters $(D-1,a-1)$.
    In other words, $Y_n=q_n(X)$ form a sequence of orthonormal random variables in the sense that $\Exp[Y_nY_m]=\delta_{nm}$ when $X\sim\Beta(a,D)$.
    The case of a single entry is $a=1$, and we write $\mu_D:=\mu_{1,D}$ for the distribution of $\Beta(1,D)$; this is the case used in the rest of this section.
    (The probability of a block of $a$ basis states, which appears in \sec{local-translation-invariant}, has the Haar law $\Beta(a,N-a)$; see \rmk{block-kolmogorov} below.)

    Knowing that the orthogonal polynomials are the shifted Jacobi polynomials, we can exploit the Jacobi--Gauss quadrature to construct a moment matching discrete measure. 
    In particular, let $0< x_1<x_2<\ldots <x_\ell< 1$ be the roots of the shifted Jacobi polynomial $q_\ell$ of degree $\ell$; they are real, simple, and lie in the interior of the support (\sec{jacobi}).
    Choosing the quadrature nodes to be the roots of $q_\ell$ allows us to match the expectation of any polynomial of degree at most $2\ell-1$ with respect to the measure $\mu_{a,D}$.
    In more detail, let $\ell_i$ be the Lagrange basis polynomial of degree $\ell-1$ with $\ell_i(x_j)=\delta_{ij}$, and define the weight put on $x_i$ as $w_i:=\int_0^1\ell_i(x)\d\mu_{a,D}(x)$; by Lagrange interpolation, $\int_0^1 r(x)\d\mu_{a,D}(x)=\sum_{i=1}^\ell w_i r(x_i)$ for every polynomial $r$ of degree at most $\ell-1$. 
    Any polynomial $g$ of degree at most $2\ell-1$ can be written as
    \begin{align}
        g(x) = q_\ell(x) h(x) + r(x), 
    \end{align}
    for some polynomial $h$ of degree at most $\ell-1$ and a remainder polynomial $r$ of degree at most $\ell-1$. 
    Since $h$ can always be written as a linear combination of $q_0,q_1,\ldots,q_{\ell-1}$, by the orthogonality of $q_\ell$, we have
    \begin{align}
        \int_0^1 g(x) \d\mu_{a,D}(x) 
        = \int_0^1 r(x) \d\mu_{a,D}(x) 
        = \sum_{i=1}^\ell w_i r(x_i)
        = \sum_{i=1}^\ell w_i g(x_i),
    \end{align}
    where the last equality holds since $q_\ell(x_i)=0$ for all $i\in[\ell]$.
    
    While we will be focusing on the beta distributions, the method for finding moment matching discrete measures can be applied to other distributions.
    For a distribution $\nu$ on the real line and an integer $\ell\geq 1$, the weights of the $\ell$-node quadrature rule can be found through the Christoffel function $\lambda_{\nu,\ell}$, which is defined as
    \begin{align}\label{eq:christoffel-mu-d}
        \lambda_{\nu,\ell}(x) := \min_{\deg(p)\leq\ell-1,\ p(x)=1}\int_{\mathbb R} p(t)^2 \d\nu(t).
    \end{align}
    We use this convention throughout the paper: the index $\ell$ is the number of quadrature nodes, and the admissible polynomials have degree at most $\ell-1$.
    The variational form is useful in determining the upper bound of the weight $w_i$.

    \begin{lemma}\label{lem:christoffel-function-weight}
        For a distribution $\nu$ on $\mathbb R$ with infinite support, let $x_1<\ldots<x_i<\ldots<x_\ell$ be the roots of its orthonormal polynomial of degree $\ell$ (\sec{jacobi}) and $w_1,\ldots,w_\ell$ the weights of the corresponding $\ell$-node Gauss quadrature rule, constructed exactly as above with $\nu$ in place of $\mu_{a,D}$.
        Then $\lambda_{\nu,\ell}(x_i) = w_i$ for all $i\in[\ell]$.
    \end{lemma}
    \begin{proof}
        Proving the direction $w_i\geq\lambda_{\nu,\ell}(x_i)$ is simple: take $p(x)=\ell_i(x)$, the Lagrange basis polynomial. 
        By definition and the exactness of the quadrature for the polynomial $\ell_i^2$ of degree $2\ell-2$, we have $\lambda_{\nu,\ell}(x_i)\leq\int_{\mathbb R} \ell_i(x)^2 \d\nu(x) = \sum_{j=1}^\ell w_j\ell_i(x_j)^2 = w_i$; in particular, $w_i>0$.
        To prove $w_i\leq\lambda_{\nu,\ell}(x_i)$, for every polynomial $p$ of degree at most $\ell-1$, we can use the quadrature rule to write
        \begin{align}
            \int_{\mathbb R} p(x)^2 \d\nu(x) = \sum_{j=1}^\ell w_j p(x_j)^2 \geq w_i p(x_i)^2 = w_i.
        \end{align}
        The inequality follows from the fact that $p^2\geq 0$ and $w_j\geq 0$ for every $j$ (shown above), and the final equality holds since every admissible $p$ satisfies $p(x_i)=1$. 
    \end{proof}

    In particular, we will be using the Christoffel function to bound $w_i$ by considering measures for which we know how to bound the Christoffel function.
    The following alternative form for finding the weight $w_i$ will also be useful.

    \begin{lemma}\label{lem:christoffel-function-orthogonal-polynomial}
        For a distribution $\nu$ on $\mathbb R$ with finite moments and infinite support (as in \sec{jacobi}) and an integer $\ell\geq 1$, let $\lambda_{\nu,\ell}$ be its Christoffel function, defined as in \eq{christoffel-mu-d}.
        Then for 
        \begin{align}
            p_x(t) = \frac{\sum_{j=0}^{\ell-1}q_j(x)q_j(t)}{\sum_{j=0}^{\ell-1}q_j(x)^2},
        \end{align}
        \begin{align}
        \lambda_{\nu,\ell}(x) = \int_{\mathbb R} p_x(t)^2\d\nu(t) = \frac{1}{\sum_{j=0}^{\ell-1} q_j(x)^2},
        \end{align} 
        where $q_j$ is the orthonormal polynomial of degree $j$ with respect to $\nu$.
    \end{lemma}
    \begin{proof}
        Fix any $x$ and write 
        $S(x)=\sum_{j=0}^{\ell-1}q_j(x)^2$.
        For every polynomial $p$ of degree less than $\ell$, $p$ can be written as $p(x)=\sum_{j=0}^{\ell-1} a_j q_j(x)$. 
        Then by orthonormality, we have 
        \begin{align}\label{eq:orthogonality}
            \int_{\mathbb R} p(t)^2 \d\nu(t) = \sum_{j=0}^{\ell-1} |a_j|^2.
        \end{align}
        If $p(x)=1$, then the Cauchy--Schwarz inequality yields
        \begin{align}\label{eq:s-christoffel}
            1 &= \left(\sum_{j=0}^{\ell-1} a_j q_j(x)\right)^2 \leq 
            S(x) \sum_{j=0}^{\ell-1} |a_j|^2 \leq 
            S(x) \int_{\mathbb R} p(t)^2 \d\nu(t).
        \end{align}
        Since the inequality in \eq{s-christoffel} holds for any polynomial $p$ of degree at most $\ell-1$ such that $p(x)=1$, taking the infimum over such $p$ in \eq{christoffel-mu-d} gives $1\leq S(x)\lambda_{\nu,\ell}(x)$. 

        To show $1\geq S(x)\lambda_{\nu,\ell}(x)$, since $p_x(x)=1$ and $\deg(p_x)\leq\ell-1$, $p_x$ is admissible.  
        By \eq{orthogonality}, we have
        \begin{align}
            \int_{\mathbb R} p_x(t)^2 \d\nu(t) = \frac{1}{S(x)^2} \sum_{j=0}^{\ell-1} q_j(x)^2 = \frac{1}{S(x)}.
        \end{align}
        By definition of $\lambda_{\nu,\ell}$ in \eq{christoffel-mu-d}, we have $\lambda_{\nu,\ell}(x)\leq 1/S(x)$. 
    \end{proof}

\lem{christoffel-function-orthogonal-polynomial} turns the variational form in \eq{christoffel-mu-d} to an expression that depends on a finite sum of the orthonormal polynomials. 
We will use the variational form to prove \thm{jacobi-weight-upper-bound}, which bounds the weights for the quadrature nodes; the second form underlies the Chebyshev computation in \lem{exponential}. 

Before doing so, we prove a lemma which bounds the Christoffel function of the exponential measure.

\begin{lemma}[Uniform exponential Christoffel bound]\label{lem:exponential}
 For the mean-one
exponential probability measure $\omega$, $\dd\omega(t)=\e^{-t}\dd t$ on $[0,\infty)$,
let $\lambda_{\omega,\ell}$ be its Christoffel function defined as in \eq{christoffel-mu-d}. 
Then there is an absolute constant $C>0$ such that, for every integer $\ell\ge1$,
\begin{equation}
 \sup_{u\ge0}\lambda_{\omega,\ell}(u)\le\frac{C}{\sqrt\ell}.
 \label{eq:exponential-bound}
\end{equation}
\end{lemma}

\begin{proof}
If $\ell=1$, the only admissible polynomial is $p\equiv1$, so
$\lambda_{\omega,1}(u)=1$ for every $u$. Hence assume $\ell\ge2$ and put $m=\ell-1$.

\textit{Step 1: a Chebyshev polynomial with controlled norm and growth.}
For $R>0$, let $\rho_R$ be the arcsine probability measure \eq{arcsine} on $[0,R]$, let $q_j$ be its orthonormal polynomials \eq{arcsine-basis}, and define
\begin{equation}
 K_{\ell,R}(u,t):=\sum_{j=0}^{\ell-1}q_j(u)q_j(t).
 \label{eq:kernel}
\end{equation}
We claim that
\begin{equation}
 K_{\ell,R}(u,u)\ge\frac\ell4\qquad(0\le u\le R).
 \label{eq:kernel-lower}
\end{equation}
To see this, write $2u/R-1=\cos\theta$. The squared cosines are invariant
under $\theta\mapsto\pi-\theta$, so it suffices to consider
$0\le\theta\le\pi/2$. For $\theta>0$,
\begin{align}
 K_{\ell,R}(u,u)
 &=1+2\sum_{j=1}^{\ell-1}\cos^2(j\theta)
 =\ell-\frac12+
   \frac{\sin((2\ell-1)\theta)}{2\sin\theta}.
 \label{eq:trigonometric-kernel}
\end{align}
If $\sin\theta\ge1/\ell$, the last expression is at least
$(\ell-1)/2\ge\ell/4$. If $\sin\theta<1/\ell$, the inequality
$\sin\theta\ge2\theta/\pi$ gives $\theta<\pi/(2\ell)$. Therefore
$\cos^2(j\theta)\ge1/2$ for
$1\le j\le\lfloor(\ell-1)/2\rfloor$, and
\begin{align}\label{eq:kernel-small-theta}
 K_{\ell,R}(u,u)
 \ge1+\left\lfloor\frac{\ell-1}{2}\right\rfloor
 \ge\frac\ell2.
\end{align}
This also covers $\theta=0$, proving~\eq{kernel-lower}.

For a fixed $u\in[0,R]$, take
\begin{equation}
 p_u(t):=\frac{K_{\ell,R}(u,t)}{K_{\ell,R}(u,u)}.
 \label{eq:test-polynomial}
\end{equation}
Then $\deg p_u\le m$ and $p_u(u)=1$. By orthonormality,
\begin{equation}
 \int_0^R p_u(t)^2\d \rho_R(t)
 =\frac{1}{K_{\ell,R}(u,u)}\le\frac4\ell.
 \label{eq:chebyshev-norm}
\end{equation}
By \eq{chebyshev-growth}, for $t\ge R$,
\begin{align}
 |T_j(2t/R-1)|\le(4t/R)^j.
\end{align}
Since $|q_j(u)|\le\sqrt2$ for $j\ge1$ and $u\in[0,R]$, it follows that
\begin{align}
 |K_{\ell,R}(u,t)|
 \le1+2\sum_{j=1}^{m}(4t/R)^j
 \le2\ell(4t/R)^m.
\end{align}
Together with~\eq{kernel-lower}, this yields
\begin{equation}
 |p_u(t)|\le8(4t/R)^m\qquad(t\ge R).
 \label{eq:polynomial-growth}
\end{equation}

\textit{Step 2: the range $0\le u\le8\ell$.}
Set $R=8\ell$ and use the polynomial~\eq{test-polynomial}.
On $(0,R)$, the exponential measure is dominated by a constant multiple
of the arcsine measure:
\begin{equation}
 \frac{e^{-t}\d t}{\d\rho_R(t)}
 =\pi e^{-t}\sqrt{t(R-t)}
 \le\pi\sqrt R\,\sup_{s\ge0}\sqrt s\,e^{-s}
 =\pi\sqrt{\frac{R}{2e}}.
 \label{eq:interior-domination}
\end{equation}
Consequently,
\begin{equation}
 \int_0^R p_u(t)^2 e^{-t}\d t
 \le\frac{4\pi}{\ell}\sqrt{\frac{R}{2e}}
 =\frac{8\pi}{\sqrt{e\ell}}.
 \label{eq:interior-bound}
\end{equation}
For the remaining tail, \eq{polynomial-growth} gives
\begin{equation}
 \int_R^\infty p_u(t)^2e^{-t}\d t
 \le64\int_R^\infty(4t/R)^{2m}e^{-t}\d t.
 \label{eq:tail-start}
\end{equation}
For $t\ge R$, use $\ln(t/R)\le(t-R)/R$ and $2m/R\le1/4$ to obtain
\begin{align}
 (4t/R)^{2m} e^{-t}
 &=4^{2m} e^{-R}
   \exp\bigl(2m\ln(t/R)-(t-R)\bigr)\\
 &\le4^{2m} e^{-R}\exp\bigl(-3(t-R)/4\bigr).
\end{align}
Integrating this bound in~\eq{tail-start} yields
\begin{equation}
 \int_R^\infty p_u(t)^2e^{-t}\d t
 \le\frac{256}{3}\,4^{2m}e^{-8\ell}
 \le\frac{256}{3}e^{-c\ell},
 \qquad c:=8-2\ln4>0.
 \label{eq:tail-bound}
\end{equation}
The polynomial $p_u$ is admissible in~\eq{christoffel-mu-d}.
Combining~\eq{interior-bound} and~\eq{tail-bound}, we find
\begin{equation}
 \lambda_{\omega,\ell}(u)
 \le\frac{8\pi}{\sqrt{e\ell}}+
       \frac{256}{3}e^{-c\ell}
 =O(\ell^{-1/2})
 \qquad(0\le u\le8\ell).
 \label{eq:moderate-u}
\end{equation}

\textit{Step 3: the range $u>8\ell$.}
Here use $p(t)=(t/u)^m$, which has degree $m$ and satisfies $p(u)=1$.
The gamma integral and the elementary bound $(2m)!\le(2m)^{2m}$ give
\begin{equation}
 \lambda_{\omega,\ell}(u)
 \le\int_0^\infty(t/u)^{2m} e^{-t} \d t
 =\frac{(2m)!}{u^{2m}}
 \le\left(\frac{2m}{u}\right)^{2m}
 \le4^{-2m}.
 \label{eq:large-u}
\end{equation}
Since $m=\ell-1$ and $\ell\ge2$, this is $O(\ell^{-1/2})$ with an
absolute constant. Together with~\eq{moderate-u} and the case
$\ell=1$, this proves~\eq{exponential-bound}.
\end{proof}

Now we are ready to prove an upper bound on the quadrature weights.

\begin{theorem}[Uniform upper bound on quadrature weights]\label{thm:jacobi-weight-upper-bound}
Let $D\ge1$ and let $x_1<\cdots<x_\ell$ and $w_1,\ldots,w_\ell$ be the
nodes and positive weights of the $\ell$-node Gaussian quadrature rule for
$\mu_D=\Beta(1,D)$, where $\ell\ge1$ is an integer. 
Thus
\begin{align}
 \dd\mu_D(x)=D(1-x)^{D-1}\dd x,\qquad 0<x<1,
\end{align}
and the rule is exact for every polynomial of degree at most $2\ell-1$.
There is an absolute constant $C>0$ such that
\begin{equation}
 \max_{1\le i\le\ell}w_i
 \le C\sqrt{\frac{D}{\ell(\ell+D)}}.
 \label{eq:main-bound}
\end{equation}
In particular, the constant is uniform when $D$ and $\ell$ grow together.
\end{theorem}

\begin{proof}
We first identify the weights with the Christoffel function, then prove
two complementary upper bounds.
By \lem{christoffel-function-weight}, 
\begin{equation}
 w_i=\lambda_{\mu_D,\ell}(x_i).
 \label{eq:weight-identity}
\end{equation}

\textit{Step 1: comparison with the arcsine measure.}
Let $\rho_1$ be the measure in~\eq{arcsine} with $R=1$. Its
density ratio with $\mu_D$ is
\begin{equation}
 \frac{\dd\mu_D}{\dd\rho_1}(x)
 =\pi D\sqrt x\,(1-x)^{D-1/2}.
 \label{eq:beta-arcsine-ratio}
\end{equation}
The function $\sqrt x\,(1-x)^{D-1/2}$ is maximized at $x=1/(2D)$.
Indeed, its logarithmic derivative is
\begin{align}
 \frac1{2x}-\frac{D-1/2}{1-x},
\end{align}
which changes sign there from positive to negative. Consequently,
\begin{equation}
 \frac{\dd\mu_D}{\dd\rho_1}(x)
 \le\pi\sqrt{\frac D2}
       \left(1-\frac1{2D}\right)^{D-1/2}
 \le\pi\sqrt{\frac D2}.
 \label{eq:beta-arcsine-domination}
\end{equation}
The Chebyshev polynomial construction in the proof of
\lem{exponential}, specifically~\eq{chebyshev-norm},
shows that $\lambda_{\rho_1,\ell}(x)\le4/\ell$ for $x\in[0,1]$.
This also holds for $\ell=1$, by using the constant polynomial.
Applying the variational definition to~\eq{beta-arcsine-domination}
and then using~\eq{weight-identity}, we obtain
\begin{equation}
 w_i\le\pi\sqrt{\frac D2}\,
              \lambda_{\rho_1,\ell}(x_i)
 \le\frac{4\pi}{\ell}\sqrt{\frac D2}
 =O\!\left(\frac{\sqrt D}{\ell}\right).
 \label{eq:first-weight-bound}
\end{equation}

\textit{Step 2: comparison with the rate-$D$ exponential measure.}
Let $\omega_D$ have density $D\e^{-Dx}$ on $[0,\infty)$, so that $\omega_1=\omega$ is the measure of \lem{exponential}. On $[0,1)$,
\begin{equation}
 \frac{\dd\mu_D}{\dd\omega_D}(x)
 =\e^{Dx}(1-x)^{D-1}\le\e.
 \label{eq:beta-exponential-domination}
\end{equation}
For $D>1$, the logarithmic derivative of this ratio is
$D-(D-1)/(1-x)$, so the maximum occurs at $x=1/D$ and equals
$\e(1-1/D)^{D-1}\le\e$. For $D=1$, the ratio is $\e^x\le\e$.
Extending $\mu_D$ by zero outside $[0,1]$, we have the measure domination
$\mu_D\le\e\omega_D$. It follows that
\begin{equation}
 \lambda_{\mu_D,\ell}(x_i)
 \le\e\,\lambda_{\omega_D,\ell}(x_i).
 \label{eq:variational-domination}
\end{equation}
The substitution $u=Dx$ gives the exact scaling identity
\begin{equation}
 \lambda_{\omega_D,\ell}(x_i)=\lambda_{\omega,\ell}(Dx_i).
 \label{eq:scaling}
\end{equation}
To verify it, associate to every polynomial $p$ the polynomial
$q(u)=p(u/D)$. This preserves degree, takes $p(x_i)=1$ to $q(Dx_i)=1$,
and satisfies
\[
 \int_0^\infty p(x)^2D\e^{-Dx}\dd x
 =\int_0^\infty q(u)^2\e^{-u}\dd u.
\]
By \lem{exponential}, \eq{weight-identity}--\eq{scaling}
therefore imply
\begin{equation}
 w_i\le\e\,\lambda_{\omega,\ell}(Dx_i)=O(\ell^{-1/2}).
 \label{eq:second-weight-bound}
\end{equation}
The uniformity of the lemma is essential here: $Dx_i$ is not restricted
to the interval $[0,1]$.

\textit{Step 3: combine the estimates.}
Equations~\eq{first-weight-bound} and~\eq{second-weight-bound}
give, with an absolute constant $C_0$,
\begin{equation}
 \max_i w_i
 \le C_0\min\left\{\frac{\sqrt D}{\ell},\frac1{\sqrt\ell}\right\}.
 \label{eq:combined-bound}
\end{equation}
For every $D,\ell>0$,
\begin{align}
 \sqrt{\frac{D}{\ell(\ell+D)}}
 \le\min\left\{\frac{\sqrt D}{\ell},\frac1{\sqrt\ell}\right\}
 \le\sqrt{\frac{2D}{\ell(\ell+D)}}.
\end{align}
The right inequality follows by separating $D\le\ell$ and $D\ge\ell$.
Taking $C=\sqrt2C_0$ proves~\eq{main-bound}.
\end{proof}

\subsection{Majorants and minorants of indicator functions}\label{sec:majorant-minorant}

We now show a few lemmas which will be used to prove our main result \thm{approximate-design-upper-bound} for this section. 
The following lemma shows that any non-negative univariate polynomial can be decomposed into a sum of two squares.
Since the idea will be used for proving \lem{minorant-indicator}, we include a brief proof here.

\begin{lemma}[{Sum of squares decomposition \cite[Lemma~3.5]{laurentSumsSquaresMoment2009}}]\label{lem:sos-decomposition}
    Any non-negative univariate polynomial of degree $2d$ is a sum of two squares of polynomials of degree at most $d$.
\end{lemma}
\begin{proof}
   Let $p$ be any non-negative univariate polynomial over the real numbers.
   Since $p$ is non-negative, the roots of $p$ must be either real with even multiplicity or appear in conjugate pairs, i.e.,
   \begin{align}
       p(x) &= c\prod_{i=1}^{m}(x-r_i)^{2k_i}\prod_{j=1}^{n}(x-z_j)^{s_j}(x-\bar z_j)^{s_j},
    \end{align}
    for some real number $c>0$. 
    Then define 
    \begin{align}
        h(x) = \sqrt{c}\prod_{i=1}^{m}(x-r_i)^{k_i}\prod_{j=1}^{n}(x-z_j)^{s_j}.
    \end{align}
    Let $h^*(x)$ be the polynomial obtained by taking the complex conjugate of the coefficients of $h(x)$. Then we have $p(x)=h(x)h^*(x)$.
    Let $h(x)=r(x)+\i s(x)$ for real polynomials $r$ and $s$. Then we have $h^*(x)=r(x)-\i s(x)$, and therefore $p(x) = r(x)^2 + s(x)^2$.
\end{proof}

For $x_i\in\mathbb R$, we write $1_{x_i}$ for the indicator function of $(-\infty,x_i]$, i.e., $1_{x_i}(x)=1$ if $x\leq x_i$ and $1_{x_i}(x)=0$ otherwise.

\begin{lemma}[Majorant of indicator function]\label{lem:majorant-indicator}
    For integer $\ell>1$, let $-\infty<x_1<x_2<\ldots<x_{\ell}<\infty$ be $\ell$ real numbers. 
    Then for every $i\in[\ell]$, the indicator function $1_{x_i}(x)$ has a polynomial majorant $U_i$ of degree at most $2\ell-2$ over $\mathbb R$ satisfying the following conditions:
    \begin{enumerate}
        \item $U_i(x)\geq 1_{x_i}(x)$ for all $x\in\mathbb R$.
        \item $U_i(x_j)=1$ if $j\leq i$ and $U_i(x_j)=0$ if $j>i$.
        \item $U_i'(x_j)=0$ for all $j\in[\ell]\setminus\{i\}$.
    \end{enumerate}
\end{lemma}
\begin{proof}
    The second and third conditions impose $2\ell-1$ linear constraints on the $2\ell-1$ coefficients of a polynomial $U_i$ of degree at most $2\ell-2$.
    The corresponding linear system is nonsingular: a polynomial of degree at most $2\ell-2$ satisfying the homogeneous constraints vanishes at the $\ell$ nodes and has vanishing derivative at $\ell-1$ of them, i.e., it has at least $2\ell-1$ roots counted with multiplicity, so it is identically zero.
    Therefore $U_i$ exists and is uniquely determined by these constraints.
    If $i=\ell$, the constant polynomial $U_\ell\equiv 1$ satisfies all three conditions, so we assume $i\leq\ell-1$ in the rest of the proof.

    The degree of $U_i'$ is at most $2\ell-3$. 
    To see why $U_i$ is a majorant of the indicator function, 
    we first observe that $U_i(x_1)=\ldots=U_i(x_i)$ and $U_{i}(x_{i+1})=\cdots=U_i(x_\ell)$, and by Rolle's theorem, $U_i'$ must have $\ell-2$ roots in the intervals $(x_1,x_2),\ldots,(x_{i-1},x_i)$ and $(x_{i+1},x_{i+2}),\ldots,(x_{\ell-1},x_\ell)$.
    Let $y_j\in (x_j,x_{j+1})$ for $j\in[\ell-1]\setminus\{i\}$ be a root of $U_i'$.
    Combining these derived constraints (from the second condition) and the third condition, $U_i'$ has all $2\ell-3$ roots located.
    Thus all roots of $U_i'$ are real and simple. 
    
    Since $U_i'$ has no root in $(x_i,x_{i+1})$ (recall that $x_i$ is not among its roots) and $U_i(x_{i+1})=0<1=U_i(x_i)$, we have $U_i'<0$ in this interval, and this implies $U_i'>0$ in the interval $(x_{i+1},y_{i+1})$ and $y_{i+1}$ must be a local maximum (when $i=\ell-1$ there is no $y_\ell$, and $(x_{i+1},y_{i+1})$ is to be read as $(x_\ell,\infty)$, which is treated in the last paragraph of the proof).
    Repeating the same argument, we have seen $x_{i+1},\ldots,x_{\ell}$ are all local minima, and thus $U_i(x)\geq 0$ for $x\geq x_i$. 

    Similarly, since $U_i'(y_{i-1})=0$ and $U_i'(x_{i+1})=0$ and there is no root in $(y_{i-1}, x_{i+1})$, we must have $U_i'<0$ in the interval $(y_{i-1}, x_{i+1})$.
    Therefore $y_{i-1}$ is a local maximum, and for $x\in(x_{i-1},y_{i-1})$, $U_i'(x)>0$. 
    This implies that $x_{i-1}$ is also a local minimum. 
    Repeating the same argument, we have seen $x_1,\ldots,x_{i-1}$ are all local minima, and thus $U_i(x)\geq 1$ for $x_1\leq x\leq x_i$ (for $i=1$ there is nothing to show here).
    Finally, since all roots of $U_i'$ are simple and lie in $[x_1,x_\ell]$, $U_i'$ has a constant sign on each of the unbounded intervals $(-\infty,x_1)$ and $(x_\ell,\infty)$: $U_i$ is increasing on $(x_\ell,\infty)$ because $x_\ell$ is a local minimum, so $U_i(x)\geq U_i(x_\ell)=0$ for $x\geq x_\ell$; and $U_i$ is decreasing on $(-\infty,x_1)$, because $x_1$ is a local minimum when $i\geq 2$ and because $U_1'<0$ on $(-\infty,x_2)$ when $i=1$, so $U_i(x)\geq U_i(x_1)=1$ for $x\leq x_1$.
    This concludes that $U_i$ is a majorant of the indicator function $1_{x_i}$.
\end{proof}

\begin{lemma}[Minorant of indicator function]\label{lem:minorant-indicator}
    For integer $\ell>1$, let $-\infty<x_1<x_2<\ldots<x_{\ell}<\infty$ be $\ell$ real numbers. 
    Then for every $i\in[\ell]$, the indicator function $1_{x_i}(x)$ has a polynomial minorant $L_i$ of degree at most $2\ell-2$ over $\mathbb R$ satisfying the following conditions:
    \begin{enumerate}
        \item $L_i(x)\leq 1_{x_i}(x)$ for all $x\in\mathbb R$.
        \item $L_i(x_j)=1$ if $j<i$ and $0$ if $j\geq i$.
        \item $L_i'(x_j)=0$ for all $j\in[\ell]\setminus\{i\}$.
    \end{enumerate}
    Furthermore, $L_i=U_i-\ell_i^2$, where $\ell_i(x)=\prod_{j\in[\ell]\setminus\{i\}}\frac{x-x_j}{x_i-x_j}$ is the Lagrange interpolation polynomial.
\end{lemma}
\begin{proof}
    The proof is similar to the proof of \lem{majorant-indicator}.
    As in that proof, the last two conditions determine $L_i$ uniquely among polynomials of degree at most $2\ell-2$.
    If $i=1$, the zero polynomial $L_1\equiv 0$ satisfies all three conditions, so we assume $i\geq 2$ in the rest of the proof.
    
    The degree of $L_i'$ is at most $2\ell-3$.
    We observe that from the second condition and by Rolle's theorem, $L_i'$ must have $\ell-2$ roots in the intervals $(x_1,x_2),\ldots,(x_{i-2},x_{i-1})$ and $(x_{i},x_{i+1}),\ldots,(x_{\ell-1},x_\ell)$.
    Let $y_j\in (x_j,x_{j+1})$ for $j\in[\ell-1]\setminus\{i-1\}$ be a root of $L_i'$.
    Then we have seen $L_i'$ has all $2\ell-3$ roots located, in order $x_1<y_1<\ldots<x_{i-2}<y_{i-2}<x_{i-1}<y_i<x_{i+1}<y_{i+1}<\ldots<x_{\ell-1}<y_{\ell-1}<x_\ell$ (for $i=2$ the list starts with $x_1<y_2$, and for $i=\ell$ it ends with $y_{\ell-2}<x_{\ell-1}$), and they are all distinct. 

    To see why $L_i$ is a minorant, notice that $L_i'$ has no root in $(x_{i-1},y_i)$ (recall that $x_i$ is not among its roots) and $L_i(x_i)=0<1=L_i(x_{i-1})$, so $L_i'<0$ in the interval $(x_{i-1},y_i)$ and $y_i$ must be a local minimum (for $i=\ell$ the interval is $(x_{\ell-1},\infty)$, which contains no root of $L_\ell'$, so $L_\ell$ is decreasing there and $x_{\ell-1}$ is a local maximum).
    Since $L_i$ must alternate between local maxima and local minima at the (simple) roots of $L_i'$ in order, this implies that $x_1,\ldots,x_{i-1},x_{i+1},\ldots,x_\ell$ are local maxima and the $y_j$ are local minima; hence $L_i\leq 1$ on $[x_1,x_{i-1}]$ and $L_i\leq 0$ on $[x_i,x_\ell]$, while on $(x_{i-1},x_i)$ the function $L_i$ decreases from $1$ to $0$, so that $L_i\leq 1_{x_i}$ there.
    On the unbounded intervals $(-\infty,x_1)$ and $(x_\ell,\infty)$, which contain no roots of $L_i'$, $L_i$ is monotone with $L_i\leq L_i(x_1)=1$ and $L_i\leq L_i(x_\ell)=0$ respectively, because $x_1$ and $x_\ell$ are local maxima (for $i=\ell$, $L_\ell$ is decreasing on $(x_{\ell-1},\infty)$ as shown).
    Thus $L_i(x)\leq 1_{x_i}(x)$ for all $x\in\mathbb R$.

    Recall that in the proof of \lem{sos-decomposition}, we have seen that a non-negative polynomial must have real roots of even multiplicity and conjugate pairs of complex roots.
    Finally, since $U_i$ and $L_i$ intersect at $x_1,\ldots,x_{i-1},x_{i+1},\ldots,x_\ell$ and $U_i-L_i$ is a non-negative polynomial of degree at most $2\ell-2$, we have
    \begin{align}
        U_i-L_i = c \cdot \prod_{j\neq i} (x-x_j)^2. 
    \end{align}
    Now given that $U_i(x_i)-L_i(x_i)=1-0=1$, we have $c=\prod_{j\neq i} (x_i-x_j)^{-2}$, and thus $U_i-L_i = \ell_i^2$.
\end{proof}

\subsection{Upper bound on the Kolmogorov distance}

The majorant and minorant also yield the following classical bracketing of a distribution function at the nodes of its Gauss quadrature (the Chebyshev--Markov--Stieltjes inequalities), which we use in the proof of the main theorem and again in \sec{convergence-lower-bound}.

\begin{lemma}[Bracketing at the quadrature nodes]\label{lem:markov-stieltjes}
    Let $\mu$ be a probability distribution on $\mathbb R$ with distribution function $F_\mu$, and for an integer $\ell\geq 2$ let $x_1<\cdots<x_\ell$ and $w_1,\ldots,w_\ell$ be the nodes and weights of an $\ell$-node quadrature rule that is exact for $\mu$ on all polynomials of degree at most $2\ell-2$.
    Then, with $S_i:=\sum_{j\leq i}w_j$ and $S_0:=0$,
    \begin{align}
        S_{i-1}\leq F_\mu(x_i)\leq S_i \qquad\text{for all } i\in[\ell].
    \end{align}
\end{lemma}
\begin{proof}
    By \lem{majorant-indicator} and \lem{minorant-indicator}, $L_i\leq 1_{x_i}\leq U_i$ with $\deg U_i,\deg L_i\leq 2\ell-2$.
    Taking expectations under $\mu$ and using exactness, $F_\mu(x_i)=\Exp_{\mu}[1_{x_i}(X)]\leq\Exp_{\mu}[U_i(X)]=\sum_j w_jU_i(x_j)=S_i$ and $F_\mu(x_i)\geq\Exp_{\mu}[L_i(X)]=\sum_j w_jL_i(x_j)=S_{i-1}$.
\end{proof}

Given the tools established in \sec{jacobi-gauss} and \sec{majorant-minorant}, we are now ready to prove the main theorem of this section, which gives an upper bound on the Kolmogorov distance between the distribution function of a unitary $k$-design and that of the Haar measure.

\begin{theorem}\label{thm:approximate-design-upper-bound}
    For integers $N\geq 2$ and $k\geq 1$, and $\epsilon\leq 1$, let $\mu(k,\epsilon)$ be any $\epsilon$-approximate unitary $k$-design on $\mathrm U(N)$ (see \sec{design}), let $\nu(k,\epsilon)$ be the distribution of $X=|\bra{0}U\ket{0}|^2$ for $U\sim\mu(k,\epsilon)$, and let $\nu=\Beta(1,N-1)$ be the distribution of the same quantity for a Haar random $U$.
    Then
    \begin{align}
        d_{\mathrm K}\big(\nu(k,\epsilon),\nu\big)=\sup_{x\in\R}|F_{\nu(k,\epsilon)}(x)-F_{\nu}(x)| = O\left(\sqrt{\frac{N}{k(k+N)}}+\epsilon\right).
    \end{align}
The same proof applies to any fixed entry $|\bra{x}U\ket{y}|^2$, since \thm{polynomial-relative-error} holds for every entry.
\end{theorem}
\begin{proof}
    We choose 
    $\ell=k+1$
    to be the number of quadrature nodes.
    Let $0< x_1 < x_2 < \cdots < x_\ell < 1$ be the shifted Jacobi--Gauss quadrature nodes for $\nu=\mu_{N-1}=\Beta(1,N-1)$ (\sec{jacobi-gauss} with $D=N-1$), and let $w_i$ be the corresponding weight of $x_i$. 
    By the exactness of the Gauss quadrature rule (\sec{jacobi-gauss}), for every polynomial $f$ of degree at most $2\ell-1$, we have
    \begin{align}
        \Exp_{X\sim\Beta(1,N-1)}[f(X)] = \sum_{i=1}^\ell w_i f(x_i).
    \end{align}
    For every non-negative polynomial $f$ of degree at most $2k$, by \lem{sos-decomposition}, we can decompose it into a sum of two squares, i.e., $f(x)=p(x)^2+q(x)^2$, where $\deg(p),\deg(q)\leq k$. 
    Then by \thm{polynomial-relative-error} applied to $p$ and to $q$, we have
    \begin{align}
        &(1-\epsilon)\sum_{i=1}^\ell w_i f(x_i) = (1-\epsilon)\Exp_{X\sim\nu}[f(X)] \leq \Exp_{X\sim\nu(k,\epsilon)}[f(X)] \nonumber\\
        &\qquad\leq (1+\epsilon)\Exp_{X\sim\nu}[f(X)] = (1+\epsilon)\sum_{i=1}^\ell w_i f(x_i).
    \end{align}
    Now consider the indicator function $1_{x_i}$ of $(-\infty,x_i]$.
    Let $S_i=\sum_{m=1}^i w_m$.
    By \lem{majorant-indicator}, there is a degree-$(2\ell-2)$ polynomial majorant $U_i$ of $1_{x_i}$ over $\mathbb R$. 

    Since $U_i$ is a majorant of $1_{x_i}$, it is non-negative, and we have
    \begin{align}
        F_{\nu(k,\epsilon)}(x_i) 
        &= \Exp_{X\sim\nu(k,\epsilon)}[1_{x_i}(X)] 
        \leq \Exp_{X\sim\nu(k,\epsilon)}[U_i(X)] 
        \leq (1+\epsilon)\Exp_{X\sim\nu}[U_i(X)] \\
        &= (1+\epsilon)\sum_{j=1}^\ell w_j U_i(x_j) 
        = (1+\epsilon)\sum_{j=1}^i w_j = (1+\epsilon)S_i.
    \end{align}
    The second equality follows from the fact that $U_i$ is a polynomial of degree at most $2\ell-2$, and the third follows from the second property of $U_i$ in \lem{majorant-indicator}.
    On the other hand, by \lem{minorant-indicator}, we can write $L_i$ as the difference of two non-negative functions $L_i=U_i-\ell_i^2$ with each of degree no more than $2\ell-2$, and therefore
    \begin{align}
        F_{\nu(k,\epsilon)}(x_i)
        &\geq \Exp_{X\sim\nu(k,\epsilon)}[L_i(X)] 
        = \Exp_{X\sim\nu(k,\epsilon)}[U_i(X)-\ell_i^2(X)] \\
        &\geq (1-\epsilon)\Exp_{X\sim\nu}[U_i(X)] - (1+\epsilon)\Exp_{X\sim\nu}[\ell_i^2(X)].
    \end{align}
    The second expectation can be obtained using the quadrature rule:
    \begin{align}
        \Exp_{X\sim\nu}[\ell_i^2(X)] = \sum_{j=1}^\ell w_j \ell_i^2(x_j) = w_i \ell_i^2(x_i) = w_i.
    \end{align}
    This implies that we can bound
    \begin{align}
        F_{\nu(k,\epsilon)}(x_i) \geq (1-\epsilon)S_i - (1+\epsilon)w_i.
    \end{align}
    We take the convention that $x_0=0$, $x_{\ell+1}=1$, $w_0=w_{\ell+1}=0$, $S_{-1}=S_0=0$, and $S_{\ell+1}=1$.
    The intervals $[x_i,x_{i+1})$, $0\leq i\leq\ell$, partition $[0,1)$, and for $x\notin[0,1)$ both distribution functions coincide, so it suffices to consider $x\in[x_i,x_{i+1})$ for some $0\leq i\leq\ell$. By monotonicity of distribution functions, for such $x$ we have 
    \begin{align}
        (1-\epsilon)S_i - (1+\epsilon)w_i \leq F_{\nu(k,\epsilon)}(x) \leq (1+\epsilon)S_{i+1}
    \end{align}
    and, by \lem{markov-stieltjes} and monotonicity (for $i=0$ and $i=\ell$, the inequalities involving $x_0=0$ and $x_{\ell+1}=1$ follow from the conventions above, since $F_\nu(0)=0$ and $F_\nu(1)=1$), 
    \begin{align}
        S_{i-1} \leq F_\nu(x_i) \leq F_\nu(x) \leq F_\nu(x_{i+1}) \leq S_{i+1}.
    \end{align}
    We have 
    \begin{align}
        F_{\nu(k,\epsilon)}(x) - F_\nu(x) 
        &\leq (1+\epsilon)S_{i+1} - S_{i-1} 
        \leq (1+\epsilon)w_{i+1} + w_i + \epsilon S_i 
        \leq (2+\epsilon)\max_i w_i + \epsilon,
    \end{align}
    and 
    \begin{align}
        F_\nu(x) - F_{\nu(k,\epsilon)}(x) \leq S_{i+1} - (1-\epsilon)S_i + (1+\epsilon)w_i \leq (2+\epsilon)\max_i w_i + \epsilon.
    \end{align}
    By \thm{jacobi-weight-upper-bound} applied with $D=N-1$ and $\ell=k+1$, $\max_i w_i=O\big(\sqrt{(N-1)/((k+1)(k+N))}\big)=O\big(\sqrt{N/(k(k+N))}\big)$, and we conclude that
    \begin{align}
        \sup_{x\in\R}|F_{\nu(k,\epsilon)}(x)-F_\nu(x)| \leq (2+\epsilon)\max_i w_i + \epsilon = O\left(\sqrt{\frac{N}{k(k+N)}}+\epsilon\right),
    \end{align}
    where we used $\epsilon\leq 1$ to absorb the factor $2+\epsilon$ into the constant.
\end{proof}

\begin{remark}\label{rmk:block-kolmogorov}
    The proof of \thm{approximate-design-upper-bound} uses the specific law $\Beta(1,N-1)$ only through the weight bound of \thm{jacobi-weight-upper-bound}.
    The remaining steps apply to the probability $X_k$ of a fixed block of $a$ basis states in a fixed column, whose Haar counterpart is $\Beta(a,N-a)$ (\lem{block-factorization}) and to which \thm{polynomial-relative-error} applies, with the Gauss quadrature rule of $\Beta(a,N-a)$ in place of that of $\Beta(1,N-1)$: at the nodes $x_i$ of that rule one obtains $|F_{X_k}(x_i)-F_{\Beta(a,N-a)}(x_i)|\leq(1+\epsilon)w_i+\epsilon$, and hence a Kolmogorov bound $O(\max_i w_i+\epsilon)$.
    We do not pursue this generalization, since the statements about block probabilities needed in \sec{local-translation-invariant} are obtained there by different means: for $a\geq 2\pi k$ through characteristic functions (\thm{convergence-distribution}), and for $a=2$ through a pointwise comparison of the distribution functions in terms of the Christoffel function of $\Beta(2,N-2)$ (\lem{pointwise-df-to-christoffel} and \lem{pointwise-christoffel}), which in particular yields the Kolmogorov bound $O(k^{-1/2}+\epsilon)$ for the probability of a block of two basis states.
\end{remark}

\section{Optimality of the convergence bounds}\label{sec:convergence-lower-bound}

In this section, we show that there exists a unitary $k$-design where a single entry has a distribution far from that of its Haar counterpart.
Throughout this section, $D$ denotes the dimension of the unitary matrices ($D=N=2^n$ for an $n$-qubit design), so that $\Beta(1,D-1)$ is the Porter--Thomas law of a single entry.
In particular, we show that the Kolmogorov distance is 
\begin{align}\label{eq:ks-lower-bound}
    \Omega\left(\sqrt{\frac{D}{k(D+k)}}+\epsilon\right).
\end{align} 

To see how we show this, we first note that in \sec{convergence-upper-bound}, we have seen that there is a discrete measure that is close to $\Beta(1,D-1)$ in Kolmogorov distance bounded by $O(\max_i w_i)$, the maximum of the shifted Jacobi--Gauss quadrature weights $w_i$. 
In this section, we show this bound is tight: the Kolmogorov distance is at least $\Omega(\max_i w_i)$.
Then by embedding the discrete measure into a unitary distribution so that the resulting unitary forms an exact design (i.e., the relative error in \dfn{strong-design} is $\epsilon=0$) and the top-left entry is marginally distributed according to the discrete measure, we show that there is a design which fails to approximate a Haar random unitary to error $o(\max_i w_i)$ in Kolmogorov distance. 
Following this idea closely, by perturbing the weight of the discrete measure by an additive error $O(\epsilon)$, we can further strengthen our result to obtain an approximate design that justifies our bound in \eq{ks-lower-bound}.

In \sec{induced-measure}, we show how to embed a moment-matching distribution on $[0,1]$ into a unitary ensemble so that the resulting distribution forms a unitary design. 
Then in \sec{worst-case}, we establish a lower bound on the quadrature weight, and a linear lower bound on the Kolmogorov distance in the maximum quadrature weight.

\subsection{Unitary designs from moment matching distributions}\label{sec:induced-measure}

Our first step to establish a lower bound is to embed a moment-matching distribution $p$ (in the sense of \eq{moment-matching} below) into a unitary distribution so that marginally an entry's distribution is $p$. 
We first define the unitary distribution induced by a marginal distribution. 

\begin{definition}[{Unitary distribution induced by distribution on $[0,1]$}]\label{dfn:induced-measure}
    For integer $D>1$, let $L,R\sim\Haar(D-1)$ be independent Haar random unitaries of dimension $D-1$ and let $\Theta\sim\Uniform(0,2\pi)$ be independent of $(L,R)$; when $X\sim p$ is random below, it is taken independent of $(L,R,\Theta)$. For $x\in[0,1]$, define
    \begin{align}
        W(x) := \begin{bmatrix}
            \sqrt{x}e^{\i\Theta} & \sqrt{1-x} e_1^\top \\
            \sqrt{1-x} e_1 & -\sqrt{x}e^{-\i\Theta} e_1e_1^\top + (\Id-e_1e_1^\top)
        \end{bmatrix},
    \end{align}
    where $e_1$ is the first standard basis vector in $\mathbb C^{D-1}$, 
    and
    \begin{align}
        U(x) := 
        \begin{bmatrix}
            1 & 0 \\
            0 & L
        \end{bmatrix} 
        W(x) \begin{bmatrix}
            1 & 0 \\
            0 & R
        \end{bmatrix},
    \end{align}
    and let $\mu(p)$ be the distribution of $U(X)$ for $X\sim p$.
    We say that $\mu(p)$ is a distribution on $\mathrm U(D)$ induced by the distribution $p$ on $[0,1]$.
    For fixed $x\in[0,1]$, we say $U(x)$ is the unitary distribution induced by $x$. 
\end{definition}

We show that a unitary distribution induced by $\Beta(1,D-1)$ resembles a Haar random unitary on $\mathrm U(D)$. 
When we replace the beta distribution with a moment-matching distribution $p$, the Kolmogorov distance of that entry is exactly that between $\Beta(1,D-1)$ and $p$. 

\begin{lemma}\label{lem:haar-induced}
    For integer $D>1$, the unitary distribution on $\mathrm U(D)$ induced by $\Beta(1,D-1)$ (see \dfn{induced-measure} for the definition) is the Haar measure on $\mathrm U(D)$.
\end{lemma}
\begin{proof}
    Let $X\sim\Beta(1,D-1)$ and $\Theta\sim\Uniform(0,2\pi)$ be independent random variables. 
    Following the notation in \dfn{induced-measure}, the first column of $U(X)$ is 
    \begin{align}
        \begin{bmatrix}
            \sqrt X e^{\i\Theta} \\
            \sqrt{1-X} L e_1
        \end{bmatrix},
    \end{align}
    where $L\sim\Haar(D-1)$ is independent of $X$ and $\Theta$. 
    Since $Le_1$ is uniform on the unit sphere in $\mathbb C^{D-1}$, the first column of $U(X)$ is uniform on the unit sphere in $\mathbb C^D$, which is exactly the same distribution as the first column of a Haar random unitary matrix: by the Gaussian representation of \sec{preliminaries-beta}, a uniform unit vector can be written as $g/\|g\|$ for a standard complex Gaussian vector $g\in\mathbb C^D$, and $|g_1|^2/\|g\|^2\sim\Beta(1,D-1)$, $\arg g_1\sim\Uniform(0,2\pi)$, and the direction $g_{2:D}/\|g_{2:D}\|$, which is uniform on the unit sphere in $\mathbb C^{D-1}$, are independent, so that $g/\|g\|$ has the law of the displayed vector. 
    Conditioned on the first column, the remaining $D-1$ columns of $U(X)$ form a basis in the subspace orthogonal to the first column. 
    Multiplication on the right by $R$ replaces this basis by a uniformly random orthonormal basis in the same subspace, and thus $U(X)$ has the same distribution as a Haar random unitary matrix.
\end{proof}

Next, we show this embedding gives a unitary design if the marginal distribution is moment-matching.

\begin{theorem}\label{thm:induced-exact-design}
    For every distribution $p$ on $[0,1]$
    whose moments of order up to $2k$ agree with those of $\Beta(1,D-1)$, i.e., such that
    for every polynomial $f$ of degree at most $2k$,
    \begin{align}\label{eq:moment-matching}
        \Exp_{X\sim p}[f(X)] = \Exp_{X\sim\Beta(1,D-1)}[f(X)],
    \end{align}
    the unitary distribution induced by $p$ (see \dfn{induced-measure} for the definition) is an exact unitary $k$-design.
\end{theorem}
\begin{proof}

    By \dfn{induced-measure}, let $U(x)$ be the unitary distribution induced by $x\in[0,1]$.
    For ease of notation, we define $c(x)=\sqrt x e^{\i\Theta}$ and $s(x)=\sqrt{1-x}$ for $x\in[0,1]$. 
    Then, writing $L,R$ also for the block matrices $[1]\oplus L$ and $[1]\oplus R$ with $L,R\sim\Haar(D-1)$ (a slight abuse of notation), the entries of $U(x)=LW(x)R$ can be written as linear polynomials of $c(x)$, $c^*(x)$, and $s(x)$. 
    By \lem{haar-induced}, when $X\sim\Beta(1,D-1)$, $U(X)$ is Haar random. 
    We now show that $U(X)$ for $X\sim p$ is an exact unitary $k$-design.

    For every polynomial $f$ of degree at most $4k$ in the matrix entries of $U$ and their conjugates, $f(U(x))=g_{L,R}(c(x),c^*(x),s(x))$ can be written as a polynomial $g_{L,R}$ of degree at most $4k$ in $c$, $c^*$, and $s$ whose coefficients only depend on the entries of $L,R$.
    Since $L,R$ are invariant under left and right multiplication by any unitary matrix on the $(D-1)$-dimensional subspace, for $M=[1]\oplus\diag(-1,1,1,\ldots,1)$, we have $\Exp[f(U)]=\Exp[f(LMWMR)]$ (conjugating $W(x)$ by $M$ flips the sign of $s(x)$ and leaves $c(x)$ unchanged), and consequently, for $X$ sampled either from $p$ or $\Beta(1,D-1)$, we have
    \begin{align}
        \Exp[f(U)] 
        &= \Exp_{L,R,\Theta,X}[f(LW(X)R)] 
        = \Exp_{L,R,\Theta,X}[g_{L,R}(C,C^*,S)] 
        = \Exp_{L,R,\Theta,X}[g_{L,R}(C,C^*,-S)],
    \end{align}
    where $C=c(X)$ and $S=s(X)$.
    Hence $\Exp[f(U)]=\Exp_{L,R,\Theta,X}[g^{\mathrm e}_{L,R}(C,C^*,S)]$, where $g^{\mathrm e}_{L,R}(c,c^*,s):=\frac12\big(g_{L,R}(c,c^*,s)+g_{L,R}(c,c^*,-s)\big)$ is the even part of $g_{L,R}$ in the variable $s$, i.e., the sum of those monomials of $g_{L,R}$ in which the power of $s$ is even.

    Furthermore, $\Theta$ is independent of $(L,R,X)$ and uniform on $[0,2\pi)$, and every monomial of $g^{\mathrm e}_{L,R}$ evaluated at $(C,C^*,S)$ takes the form $C^m (C^*)^n S^{2r}=X^{(m+n)/2}(1-X)^r e^{\i(m-n)\Theta}$ for some $m,n,r\in\mathbb{N}$ satisfying $m+n+2r\leq 4k$; its expectation over $\Theta$ vanishes unless $m=n$, in which case it equals $X^m(1-X)^r$ with $m+r\leq 2k$.
    Therefore, for each fixed $L,R$, the function $\tilde g_{L,R}(x):=\Exp_\Theta\big[g^{\mathrm e}_{L,R}(c(x),c^*(x),s(x))\big]$ is a polynomial in $x$ of degree at most $2k$, and $\Exp[f(U)]=\Exp_{L,R}\big[\Exp_X[\tilde g_{L,R}(X)]\big]$ for $X$ sampled either from $p$ or $\Beta(1,D-1)$.

    Now by \eq{moment-matching}, applied to the polynomial $\tilde g_{L,R}$ of degree at most $2k$ for each fixed $L,R$, we have
    \begin{align}
        \Exp_{U\sim\Haar(D)}[f(U)] 
        &= \Exp_{X\sim\Beta(1,D-1),L,R\sim\Haar(D-1)}[\tilde g_{L,R}(X)] \\
        &= \Exp_{X\sim p,L,R\sim\Haar(D-1)}[\tilde g_{L,R}(X)] 
        = \Exp_{U\sim\mu(p)}[f(U)].
    \end{align}
    Applying this identity to $f=|h|^2$ for every polynomial $h$ whose monomials contain $k$ entries of $U$ and $k$ entries of $U^*$, \thm{polynomial-to-design} shows that $\mu(p)$ is an exact unitary $k$-design.
\end{proof}

Now given the construction of an exact unitary $k$-design from a moment matching distribution, we can lift the statement in \thm{induced-exact-design} to present an explicit construction of an $\epsilon$-approximate unitary design induced by an approximately moment-matching distribution.

\begin{theorem}\label{thm:induced-approximate-design}
    Let $\epsilon\in[0,1]$. Let $p$ be a distribution on $[0,1]$ whose moments of order up to $2k$ agree with those of $\Beta(1,D-1)$, i.e., \eq{moment-matching} holds for every polynomial $f$ of degree at most $2k$.
    Suppose that there exists a measurable function $h:[0,1]\to \mathbb R$ such that
    \begin{align}
        \Exp_{X\sim p}[h(X)]=0, \quad \sup_{x\in[0,1]}|h(x)|\leq 1,
    \end{align}
    and define the distribution $p_\epsilon$ by $\d p_\epsilon(x):=(1+\epsilon h(x))\d p(x)$; it is a probability distribution since $1+\epsilon h\geq 0$ and $\Exp_{X\sim p}[h(X)]=0$, and it satisfies $\Exp_{X\sim p_\epsilon}[f(X)] = \Exp_{X\sim p}[f(X)] + \epsilon \Exp_{X\sim p}[h(X) f(X)]$ for every integrable $f$.
    Then the unitary distribution induced by $p_\epsilon$ is an $\epsilon$-approximate unitary $k$-design (see \rmk{design}).
\end{theorem}
\begin{proof}
    By \thm{induced-exact-design}, the unitary distribution induced by $p$ is an exact unitary $k$-design. 
    Let $\mu(p)$ be the distribution of the unitary induced by $p$.
    Then let $\Delta_\mu=\Delta_{\mu}^{(k,k)}$, where $\Delta_{\mu}^{(k,k)}$ is defined as in \dfn{strong-design}, and let $\Delta_x$ be the channel induced by any $x\in[0,1]$. 
    By \thm{induced-exact-design}, we have
    \begin{align}
        \Exp_{X\sim p}[\Delta_X] = \Exp_{X\sim\Beta(1,D-1)}[\Delta_X].
    \end{align}
    Writing $\Delta_r:=\Exp_{X\sim r}[\Delta_X]$ for a distribution $r$ on $[0,1]$, and noting that each $\Delta_x$ is completely positive (it is an average of unitary conjugations) and that $\Delta_{p_\epsilon}=\Exp_{X\sim p}[(1+\epsilon h(X))\Delta_X]$ by the definition of $p_\epsilon$, we have
    \begin{align}
        \Delta_{p_\epsilon} - (1-\epsilon)\Delta_{\Beta(1,D-1)}
        &= 
        \Delta_{p_\epsilon} - (1-\epsilon)\Delta_{p} \\
        &= \epsilon \Exp_{X\sim p}[(1+h(X))\Delta_X] \geq 0,
    \end{align}
    where the inequality is in the completely positive order since $1+h\geq 0$. Similarly, we have
    \begin{align}
        (1+\epsilon)\Delta_{\Beta(1,D-1)} - \Delta_{p_\epsilon} 
        &= (1+\epsilon)\Delta_{p} - \Delta_{p_\epsilon} \\
        &= \epsilon \Exp_{X\sim p}[(1-h(X))\Delta_X] \geq 0.
    \end{align}
    We have shown that
    \begin{align}
        (1-\epsilon)\Delta_{\Beta(1,D-1)} \leq \Delta_{p_\epsilon} \leq (1+\epsilon)\Delta_{\Beta(1,D-1)},
    \end{align}
    and consequently, since $\Delta_{\Beta(1,D-1)}=\Delta_{\Haar}^{(k,k)}$ by \lem{haar-induced}, the unitary distribution induced by $p_\epsilon$ is an $\epsilon$-approximate unitary $k$-design.
\end{proof}

We have seen in \sec{jacobi-gauss} that we can use shifted Jacobi--Gauss quadrature to construct a distribution $p$ with support size $\ell=k+1$ that matches the moments of $\Beta(1,D-1)$ up to order $2\ell-1=2k+1$.
Then by \thm{induced-exact-design}, we can construct an exact unitary $k$-design whose distribution of the squared magnitude of the top-left entry is exactly $p$.
By \thm{induced-approximate-design}, we can adapt the shifted Jacobi--Gauss quadrature to construct an $\epsilon$-approximate unitary $k$-design.
This is formalized in the following corollary.

\begin{corollary}\label{cor:approximate-quadrature-design}
    For integers $D\geq 2$ and $k\geq 1$ and $\epsilon\in[0,1]$, let $0< x_1<x_2<\ldots<x_{k+1}< 1$ be the shifted Jacobi--Gauss quadrature nodes of $\Beta(1,D-1)$, and $w_1,\ldots,w_{k+1}$ be the corresponding weights.
    For a vector $\alpha=(\alpha_1,\ldots,\alpha_{k+1})$ satisfying $\max_{i\in[k+1]} |\alpha_i|\leq 1$ and $\sum_{i=1}^{k+1}\alpha_i w_i=0$, the unitary distribution induced by the discrete distribution with atoms $x_i$ and weights
    \begin{align}\label{eq:discrete-approximate-design}
        w_i' = w_i (1+\epsilon \alpha_i)
    \end{align}
    is an $\epsilon$-approximate unitary $k$-design.
\end{corollary}
\begin{proof}
    Let $p$ be the distribution specified by the shifted Jacobi--Gauss quadrature nodes and weights; it matches the moments of $\Beta(1,D-1)$ up to order $2k+1$ (\sec{jacobi-gauss}).
    By \thm{induced-exact-design}, the unitary distribution induced by $p$ is an exact unitary $k$-design.
    To apply \thm{induced-approximate-design}, we choose $h$ to satisfy $h(x_i)=\alpha_i$ and $h(x)=0$ for $x\notin\{x_1,\ldots,x_{k+1}\}$ for $\alpha$ satisfying the given conditions; then $(1+\epsilon h)p$ is exactly the discrete distribution with weights $w_i'$.
    The new weight $w_i'=w_i(1+\epsilon\alpha_i)$ induces an $\epsilon$-approximate unitary $k$-design.
\end{proof}

\cor{approximate-quadrature-design} implies that there is an $\epsilon$-approximate $k$-design whose top-left entry has a discrete distribution specified by the weights in \eq{discrete-approximate-design}.

\subsection{Worst-case moment matching distribution}\label{sec:worst-case}

In \sec{induced-measure}, we have reduced the problem of lower bounding the Kolmogorov distance of an entry in an approximate design from its Haar counterpart to bounding the distance between the beta distribution and the corresponding shifted Jacobi--Gauss quadrature.
In this section, we prove a distance lower bound.

We start with a lemma that lower bounds the distance between any continuous distribution on $[0,1]$ and any discrete distribution supported on $\ell$ points.
Though we only consider the case that both measures are supported on $[0,1]$, the lemma can be generalized to any continuous distribution and any discrete distribution supported on finitely many points in $\mathbb R$.
Since the distribution is discrete, across the jumps at quadrature nodes $x_i$, the beta distribution must differ from one side by at least $w_i/2$. 
Therefore, we can establish the lower bound by bounding the largest quadrature weight.

\begin{lemma}\label{lem:kolmogorov-distance-continuous-discrete}
    For an integer $\ell\geq 1$, let $P$ be a continuous distribution function and $Q$ be a discrete distribution function supported on $\ell$ points in $[0,1]$ with maximum weight $w_*$. 
    Then the Kolmogorov distance between $P$ and $Q$ is at least $\frac{1}{2}w_*$.
\end{lemma}
\begin{proof}
    Let $0\leq x_1<\ldots<x_\ell\leq 1$ be the support of $Q$, and let $w_i$ be the weight of $Q$ at $x_i$. 
    Let $j=\arg\max_i w_i$. 
    
    Since $P$ is continuous, for any $\eta>0$, there exists $\delta>0$ such that $|P(x_j+\delta)-P(x_j-\delta)|<\eta$. 
    By monotonicity of $Q$, since $Q(x_j+\delta)\geq S_j$ and $Q(x_j-\delta)\leq S_{j-1}$ for $S_j:=\sum_{i\leq j}w_i$, where we also take $\delta$ smaller than the distance from $x_j$ to its neighboring support points, we have $|Q(x_j+\delta)-Q(x_j-\delta)|\geq w_j$. 
    Let the Kolmogorov distance between $P$ and $Q$ be $\Delta$. 
    By definition, we have
    \begin{align}
        \Delta &\geq \max\left\{ |P(x_j+\delta)-Q(x_j+\delta)|, |P(x_j-\delta)-Q(x_j-\delta)|
        \right\} \\
        &\geq \frac{1}{2} \left( |P(x_j+\delta)-Q(x_j+\delta)| + |P(x_j-\delta)-Q(x_j-\delta)| \right) \\
        &\geq \frac{1}{2} \left( |P(x_j+\delta)-Q(x_j+\delta)| + |P(x_j+\delta)-Q(x_j-\delta)| - \eta \right) \\
        &\geq \frac{1}{2} \left( |Q(x_j+\delta)-Q(x_j-\delta)| - \eta \right).
    \end{align}
    Finally, taking $\eta\to 0$, the bound is at least $w_*/2$. 
\end{proof}

By \lem{kolmogorov-distance-continuous-discrete}, it suffices to give a lower bound on $w_*$. 
In the following theorem, we establish a matching lower bound to \thm{jacobi-weight-upper-bound}. 
As in \sec{jacobi-gauss}, $D$ there denotes the second parameter of $\Beta(1,D)$ rather than the dimension, and we apply it with $D-1$ in place of $D$ below.

\begin{theorem}\label{thm:quadrature-weights-lower-bound}
    For integers $D\geq 1$ and $\ell\geq 2$, let $0< x_1<x_2<\ldots<x_\ell< 1$ be the Gaussian quadrature points of $\Beta(1,D)$, and $w_1,\ldots,w_\ell$ be the corresponding weights. 
    Then it holds that 
    \begin{align}
        \max_{i\in[\ell]} w_i = \Omega\left(\sqrt{\frac{D}{\ell(\ell+D)}}\right).
    \end{align}
\end{theorem}
\begin{proof}
    For the lower bound, we first prove that in the interval $[0,1/D]$, there are at most $O\left(\sqrt{\frac{\ell(\ell+D)}{D}}\right)$ quadrature nodes. 
    From this we deduce that there exists a weight $w_i$ that is at least $\Omega\left(\sqrt{\frac{D}{\ell(\ell+D)}}\right)$.
    
    Let $\mu_D$ be the distribution of $\Beta(1,D)$; recall from \sec{jacobi-gauss} that the nodes $x_i$ are the roots of the orthonormal polynomial $q_\ell$ of degree $\ell$ and that the weights $w_i$ are positive with $\sum_i w_i=1$.
    We now lower bound $\max_i w_i$.
    Recall that the distribution function of $\Beta(1,D)$ is $F(x) = 1-(1-x)^D$.
    This means that $F(1/D)=1-(1-1/D)^D\geq 1-1/e=\Omega(1)$.
    Thus it suffices to show that $\{x_i\}\cap [0,1/D]$ has size at most $O\left(\sqrt{\frac{\ell(\ell+D)}{D}}\right)$.

    To see this, we invoke the known results on the spacing of the roots of Jacobi polynomials to prove that if $x_j\leq 1/D$, then $j=O\left(\sqrt{\frac{\ell(\ell+D)}{D}}\right)$.
    Let $z_1^{(\alpha,\beta)}>z_2^{(\alpha,\beta)}>\ldots >z_\ell^{(\alpha,\beta)}$ be the roots of the Jacobi polynomial $P_{\ell}^{(\alpha,\beta)}$ of degree $\ell$ with respect to the weight $(1-z)^\alpha (1+z)^\beta$.
    By the symmetry relation \eq{jacobi-symmetry} and \eq{beta-jacobi}, up to a nonzero normalization constant (which does not affect the roots), 
    \begin{align}
        q_\ell(x)\propto P_{\ell}^{(D-1,0)}(2x-1)=(-1)^\ell P_{\ell}^{(0,D-1)}(1-2x). 
    \end{align}
    Therefore we have 
    \begin{align}
        x_j = \frac{1}{2} \left(1 - z_j^{(0,D-1)}\right).
    \end{align}
    Now write each $z_j^{(\alpha,\beta)}=\cos\theta_{j}^{(\alpha,\beta)}$, and thus the $\theta_j^{(\alpha,\beta)}$ are ordered increasingly:
    \begin{align}
        0<\theta_1^{(\alpha,\beta)}<\ldots<\theta_\ell^{(\alpha,\beta)}<\pi.
    \end{align}
    By Bruns' inequality (\prop{jacobi-zeros}), we have
    \begin{align}\label{eq:theta-bound}
        \theta_j^{(0,0)} \geq \frac{(2j-1)\pi}{2\ell+1} \geq \frac{j\pi}{3\ell}. 
    \end{align}
    Then by the monotonicity of the zeros of Jacobi polynomials in \prop{jacobi-zeros}, with $A_\ell(\alpha,\beta)=2\ell^2 + 2\ell(\alpha+\beta+1)+(\alpha+1)(\beta+1)$, the quantity $A_\ell(\alpha,\beta)(1-z_j^{(\alpha,\beta)})$ is increasing in $\beta$.
    We have 
    \begin{align}
        1-z_j^{(0,D-1)} 
        &\geq \frac{A_\ell(0,0)}{A_\ell(0,D-1)} (1-z_j^{(0,0)}) \\
        &\geq \frac{A_\ell(0,0)}{A_\ell(0,D-1)} \cdot \frac{2j^2}{9\ell^2},
    \end{align}
    where the last inequality follows from \eq{theta-bound} and from the inequality that
    \begin{align}
        1-\cos x = 2\sin^2(x/2) \geq \frac{2x^2}{\pi^2}
    \end{align}
    for $x\in[0,\pi]$.
    We can further simplify the bound since 
    \begin{align}
        A_\ell(0,0) &= 2\ell^2+2\ell+1 \geq 2\ell^2, \\ 
        A_\ell(0,D-1) &= 2\ell^2+2\ell D +D \leq 2\ell^2 + 3\ell D\leq 3\ell(\ell+D).
    \end{align}
    These together imply that
    \begin{align}
        x_j \geq \frac{1}{2} \cdot \frac{2\ell}{3(\ell+D)} \cdot \frac{2j^2}{9\ell^2} = \frac{2j^2}{27\ell(\ell+D)}.
    \end{align}
    Therefore, if $x_j\leq 1/D$, then 
    \begin{align}
        j \leq \sqrt{\frac{27\ell(\ell+D)}{2D}}.
    \end{align}

    Suppose there are $T$ zeros in $[0,1/D]$.
    Consider the following two cases.
    \begin{enumerate}
    \item If $T=\ell$, then we have $x_{\ell}\leq 1/D$; since $\sum_i w_i=1$, we have $\ell \max_i w_i \geq 1$.
    This means that $\max_i w_i \geq 1/\ell=1/T$.
    
    \item If $T<\ell$, then we have $x_{T+1}>1/D$. 
    By monotonicity of $F$ and \lem{markov-stieltjes}, 
    \begin{align}
        F(1/D)\leq F(x_{T+1})\leq \sum_{i=1}^{T+1} w_i.
    \end{align} 
    This implies that there must be some $i\in[T+1]$ such that 
    \begin{align}
    w_i \geq \frac{F(1/D)}{T+1} \geq \frac{1-1/e}{T+1}.
    \end{align}
     
    \end{enumerate}
    Combining these two cases, since $T\leq \sqrt{\frac{27\ell(\ell+D)}{2D}}$, we have 
    \begin{align}
        \max_i w_i \geq \frac{1-1/e}{T+1} = \Omega\left(\sqrt{\frac{D}{\ell(\ell+D)}}\right).
    \end{align}
\end{proof}

Given the lower bound on the quadrature weights, we are ready to prove the existence of an $\epsilon$-approximate unitary $k$-design which approximates Porter--Thomas to error $\Omega\left(\sqrt{\frac{D}{k(D+k)}}+\epsilon\right)$ in Kolmogorov distance.

\begin{theorem}\label{thm:approximate-design-lower-bound}
    For integers $D\geq 2$ and $k\geq 1$ and $\epsilon\in[0,1]$, there exists an $\epsilon$-approximate unitary $k$-design on $\mathrm U(D)$ such that the squared magnitude of its top-left entry is at Kolmogorov distance
    $\Omega\left(\sqrt{\frac{D}{k(D+k)}}+\epsilon\right)$ from $\Beta(1,D-1)$.
\end{theorem}
\begin{proof}
    By \cor{approximate-quadrature-design}, we can construct an $\epsilon$-approximate unitary $k$-design from a shifted Jacobi--Gauss quadrature with weights $\{w_i\}_{i=1}^{\ell}$ for $k=\ell-1$ and some weight vector $\alpha$.
    By \thm{quadrature-weights-lower-bound} applied to $\Beta(1,D-1)$ with $\ell=k+1$ nodes, $\max_i w_i=\Omega\big(\sqrt{(D-1)/((k+1)(k+D))}\big)=\Omega\left(\sqrt{\frac{D}{k(D+k)}}\right)$, where the last step uses $D\geq 2$ and $k\geq 1$.
    It suffices to find $\alpha$ to justify the statement. 

    Let $w_*=\max_i w_i$ be the largest weight. 
    If $w_*\geq 1/6$, choose $\alpha_i=0$ for every $i$.
    The resulting ensemble is an exact design, so it is also an $\epsilon$-approximate design.
    By \lem{kolmogorov-distance-continuous-discrete}, the Kolmogorov distance is at least $w_*/2\geq 1/12\geq \frac{1}{24}\left(\sqrt{\frac{D}{k(D+k)}}+\epsilon\right)$, since both terms in the parentheses are at most $1$.

    Otherwise, we have $w_*<1/6$.
    Let $m$ be the largest index such that $S_m=\sum_{i=1}^m w_i \leq 1/3$. 
    Such an $m$ exists and satisfies $1\leq m\leq\ell-1$, since $S_1=w_1<1/6$ and $S_\ell=1$.
    Since every $w_i\leq w_*$, we have $S_m \in (1/3-w_*,1/3]$.
    We set $\alpha_i=1$ for $i\in[m]$ and $\alpha_{i}=s$ for $m<i\leq\ell$, where $s$ is to be determined now.
    Since $\sum_i \alpha_i w_i=0=S_m + s(1-S_m)$, we have 
    $s = \frac{S_m}{S_m-1}$.
    Since $0\leq S_m\leq 1/3$, we have $s\in[-1/2,0]$ and therefore $|s|\leq 1$.
    Now we have $\alpha$ satisfying $\max_i|\alpha_i|\leq 1$ and $w\cdot\alpha=0$. 

    Now we show that the new weight vector $w_i'=(1+\epsilon\alpha_i)w_i$ forms a distribution that has Kolmogorov distance $\Omega\left(\sqrt{\frac{D}{k(D+k)}}+\epsilon\right)$ from $\Beta(1,D-1)$. 
    Consider the event that $X\leq x_m$. 
    Let $P$ be the distribution function of $\Beta(1,D-1)$. By \lem{markov-stieltjes}, we have 
    $S_{m-1} \leq P(x_m) \leq S_m$.
    Let $Q$ be the distribution function of the discrete distribution with weights $w_i'$.
    The probability of the event $x\leq x_m$ under the discrete distribution with weights $w_i'$ is 
    \begin{align}
        Q(x_m) = \sum_{i=1}^m w_i' = \sum_{i=1}^m (1+\epsilon\alpha_i)w_i = (1+\epsilon)S_m.
    \end{align}
    The difference is at least 
    \begin{align}\label{eq:epsilon-lower-bound}
        Q(x_m) - P(x_m) \geq (1+\epsilon)S_m - S_m = \epsilon S_m \geq \epsilon\left(\frac{1}{3}-w_*\right).
    \end{align}
    Since $w_*<1/6$, the right-hand side is at least $\epsilon/6$. 
    Therefore, the Kolmogorov distance is at least $\Omega(\epsilon)$. 
    
    It remains to show that the Kolmogorov distance is at least $\Omega\left(\sqrt{\frac{D}{k(D+k)}}\right)$.
    By \lem{kolmogorov-distance-continuous-discrete}, the distance is at least $\frac{1}{2}\max_i w_i'\geq \frac{1}{2}w_*(1-\epsilon/2)\geq w_*/4=\Omega\left(\sqrt{\frac{D}{k(D+k)}}\right)$ for $\epsilon\leq 1$. 
    Combining this and the lower bound $\Omega(\epsilon)$ in \eq{epsilon-lower-bound}, we have shown that the Kolmogorov distance is at least $\Omega\left(\sqrt{\frac{D}{k(D+k)}}+\epsilon\right)$.
\end{proof}

\thm{approximate-design-lower-bound} shows that there exists an $\epsilon$-approximate unitary $k$-design whose distribution of the squared magnitude of the top-left entry is $\Omega\left(\sqrt{\frac{D}{k(D+k)}}+\epsilon\right)$ away from Porter--Thomas in Kolmogorov distance.
This result implies that the bound in \thm{approximate-design-upper-bound} is tight up to a constant factor.

\section{Convergence in total variation from local invariance}\label{sec:local-translation-invariant}

In \sec{convergence-upper-bound} and \sec{convergence-lower-bound}, we have established a complete analysis of the Kolmogorov distance between the squared magnitude of a single entry of a generic approximate design and that of a Haar random unitary of the same size. 
In particular, by the lower bound in \thm{approximate-design-lower-bound}, the best Kolmogorov distance that can be guaranteed uniformly over generic designs is $\Omega(k^{-1/2})$ when $k<2^n$ and $\Omega(2^{n/2}/k)$ when $k\geq 2^n$.
Therefore, without imposing further assumptions on the design, no uniform guarantee can improve on the rate $1/k$ up to a constant.

In this section, we explore which property implies fast convergence. 
To circumvent the slow rate, we introduce a notion that we call a locally invariant design. 
In the following analysis, we use the notation $\mu(k,\epsilon)$ for the distribution of a design and $\mu$ for the Haar measure on unitaries of the same size. 

\begin{definition}[Locally invariant design]\label{dfn:locally-invariant}
    An $n$-qubit approximate design $\mu(k,\epsilon)$ is $m$-locally invariant if there exists a subset $S\subseteq[n]$ of $m$ qubits such that applying any fixed $m$-qubit unitary $V$ on the qubits in $S$ after the circuit does not change the distribution, i.e., for $U\sim\mu(k,\epsilon)$, $(V\otimes I)U\sim\mu(k,\epsilon)$ for any fixed $V$ acting on the qubits in $S$ (here and below $V\otimes I$ denotes $V$ acting on the qubits in $S$ and the identity on the remaining qubits).
\end{definition}

The following lemma explains how local invariance enters our analysis: it factorizes the output probability of a single basis state into the probability of a block of $a=2^m$ basis states, which is controlled by the design condition, times an independent beta random variable supplied by a Haar random local unitary.

\begin{lemma}[Block factorization]\label{lem:block-factorization}
    Let $1\leq m<n$, let $\mu$ be an $m$-locally invariant distribution on $\mathrm U(2^n)$ with invariant qubit set $S$, let $a=2^m$ and $N=2^n$, and fix $x,y\in\bit^n$.
    Write $x=(x_S,x_{\bar S})$ and let $B(x)=\{(s,x_{\bar S}):s\in\bit^m\}$ be the block of $a$ basis states that agree with $x$ outside $S$.
    For $U\sim\mu$, let
    \begin{align}
        X_k := \sum_{r\in B(x)} |\bra{r}U\ket{y}|^2
    \end{align}
    be the probability of the block $B(x)$.
    Then $|\bra{x}U\ket{y}|^2$ has the same distribution as $X_kY$, where $Y\sim\Beta(1,a-1)$ is independent of $X_k$.
    If $\mu$ is the Haar measure, then $X_k\sim\Beta(a,N-a)$; consequently, $XY\sim\Beta(1,N-1)$ for independent $X\sim\Beta(a,N-a)$ and $Y\sim\Beta(1,a-1)$ (this is also the product rule \eq{beta-product}, which holds for all integers $2\leq a<N$).
\end{lemma}
\begin{proof}
    Let $V\sim\Haar(a)$ act on the qubits in $S$, independently of $U$.
    By local invariance, $(V\otimes I)U\sim\mu$, so $|\bra{x}U\ket{y}|^2$ has the same distribution as $|\bra{x}(V\otimes I)U\ket{y}|^2$.
    Expanding in the computational basis of the qubits in $S$,
    \begin{align}
        \bra{x}(V\otimes I)U\ket{y} = \sum_{s\in\bit^m}\bra{x_S}V\ket{s}\bra{s,x_{\bar S}}U\ket{y} = \bra{x_S}V\ket{\psi_U},
        \qquad \ket{\psi_U}:=\sum_{s\in\bit^m}\bra{s,x_{\bar S}}U\ket{y}\ket{s}\in\mathbb C^{a},
    \end{align}
    and $\|\psi_U\|^2=X_k$.
    Conditioned on $U$ with $\psi_U\neq 0$, the vector $V\ket{\psi_U}/\|\psi_U\|$ is uniformly distributed on the unit sphere of $\mathbb C^a$ by the invariance of the Haar measure, so $Y:=|\bra{x_S}V\ket{\psi_U}|^2/\|\psi_U\|^2$ has the law $\Beta(1,a-1)$ (\sec{preliminaries-beta}) regardless of $U$ (when $\psi_U=0$ we let $Y$ be an independent $\Beta(1,a-1)$ variable).
    Hence $Y$ is independent of $U$, in particular of $X_k$, and $|\bra{x}(V\otimes I)U\ket{y}|^2=X_kY$.
    Note also that the block probability is unchanged by $V$, i.e., $\sum_{r\in B(x)}|\bra{r}(V\otimes I)U\ket{y}|^2=X_k$, since $V$ acts within the block.

    If $U$ is Haar random, the column $U\ket{y}$ is uniformly distributed on the unit sphere of $\mathbb C^N$, so the vector $(|\bra{r}U\ket{y}|^2)_{r\in\bit^n}$ has the Dirichlet law $\Dir(1^N)$ (\sec{preliminaries-beta}), and the sum of $a$ of its coordinates has the law $\Beta(a,N-a)$ by the aggregation property of the Dirichlet distribution.
    The last claim follows by combining the two statements, since the Haar measure is $m$-locally invariant and $|\bra{x}U\ket{y}|^2\sim\Beta(1,N-1)$ for Haar random $U$.
\end{proof}

By \thm{polynomial-relative-error} (applied with the block $S=B(x)$ and the column $y$), the design condition controls the expectations of squared polynomials of degree at most $k$ in $X_k$ with relative error $\epsilon$; this is the only property of $X_k$ used in \sec{locally-invariant-ks} and \sec{locally-invariant-tv}.
When $m$ is small enough, the Haar random unitary can be approximated by a quantum circuit of small gate complexity. 
Then in \sec{beta-recurrence}, we show that the central moments of beta random variables can be bounded. 
In particular, for $\Beta(a,N-a)$, the $k$th central moment decays exponentially in $k$ if $a$ grows linearly in $k$ as well. 
In \sec{locally-invariant-ks}, we use this fact to prove the characteristic function can be well approximated by an expected polynomial, which in turn allows us to use \thm{polynomial-relative-error} to bound closeness. 
It is worth noting that local invariance makes the marginal absolutely continuous away from zero: a possible atom at zero remains, but its mass is exponentially small under the assumptions of \thm{convergence-distribution}, and the muting argument in \sec{locally-invariant-ks} controls this exceptional mass before Fourier inversion is applied. The barrier that a discrete marginal cannot converge in total variation therefore no longer exists.
In \sec{locally-invariant-tv}, we further strengthen the result to closeness in total variation distance. 

\subsection{Locally invariant approximate designs}\label{sec:locally-invariant-properties}

We next show that one can use an $m$-qubit Haar random unitary to convert a generic design to a locally invariant one. 

\begin{theorem}\label{thm:locally-invariant-design-haar}
    Let $U\sim\mu(k,\epsilon)$ be an $\epsilon$-approximate unitary $k$-design on $n$ qubits, and let $V$ be a Haar random unitary acting on any fixed subset $S$ of $m$ qubits, independent of $U$. Then $(V\otimes I)U$ is an $m$-locally invariant $\epsilon$-approximate unitary $k$-design.
\end{theorem}
\begin{proof}
    Write $\tilde V=V\otimes I$, and let $\mu$ be the Haar measure on $\mathrm U(2^n)$.
    For any fixed unitary $V_0$ on the qubits in $S$, the left invariance of the Haar measure on $\mathrm U(2^m)$ gives $(V_0\otimes I)\tilde V\sim\tilde V$, and hence $(V_0\otimes I)\tilde VU\sim\tilde VU$ by independence; thus $\tilde VU$ is $m$-locally invariant.
    It remains to show that $\tilde VU$ is an approximate design with the same parameters.
    Let $R(W):=W^{\otimes k}\otimes (W^*)^{\otimes k}$, so that $R(\tilde VU)=R(\tilde V)R(U)$, and define the completely positive map $\Phi(\rho):=\Exp_V[R(\tilde V)\rho R(\tilde V)^\dag]$.
    By the independence of $U$ and $V$, $\Delta^{(k,k)}_{\tilde VU}=\Phi\circ\Delta^{(k,k)}_{\mu(k,\epsilon)}$, and by the left invariance of the Haar measure on $\mathrm U(2^n)$, $\Phi\circ\Delta^{(k,k)}_{\mu}=\Delta^{(k,k)}_{\mu}$, where $\Delta^{(k,k)}_{\nu}$ is defined as in \dfn{strong-design}.
    Since $\Delta^{(k,k)}_{\mu(k,\epsilon)}-(1-\epsilon)\Delta^{(k,k)}_{\mu}$ and $(1+\epsilon)\Delta^{(k,k)}_{\mu}-\Delta^{(k,k)}_{\mu(k,\epsilon)}$ are completely positive by assumption, and the composition of completely positive maps is completely positive,
    \begin{align}
        \Delta^{(k,k)}_{\tilde VU}-(1-\epsilon)\Delta^{(k,k)}_{\mu} &= \Phi\circ\big(\Delta^{(k,k)}_{\mu(k,\epsilon)}-(1-\epsilon)\Delta^{(k,k)}_{\mu}\big)
        \quad\text{and} \nonumber\\
        (1+\epsilon)\Delta^{(k,k)}_{\mu}-\Delta^{(k,k)}_{\tilde VU} &= \Phi\circ\big((1+\epsilon)\Delta^{(k,k)}_{\mu}-\Delta^{(k,k)}_{\mu(k,\epsilon)}\big)
    \end{align}
    are completely positive.
    This justifies that $\tilde VU$ is an $\epsilon$-approximate $k$-design.
\end{proof}

\thm{locally-invariant-design-haar} implies that appending an $m$-qubit Haar random unitary turns a generic design into an $m$-locally invariant one. 
For the generation of such a design to be efficient (in the sense that the gate complexity scales no faster than a polynomial in the number of qubits $n$), 
it would be sufficient to have $m=O(\log n)$. 

Indeed, by \cite[Section~5]{mezzadriHowGenerateRandom2007}, one can sample a Haar random unitary of dimension $d$ by sampling a complex-valued Gaussian matrix using $2d^2$ standard Gaussian samples (for the real and the imaginary part of the entries). After performing a QR decomposition $Z=QR$ of the Gaussian matrix $Z$ and correcting the phases, the matrix $Q'=Q\cdot\diag(r_{11}/|r_{11}|,\ldots,r_{dd}/|r_{dd}|)$ is a Haar random unitary (the uncorrected factor $Q$ is not Haar distributed). 
Given the Haar random matrix, one can perform a standard exact synthesis (as in \cite{shendeSynthesisQuantumlogicCircuits2006}) of any $d\times d$ unitary into a sequence of $O(d^2)$ single-qubit gates and CNOT gates. 
Each single-qubit gate can then be approximated using the algorithm following from the Solovay--Kitaev theorem \cite{dawsonSolovayKitaevAlgorithm2005}. 
To obtain a circuit $C$ satisfying $\|C-V\|_{\op}\leq\epsilon$ (up to a global phase, which does not affect the quantities considered below), it suffices to approximate each of the $M=O(d^2)$ single-qubit gates to precision $\epsilon/M$ in operator norm, since the errors add up by a telescoping argument.
By the Solovay--Kitaev theorem, each single-qubit gate can be approximated using $O(\log^{3.97}(M/\epsilon))$ gates, and thus the total gate complexity is $d^2\log^{O(1)}(d^2/\epsilon)$.

Having a random quantum circuit $C$ that approximates a Haar random unitary matrix to error $\epsilon$ in operator norm, we can show that the squared magnitude of a linear combination of the entries is close in Kolmogorov distance.

For our use, we note that for any fixed $U$, $\bra{x}(V\otimes I)U\ket{y}$ can be written as $\bra{\phi'}V\ket{\psi'}$ for subnormalized vectors $\ket{\phi'}$ and $\ket{\psi'}$. 
Since scaling applied to $\ket{\phi'}$ and $\ket{\psi'}$ does not change the Kolmogorov distance, in the following theorem, we bound the Kolmogorov distance for normalized vectors, and the same bound directly applies to subnormalized ones.

When $U$ is drawn from some distribution, we can think of $\bra{x}(V\otimes I)U\ket{y}$ as a mixture of random variables specified by $U$. 
Since the same bound holds for every $U$, by an averaging argument, the following bound also holds when $U$ is either sampled from a design or a Haar random unitary.

\begin{theorem}\label{thm:solovay-kitaev-ks}
    For $d\geq 2$, let $V\sim\Haar(d)$ and let $C$ be a $d\times d$ random matrix (possibly dependent on $V$) such that
    \begin{align}
        \Pr\left[\|C-V\|_{\op}\leq\epsilon\right] \geq 1-\delta.
    \end{align}
    Then for the function $f(U):=|\bra{\phi}U\ket{\psi}|^2$ with any fixed unit vectors $\ket{\psi},\ket{\phi}$,
    \begin{align}
        \sup_{t\in\mathbb R}\left| F_{f(C)}(t) - F_{f(V)}(t)\right| \leq \delta + 2\sqrt{d-1}\epsilon,
    \end{align}
    where $F_X$ is the distribution function of $X$. 
\end{theorem}
\begin{proof}
    For $V\sim\Haar(d)$, $f(V)\sim\Beta(1,d-1)$, since $V\ket{\psi}$ is uniformly distributed on the unit sphere (\sec{preliminaries-beta}).
    Let $g(U):=|\bra{\phi}U\ket{\psi}|$.
    If $\|C-V\|_{\op}\leq\epsilon$, then  $|g(C)-g(V)|\leq\epsilon$. 
    Since $g(V)^2\sim\Beta(1,d-1)$, by change of variable, the probability density function of $g(V)$ can be written as
    \begin{align}\label{eq:density-gv}
        f_{g(V)}(t) = 2(d-1)t(1-t^2)^{d-2},
    \end{align}
    and therefore, since this density is maximized at $t^2=1/(2d-3)$ with maximum value $2(d-1)(2d-3)^{-1/2}(1-1/(2d-3))^{d-2}\leq 2(d-1)/\sqrt{2d-3}\leq 2\sqrt{d-1}$ for $d\geq 2$, 
    \begin{align}\label{eq:max-gv}
        \max_{t\in\mathbb R}f_{g(V)}(t)\leq 2\sqrt{d-1}.
    \end{align}
    
    Now let $E$ be the event that $|g(C)-g(V)|\leq\epsilon$.
    We have $\Pr[\bar E]\leq\delta$.
    Conditioned on $E$ happening, we have 
    $g(V)-\epsilon\leq g(C)\leq g(V)+\epsilon$. 
    Then this implies that when $g(C)\leq t$ happens, it must be true that $g(V)\leq t+\epsilon$.
    This implies that
    \begin{align}
        \Pr[g(C)\leq t\mid E] \leq \Pr[g(V)\leq t+\epsilon \mid E].
    \end{align}
    Since $\Pr[E]\geq 1-\delta$, we have
    \begin{align}\label{eq:cv}
        \Pr[g(C)\leq t] \leq \Pr[g(V) \leq t+\epsilon] + \delta.
    \end{align}
    For the lower bound, conditioned on $E$ happening, we have $g(C)-\epsilon\leq g(V) \leq g(C)+\epsilon$, and thus we can interchange $C$ and $V$ in \eq{cv} and replace $t$ with $t-\epsilon$ to get
    \begin{align}
        \Pr[g(V)\leq t-\epsilon] \leq \Pr[g(C) \leq t] + \delta.
    \end{align}
    Now we have 
    \begin{align}
        F_{g(V)}(t-\epsilon) - \delta \leq F_{g(C)}(t) \leq F_{g(V)}(t+\epsilon) + \delta. 
    \end{align}
    Now since $F_{g(V)}$ is $2\sqrt{d-1}$-Lipschitz by \eq{density-gv} and \eq{max-gv}, we have 
    \begin{align}
        F_{g(C)}(t) 
        &\leq F_{g(V)}(t) + 2\epsilon\sqrt{d-1} + \delta \\
        F_{g(C)}(t) 
        &\geq F_{g(V)}(t) - 2\epsilon\sqrt{d-1} - \delta,
    \end{align}
    and thus the Kolmogorov distance between $g(C)$ and $g(V)$ is at most $\delta+2\epsilon\sqrt{d-1}$.
    Since $g^2=f$ and the mapping $x\mapsto x^2$ is strictly increasing on $[0,\infty)$, we have that $\Pr[g(U)\leq \sqrt t]=\Pr[f(U)\leq t]$ holds for every $t\in\mathbb R^{\geq 0}$.
    This implies that the Kolmogorov distance between $f(C)$ and $f(V)$ is at most $\delta+2\epsilon\sqrt{d-1}$.  
\end{proof}

In \thm{solovay-kitaev-ks}, we have seen that as long as the Haar random unitary in use has small dimension, we can control the error in Kolmogorov distance induced by replacing the Haar random unitary with a quantum circuit.

\begin{remark}[Finite precision]\label{rmk:finite-precision}
    The sampling procedure described above was stated with exact Gaussian samples and exact arithmetic; we briefly indicate why a finite-precision implementation satisfies the hypothesis of \thm{solovay-kitaev-ks}.
    To reach Kolmogorov error $\zeta$, it suffices to take $\delta\leq\zeta/2$ and an operator-norm tolerance $\epsilon\leq\zeta/(4\sqrt{d-1})$, which changes the gate count $d^2\log^{O(1)}(d^2/\epsilon)$ above only inside the polylogarithm.
    The Gaussian entries can be sampled to $O(\log(d/\zeta))$ bits of precision, and the QR factorization, the phase correction, and the synthesis can be carried out in floating-point arithmetic with $O(\log(d/\zeta))$ bits.
    The smallest singular value of a $d\times d$ matrix $Z$ with independent standard complex Gaussian entries satisfies $\Pr[\sigma_{\min}(Z)\leq t/\sqrt d]\leq t^2$ for every $t>0$ \cite{edelmanEigenvaluesConditionNumbers1988}, so outside an event of probability at most $\delta/2$ the matrix $Z$ is well conditioned, $\|Z^{-1}\|_{\op}\leq\sqrt{2d/\delta}$, and the standard perturbation bounds for the QR factorization \cite{highamAccuracyStabilityNumerical2002} show that the unitary computed from the rounded data is within operator norm $\epsilon$ of $Q'$.
    The resulting classical sampler runs in time $\poly(d,\log(1/\zeta))$ and produces, together with the synthesis, a circuit $C$ coupled to $V=Q'$ with $\Pr[\|C-V\|_{\op}>\epsilon]\leq\delta$; we do not attempt to optimize the precision.
\end{remark}
Thus in the following subsections of \sec{local-translation-invariant}, we analyze the convergence of the squared magnitude of a single entry in a locally invariant design to that in a Haar random unitary. 
However, for a discrete unitary design, replacing the appended Haar random unitary with the efficiently sampled one results in a discrete measure. 
The replacement loses convergence in total variation in general.

\subsection{A recurrence relation of beta random variables}\label{sec:beta-recurrence}

We prove that the moments of a beta random variable admit a recursive form, which will be useful for our analysis in \sec{locally-invariant-ks} and \sec{locally-invariant-tv}. 

\begin{lemma}\label{lem:recurrence}
    Let $X\sim\Beta(a,N-a)$ and $Z=(X-\mu)/\mu$ for $\mu=\Exp[X]$. Let $\nu_k$ be the $k$-th moment of $Z$. Then for integer $k\geq 2$, $\nu_k$ satisfies the recurrence relation
    \begin{align}
        \nu_k = \frac{k-1}{N+k-1} \left[ \left(\frac{N}{a} - 2\right) \nu_{k-1} + \left(\frac{N}{a} - 1\right) \nu_{k-2} \right].
    \end{align}
\end{lemma}
\begin{proof}
    To prove the lemma, we apply an integration by parts to express $\nu_k$ in terms of $\nu_{k-1}$ and $\nu_{k-2}$.
    Recall that the pdf of $X$ is
    \begin{align}
        f_X(x) = \frac{\Gamma(N)}{\Gamma(a)\Gamma(N-a)}x^{a-1}(1-x)^{N-a-1}.
    \end{align}
    The natural logarithm of $f_X(x)$ is
    \begin{align}
        \ln f_X(x) = \ln \frac{1}{B(a,N-a)} + (a-1)\ln x + (N-a-1)\ln(1-x).
    \end{align}
    The derivative of $\ln f_X(x)$ is
    \begin{align}
        \frac{\d}{\d x}\ln f_X(x) 
        &= \frac{f_X'(x)}{f_X(x)} = \frac{a-1}{x} - \frac{N-a-1}{1-x}.
    \end{align}
    This implies that $x(1-x)\cdot f_X'(x) = (a-1-(N-2)x)\cdot f_X(x)$.
    Letting $x=\mu(z+1)$, we have 
    \begin{align}
        x(1-x) f_X'(x) 
        &= \Big(a-1-(N-2)\mu (z+1) \Big)\cdot f_X(x) \\
        &= \left[\left(\frac{2a}{N}-1 \right) + \left(\frac{2a}{N}-a\right)z \right] f_X(x).
    \end{align}
    Multiplying both sides by $z^{k-1}$ and integrating with respect to $x$ over $[0,1]$, by an integration by parts (the boundary term vanishes since $x(1-x)f_X(x)\propto x^a(1-x)^{N-a}$ vanishes at $x\in\{0,1\}$), we have
    \begin{align}
        \left(\frac{2a}{N}-a\right) \nu_{k} + \left(\frac{2a}{N}-1 \right)\nu_{k-1}
        &= \int_0^1 z^{k-1} x(1-x) f_X'(x) \d x \\
        &= -\int_0^1 \frac{\d}{\d x}\Big(z^{k-1} x(1-x)\Big) f_X(x) \d x \\
        &= -\int_0^1 \frac{\d}{\d z}\Big(z^{k-1} (-\mu z^2 + (1-2\mu)z + (1-\mu))\Big) f_X(x) \d x \\
        &= \mu(k+1) \nu_{k} - k (1-2\mu) \nu_{k-1} - (k-1) (1-\mu) \nu_{k-2}.
    \end{align}
    Substituting $a/N$ by $\mu$ and $a$ by $N\mu$, we have 
    \begin{align}
        \mu\Big[N+k-1\Big]\nu_{k} = \Big[(k-1)(1-2\mu)\Big] \nu_{k-1} + (k-1)(1-\mu)\nu_{k-2}.
    \end{align}
    This can be simplified to 
    \begin{align}
        \nu_{k} 
        &= \frac{k-1}{N+k-1} \left[ \frac{1-2\mu}{\mu} \nu_{k-1} + \frac{1-\mu}{\mu} \nu_{k-2} \right] \\
        &= \frac{k-1}{N+k-1} \left[ \left(\frac{N}{a}-2\right) \nu_{k-1} + \left(\frac{N}{a}-1\right) \nu_{k-2} \right].
    \end{align}
\end{proof}

Then we use the recurrence relation in \lem{recurrence} to show that the moments of a beta random variable can be bounded if the first parameter $a$ scales linearly in the order $k$ of the moment.

\begin{lemma}\label{lem:central-moment}
    Let $X\sim\Beta(a,N-a)$ and $Z=(X-\mu)/\mu$ for $\mu=\Exp[X]$. Then for integers $k\geq 1$, $a\geq \pi k$, and $N\geq 2a$, the $k$-th moment of $Z$ satisfies the bound
    \begin{align}
        \Exp_{Z}\left[Z^k\right] \leq \frac{1}{2}\cdot \frac{(k-1)!!}{(a/2)^{k/2}}.
    \end{align}
\end{lemma}
\begin{proof}
    We prove the lemma by induction on $k$ for $k\leq a/\pi$.
    Let $\nu_k:=\Exp_Z[Z^k]$.
    For $k=1$, $\nu_1 = 0$.
    For $k=2$, $\nu_2 = \frac{a(N-a)}{N^2(N+1)}\cdot \frac{N^2}{a^2} \leq \frac{1}{a}$.
    These two cases establish the bound for $k\in\{1,2\}$ (recall that $0!!=1!!=1$).
    For $k\geq 3$, we use the recurrence relation in \lem{recurrence}: since $N\geq 2a$, both coefficients $\frac{N}{a}-1$ and $\frac{N}{a}-2$ are nonnegative, so starting from $\nu_0=1$ and $\nu_1=0$, all moments $\nu_j$ are nonnegative by induction, and
    \begin{align}
        \nu_k
        &=\frac{k-1}{N+k-1} \left[ \left(\frac{N}{a} - 1\right) \nu_{k-2} + \left(\frac{N}{a} - 2\right) \nu_{k-1} \right]
        \leq \frac{k-1}{a} \left[ \nu_{k-2} + \nu_{k-1} \right].
    \end{align}
    Suppose that the bound holds for $k-1$ and $k-2$; then we have
    \begin{align}
        \nu_k 
        &\leq \frac{1}{2}\frac{k-1}{a} \left[ \frac{(k-3)!!}{(a/2)^{(k-2)/2}} + \frac{(k-2)!!}{(a/2)^{(k-1)/2}} \right] 
        \leq \frac{1}{4}\frac{(k-1)!!}{(a/2)^{k/2}} \left[ 1 + \frac{(k-2)!!}{(k-3)!!}\cdot \left(\frac{2}{a}\right)^{1/2} \right] \\
        &\leq \frac{1}{4}\frac{(k-1)!!}{(a/2)^{k/2}} \left[ 1 + \left(\frac{\pi k}{a}\right)^{1/2} \right] 
        \leq \frac{1}{2} \frac{(k-1)!!}{(a/2)^{k/2}}
    \end{align}
    for $k\leq a/\pi$, where the third inequality uses $\frac{(k-2)!!}{(k-3)!!}\leq\sqrt{\pi(k-1)/2}\leq\sqrt{\pi k/2}$ from \prop{double-factorial-ratio}. 
\end{proof}

As a corollary, the absolute moments can also be bounded.

\begin{corollary}\label{cor:absolute-moment}
    Let $X\sim\Beta(a,N-a)$ and $Z=(X-\mu)/\mu$ for $\mu=\Exp[X]$. Then for integers $k\geq 1$, $a\geq 2\pi k$, and $N\geq 2a$, the $k$-th absolute moment of $Z$ satisfies the bound
    \begin{align}
        \Exp_{Z}\left[|Z|^k\right] 
        \leq \left(\frac{4k}{ea}\right)^{k/2}.
    \end{align}
\end{corollary}
\begin{proof}
    By the Cauchy--Schwarz inequality, \lem{central-moment} applied with $2k$ in place of $k$ (which requires $a\geq 2\pi k$), and \prop{double-factorial}, we have 
    \begin{align}
        \Exp[|Z|^k]^2 \leq \Exp[Z^{2k}] \leq \frac{1}{2}\cdot \frac{(2k-1)!!}{(a/2)^{k}} \leq \left(\frac{4k}{ea}\right)^{k}.
    \end{align}
\end{proof}

The constants in the conditions $a\geq\pi k$ and $a\geq 2\pi k$ can likely be improved, though we do not attempt to optimize them here since they are not crucial for our purpose.

\subsection{Closeness in distribution with locally invariant design}\label{sec:locally-invariant-ks}

To bound the closeness between the distribution functions, we may apply the Berry--Esseen smoothing lemma (\lem{berry-esseen}). This requires us to bound the closeness between the characteristic functions, i.e., between the expectations of $e^{\i t\ln X}$ under the two distributions. 
Since we have only closeness in moments, we can at most bound the closeness for low-degree polynomials. 
However, a direct truncation applied to the exponential function would give a truncation error proportional to the 
$k$th moments multiplied by $|t|^k/k!$.

Following the argument in \lem{central-moment}, 
the $k$th central moment of $X\sim\Beta(a,N-a)$, normalized by its mean $\mu$, decays exponentially in $k$ for $a=\Omega(k)$.

Applying the Berry--Esseen smoothing lemma (\lem{berry-esseen}) directly would require us to choose a window size $T$ that is at most $O(\sqrt{ak}/\mu)$, where $\mu=a/N$ is the mean of $X\sim\Beta(a,N-a)$, which only gives a polynomially decaying bound in $k$.
In this case, the smoothing bound scales like $O(\sqrt{a}/(\mu T))=O(1/\sqrt k)$, thereby giving a polynomially decaying bound in $k$ for the closeness between the distribution functions.

To improve the bound to exponentially decaying in $k$, we rely on the fact that for independent $X\sim\Beta(a,N-a)$ and $Y\sim\Beta(1,a-1)$, we have $XY\sim\Beta(1,N-1)$ (the product rule \eq{beta-product}). 
The idea is to replace $X$ with the block probability $X_k$ of an $\epsilon$-approximate $k$-design (\lem{block-factorization}), and then bound the closeness between the distribution functions of $X_kY$ and $XY$. 
Taking $a$ to be a power of two, the random variable $Y$ is produced by a Haar random unitary acting on only $\log a$ qubits. 
We work with this idealized $Y$; its replacement by an efficiently sampled circuit was discussed in \sec{locally-invariant-properties}. 

Throughout \sec{locally-invariant-ks} and \sec{locally-invariant-tv}, $X\sim\Beta(a,N-a)$ denotes the Haar block probability, $X_k$ the probability of the same block for $U\sim\mu(k,\epsilon)$ (the dependence on $\epsilon$ is suppressed in the notation), and $Y\sim\Beta(1,a-1)$ is independent of both.
As explained at the beginning of \sec{preliminaries}, distribution functions and distances written with random variables refer to their laws.
When $\mu(k,\epsilon)$ is $m$-locally invariant and the block is $B(x)$, \lem{block-factorization} gives $X_kY\sim\nu(k,\epsilon)$ and $XY\sim\nu$, so the results below are statements about $d_{\mathrm K}(\nu(k,\epsilon),\nu)$ and $d_{\mathrm{TV}}(\nu(k,\epsilon),\nu)$; see \cor{locally-invariant-entry}.

\begin{proposition}[Beta integral]\label{prop:beta-integral}
        For integer $m\geq 2$, $\int_{-\infty}^\infty \left(\frac{1}{t^2+1}\right)^{m/2} \d t = \frac{\sqrt\pi \Gamma(\frac{m-1}{2})}{\Gamma(\frac{m}{2})}$ (substitute $t=\tan\theta$ to obtain a beta integral).
    \end{proposition}

We will need to bound the integral of $\frac{1}{|t|}\big|\binom{\i t}{\ell}\binom{\i t+a}{a}^{-1}\big|$ for $a\geq\ell+2$:
    \begin{lemma}\label{lem:binomial-integral}
        For positive integers $a\geq \ell+2$ and $\ell\geq 1$, we have 
        \begin{align}
            \int_{-\infty}^\infty \frac{1}{|t|}\left|\binom{\i t}{\ell}\binom{\i t+a}{a}^{-1}\right| \d t \leq \frac{a}{\ell}\frac{\sqrt{\pi}\Gamma((a-\ell)/2)}{\Gamma((a-\ell+1)/2)}.
        \end{align}
    \end{lemma}
    \begin{proof}
        Using \prop{beta-integral} with $m=a-\ell+1\geq 3$ in the last step, we have 
        \begin{align}
            \int_{-\infty}^\infty \frac{1}{|t|}\left|\binom{\i t}{\ell}\binom{\i t+a}{a}^{-1}\right| \d t 
            &= \frac{a!}{\ell!}\int_{-\infty}^\infty \frac{|\i t(\i t-1)\ldots (\i t-\ell+1)|}{|\i t(\i t+1)\ldots (\i t+a)|} \d t \\
            &= \frac{1}{\ell}\int_{-\infty}^\infty \left(\prod_{j=\ell}^{a} \left(\frac{t^2}{j^2} + 1\right)\right)^{-1/2} \d t \\
            &\leq \frac{1}{\ell}\int_{-\infty}^\infty \left(\frac{t^2}{a^2} + 1\right)^{-(a-\ell+1)/2} \d t \\
            &= \frac{a}{\ell}\int_{-\infty}^\infty \left(t^2 + 1\right)^{-(a-\ell+1)/2} \d t \\
            &= \frac{a}{\ell}\frac{\sqrt{\pi}\Gamma((a-\ell)/2)}{\Gamma((a-\ell+1)/2)} = O\left(\frac{a}{\ell(a-\ell)^{1/2}}\right).
        \end{align}
    \end{proof}
    When $a\geq C\ell$ for some constant $C>1$, this bound can be further simplified since
    \begin{align}
        \frac{a}{\ell(a-\ell)^{1/2}} \leq \frac{\sqrt a}{\ell (1-1/C)^{1/2}} = O\left(\frac{\sqrt a}{\ell}\right).
    \end{align}

\begin{theorem}\label{thm:convergence-distribution}
    Let $k$ be a positive integer, let $a\geq 2\pi k$ and $N>2a$ be integers, and let $\epsilon\in[0,1]$.
    Let $X_k=\sum_{r\in S}|\bra{r}U\ket{y}|^2$ be the probability of a fixed block $S$ of $a$ basis states in a fixed column $y$ for $U$ sampled from an $\epsilon$-approximate $k$-design on $\mathrm U(N)$ (e.g., the block probability of \lem{block-factorization}), let $X\sim\Beta(a,N-a)$ be its Haar counterpart (\lem{block-factorization}), and let $Y\sim\Beta(1,a-1)$ be independent of $X_k$ and $X$.
    Then
    \begin{align}
        d_{\mathrm K}(X_kY,XY)=\sup_{x\in\R}|F_{X_kY}(x) - F_{XY}(x)| = 2^{-\Omega(k)} + O(\epsilon).
    \end{align}
\end{theorem}
\begin{proof}
    Let $Z=NX/a-1$ be the shifted and scaled version of $X$ and $Z_k=NX_k/a-1$ be the shifted and scaled version of $X_k$.
    Then we have $XY = \frac{a}{N}(Z+1)Y$ and $X_kY = \frac{a}{N}(Z_k+1)Y$.
    We bound the closeness between the CDFs of the logarithms of $(Z+1)$ and $(Z_k+1)$ first, since 
    \begin{align}
        |\varphi_{\ln XY}(t) - \varphi_{\ln X_kY}(t)|
        &= |\varphi_{\ln (Z+1)}(t) - \varphi_{\ln (Z_k+1)}(t)| \cdot |\varphi_{\ln Y}(t)| \cdot |(a/N)^{\i t}|,
    \end{align}
    and clearly $|(a/N)^{\i t}|=1$. 
    (Strictly speaking, we will apply this identity to the muted versions of $Z$ and $Z_k$ introduced below, which are bounded away from $-1$ so that all logarithms are finite; note that $Z_k=-1$ may occur with positive probability for a discrete design.)

    Let $f(z)=(1+z)^{\i t}$. We write $f(z)=P_k(z)+R_k(z)$, where $P_k$ is the Taylor polynomial of degree $k-1$ and $R_k$ is the remainder term.
    By the integral form of the Taylor remainder, the remainder can be written as 
    \begin{align}
        R_k(z) = \frac{1}{(k-1)!}\int_0^z (z-w)^{k-1}f^{(k)}(w)\d w,
    \end{align}
    where $f^{(k)}(w)$ is the $k$-th derivative of $f$ at $w$, and can be simplified as
    \begin{align}
        f^{(k)}(w) = \frac{\partial^k}{\partial w^k} (1+w)^{\i t} = \prod_{m=0}^{k-1} (\i t - m) \cdot (1+w)^{\i t-k}.
    \end{align}
    Therefore the remainder for $z>0$ can be bounded as
    \begin{align}
        |R_k(z)| 
        &\leq k \cdot \left|\binom{\i t}{k}\right| \cdot \int_0^z (z-w)^{k-1} |(1+w)^{\i t -k}|\d w \\
        &= k \cdot \left|\binom{\i t}{k}\right| \cdot \int_0^z (z-w)^{k-1} (1+w)^{-k}\d w \\
        &= k \cdot \left|\binom{\i t}{k}\right| \cdot z^k \cdot \int_0^1 (1-u)^{k-1}(1+zu)^{-k}\d u \\
        &= k \cdot \left|\binom{\i t}{k}\right| \cdot z^k \cdot \int_0^1 u^{k-1}((1+z)-zu)^{-k}\d u.
    \end{align}
    We can further simplify the integral via a change of variable $v=\frac{u}{1+z-zu}$. 
    This gives $\d u = \frac{1+z}{(1+zv)^2}\d v$, and 
    \begin{align}
        \int_0^1 u^{k-1}((1+z)-zu)^{-k}\d u
        = \int_0^1 \frac{v^{k-1}}{1+zv}\d v.
    \end{align}
    For $z>0$, the integral is bounded by $1/k$. For $-1<z<0$, the same computation applies (the substitution $w=zu$ reverses the orientation of the integral, which is accounted for by taking absolute values, the map $u\mapsto v$ is still increasing, and $1+zv\geq 1+z>0$), and the integral is bounded by $\frac{1}{k(1+z)}$.
    So overall the integral is bounded by $\frac{1}{k(1+\min\{0,z\})}$.
    For real $z>-1$, we have 
    \begin{align}\label{eq:remainder-bound}
        |R_k(z)| \leq \left|\binom{\i t}{k}\right| \cdot \frac{|z|^k}{1+\min\{0,z\}}.
    \end{align}
    To show closeness between the distribution functions, we must deal with the behavior of $(z+1)^{\i t}$ as $z$ approaches the branch point $-1$.
    While $Z>-1$ almost surely and $Z_k\geq-1$, they can be arbitrarily close to $-1$ (and $Z_k$ may equal $-1$) with non-zero probability. 
    Thus the bound \eq{remainder-bound} on the remainder term $R_k(z)$ is not uniform near $z=-1$ (and $(1+z)^{\i t}$ itself is undefined at $z=-1$), which makes it hard to directly bound the closeness between the distribution functions.

    Let $\ell=k$. (We keep a separate letter for the truncation degree: the argument below uses the design property only for polynomials of degree at most $\ell$ in $X_k$, so it applies to any $1\leq\ell\leq k$ with $a\geq 2\pi\ell$.)
    To deal with this issue, we introduce a derived random variable from $Z$ and $Z_k$: for $W\in\{Z,Z_k\}$, let $\tilde W$ be a muted version of $W$, defined as
    \begin{align}\label{eq:muted}
        \tilde W = \begin{cases}
            W, &\text{if } W\geq -1/2 \\
            0, &\text{if } W<-1/2
        \end{cases}.
    \end{align}
    The distribution functions of $W$ and $\tilde W$ are close:
    \begin{align}
        \Delta_0 &:= \sup_{x\in\R} |F_{a(W+1)Y/N}(x) - F_{a(\tilde W+1)Y/N}(x)| \\\label{eq:muted-bound}
        &\leq \Pr[\tilde W\neq W] 
        = \Pr[W<-1/2] 
        \leq 2^\ell \Exp[|W|^{\ell}] 
        \leq \sqrt{1+\epsilon}\cdot 2^\ell \left(\frac{4\ell}{ea}\right)^{\ell/2} \leq 2 \cdot \left(\frac{16\ell}{ea}\right)^{\ell/2},
    \end{align}
    for $\epsilon\leq 1$, by Markov's inequality.
    Here the bound on $\Exp[|W|^\ell]$ is \cor{absolute-moment} for $W=Z$ (recall $a\geq 2\pi\ell$), while for $W=Z_k$ we use
    \begin{align}\label{eq:zk-moment-transfer}
        \Exp[|Z_k|^\ell]\leq\Exp[Z_k^{2\ell}]^{1/2}\leq\big((1+\epsilon)\Exp[Z^{2\ell}]\big)^{1/2}\leq\sqrt{1+\epsilon}\left(\frac{4\ell}{ea}\right)^{\ell/2},
    \end{align}
    where the second inequality is \thm{polynomial-relative-error} applied to the real polynomial $x\mapsto(Nx/a-1)^\ell$ of degree $\ell\leq k$ in $X_k$, and the third is the bound $\Exp[Z^{2\ell}]\leq(4\ell/(ea))^{\ell}$ from the proof of \cor{absolute-moment} (note that $|Z_k|^\ell$ itself need not be a polynomial in $Z_k$ when $\ell$ is odd).
    For notational convenience, let us denote 
    \begin{align}\label{eq:eta}
    \eta=2\left(\frac{16\ell}{ea}\right)^{\ell/2}.
    \end{align}
    This bound is $2^{-\Omega(k)}$ when $a\geq C k$ for some sufficiently large constant $C$. 
    We then bound the difference of the characteristic functions of $\ln(\tilde Z+1)$ and $\ln(\tilde Z_k+1)$, apply \lem{berry-esseen}, and finally apply a triangle inequality to finish the proof. 

    Recall that $(1+z)^{\i t}$ can be expanded as a Taylor series around $z=0$ as $f(z)=\sum_{m=0}^\infty \binom{\i t}{m} z^m$.
    Let $Q_\ell(Z)=P_\ell(Z)-P_\ell(0) = P_\ell(Z)-1$. 
    Then, for the truncation at degree $\ell-1$, the polynomial $P_\ell$ satisfies
    \begin{align}
        \left| \Exp[P_\ell(Z)]-\Exp[P_\ell(Z_k)]\right| 
        &= \left| \Exp[Q_\ell(Z)]-\Exp[Q_\ell(Z_k)]\right| \leq \epsilon \sqrt{\Exp[|Q_\ell(Z)|^2]}.
    \end{align}
    The inequality is the first-moment bound \eq{relative-error-pq} of \thm{polynomial-relative-error} (with $p\equiv 1$), applied to the polynomial $Q_\ell$ of degree $\ell-1\leq k$ with complex coefficients, viewed as a polynomial in $X_k=\frac{a}{N}(Z_k+1)$.
    We must bound the closeness when $Z$ and $Z_k$ are replaced by their muted versions $\tilde Z$ and $\tilde Z_k$.
    To show this, for $W\in\{Z,Z_k\}$, since $W$ and $\tilde W$ differ only when $W<-1/2$, we have
    \begin{align}
        \left|\Exp_W[P_\ell(W)]-\Exp_{\tilde W}[P_\ell(\tilde W)]\right|
        &= \left|\Exp_{W}\left[(P_\ell(W)-P_\ell(0))\mathbf{1}_{\{W<-1/2\}}\right]\right| \\
        &\leq \sqrt{\eta\Exp_W\left[|P_\ell(W)-P_\ell(0)|^2\right]} 
        \leq \sqrt{\eta\Exp_W\left[|Q_\ell(W)|^2\right]}.
    \end{align}
    Therefore, by the triangle inequality, the bound $\Exp_{Z_k}[|Q_\ell(Z_k)|^2]\leq(1+\epsilon)\Exp_Z[|Q_\ell(Z)|^2]$ from \eq{relative-error-p}, and Minkowski's inequality, we have for $\epsilon\leq 1$, 
    \begin{align}\label{eq:truncation-bound}
        \mathcal I_1(t) 
        &:=\left|\Exp[P_\ell(\tilde Z_k)] - \Exp[P_\ell(\tilde Z)] \right| \\
        &\leq\left|\Exp[P_\ell(Z_k)] - \Exp[P_\ell(Z)] \right| + \sqrt{\eta\Exp_Z\left[|Q_\ell(Z)|^2\right]} + \sqrt{\eta\Exp_{Z_k}\left[|Q_\ell(Z_k)|^2\right]}\\
        &\leq (\epsilon+\sqrt\eta(1+\sqrt{1+\epsilon})) \Exp_Z[|Q_\ell(Z)|^2]^{1/2} 
        \leq  (\epsilon+3\sqrt\eta) \cdot \sum_{j=1}^{\ell-1} \left|\binom{\i t}{j}\right| \cdot \sqrt{\Exp_Z\left[Z^{2j}\right]}.
    \end{align}
    It is crucial that this bound vanishes as $t\to 0$, since we will be applying the Berry--Esseen smoothing lemma, which has a factor of $1/|t|$ in the integral.

    For the remainder term, we can bound the expectation of $|R_\ell(\tilde Z)|$ as
    \begin{align}\label{eq:remainder-bound-2}
        \Exp_Z[|R_\ell(\tilde Z)|]
        &\leq \left|\binom{\i t}{\ell}\right| \cdot \Exp_Z\left[\frac{|\tilde Z|^\ell}{1+\min\{0,\tilde Z\}}\right] \\
        &\leq \left|\binom{\i t}{\ell}\right| \cdot 2 \Exp_Z[|Z|^\ell] 
        \leq \left|\binom{\i t}{\ell}\right| \cdot 2 \cdot \left(\frac{4\ell}{ea}\right)^{\ell/2}.
    \end{align}
    The same bound holds for $\tilde Z_k$ up to an extra factor $\sqrt{1+\epsilon}\leq \sqrt 2$, by \eq{zk-moment-transfer}.

    Now we bound the closeness between $\Exp[(1+\tilde Z)^{\i t}]$ and $\Exp[(1+\tilde Z_k)^{\i t}]$.
    By triangle inequality, we have
    \begin{align}
        \left|\Exp[(1+\tilde Z_k)^{\i t}] - \Exp[(1+\tilde Z)^{\i t}]\right| 
        &\leq \mathcal I_1(t)  + \mathcal I_2(t),
    \end{align}
    where $\mathcal I_1(t)$ is defined in \eq{truncation-bound}, and 
    \begin{align}\label{eq:remainder-bound-3}
        \mathcal I_2(t)
        &:=|\Exp[R_\ell(\tilde Z)]| + |\Exp[R_\ell(\tilde Z_k)]|,
    \end{align}    
    which, by \eq{remainder-bound-2} and the remark following it (note $2+2\sqrt 2\leq 6$), satisfies 
    \begin{align}
        \mathcal I_2(t) 
        &\leq  \left|\binom{\i t}{\ell}\right| \cdot 6 \cdot \left(\frac{4\ell}{ea}\right)^{\ell/2} = \frac{3\eta}{2^{\ell}} \cdot \left|\binom{\i t}{\ell}\right|.
    \end{align}

    We now apply \lem{berry-esseen} to the random variables $\ln(\tilde X_kY)$ and $\ln(\tilde XY)$, where $\tilde X:=\frac{a}{N}(\tilde Z+1)$ and $\tilde X_k:=\frac{a}{N}(\tilde Z_k+1)$; since $\ln$ is increasing, their Kolmogorov distance equals that between $\tilde X_kY$ and $\tilde XY$. The density of $\ln(\tilde XY)=\ln\tilde X+\ln Y$ is bounded by the maximum of the density of $\ln Y$, which is $(a-1)e^u(1-e^u)^{a-2}\leq 1$ for $u<0$, so we may take $M=1$ and let $T\to\infty$ in \lem{berry-esseen} (the integrals below converge). This amounts to integrating every part of the bound against the kernel $1/|t|$ multiplied by the characteristic function of $\ln Y$, which is $\varphi_{\ln Y}(t)=\Exp[Y^{\i t}]=\frac{\Gamma(a)\Gamma(1+\i t)}{\Gamma(a+\i t)}=\binom{\i t+a-1}{a-1}^{-1}$.
    For the first part, by \lem{binomial-integral} applied with $a-1$ in place of $a$ (note that $a-1\geq j+2$), we have
    \begin{align}
        \Delta_1 &:= \int_{-\infty}^\infty \frac{\mathcal I_1(t)}{|t|} |\varphi_{\ln Y}(t)| \d t \\
        &\leq (\epsilon+3\sqrt\eta) \cdot \sum_{j=1}^{\ell-1} \sqrt{\Exp[Z^{2j}]} \cdot \int_{-\infty}^\infty \frac{1}{|t|}\left|\binom{\i t}{j}\right| \cdot |\varphi_{\ln Y}(t)| \d t \\\label{eq:delta1-asymptotic}
        &\leq (\epsilon+3\sqrt\eta) \cdot \sum_{j=1}^{\ell-1} \left(\frac{4j}{ea}\right)^{j/2} \cdot \frac{a-1}{j}\frac{\sqrt{\pi}\Gamma((a-j-1)/2)}{\Gamma((a-j)/2)} 
        = O(\epsilon+\sqrt{\eta}).
    \end{align}
    The asymptotic form in \eq{delta1-asymptotic} can be shown by bounding the sequence in the summation by a decreasing sequence.
    In more detail, denoting $\beta_j:=\left(\frac{4j}{ea}\right)^{j/2} \cdot \frac{a-1}{j}\frac{\sqrt{\pi}\Gamma((a-j-1)/2)}{\Gamma((a-j)/2)}$,
    we compute the ratio $\beta_{j+1}/\beta_j$ as
    \begin{align}
        \frac{\beta_{j+1}}{\beta_j}
        &= \left(1+\frac{1}{j}\right)^{j/2}\left(\frac{4(j+1)}{ea}\right)^{1/2} \cdot \frac{a-1}{j+1} \frac{\sqrt{\pi}\Gamma((a-j-2)/2)}{\Gamma((a-j-1)/2)} \cdot \frac{j}{a-1}\frac{\Gamma((a-j)/2)}{\sqrt{\pi}\Gamma((a-j-1)/2)} \\
        &= \left(1+\frac{1}{j}\right)^{(j-1)/2} \left(\frac{4j}{ea}\right)^{1/2} \frac{\Gamma((a-j-2)/2)\Gamma((a-j)/2)}{\Gamma((a-j-1)/2)^2}  \\
        &= \left(1+\frac{1}{j}\right)^{(j-1)/2}\left(\frac{4j}{ea}\right)^{1/2} \cdot \frac{2}{a-j-2}\frac{\Gamma((a-j)/2)^2}{\Gamma((a-j-1)/2)^2} \\\label{eq:ratio-bound}
        &\leq 2\sqrt{\frac{j}{a}}\left(\frac{a-j}{a-j-2}\right) 
        \leq 2\sqrt{\frac{\ell}{2\pi\ell}}\left(1 + \frac{2}{(2\pi-1)\ell-1}\right) 
        \leq \sqrt{\frac{2}{\pi}} \cdot \left(1+\frac{2}{4\pi-3} \right) < 0.965,
    \end{align}
    since $\ell\geq 3$ and $a\geq 2\pi k=2\pi\ell$ (for $\ell\leq 2$ the sum has at most one term), where the fourth line uses $(1+1/j)^{(j-1)/2}\leq\sqrt e$ and the inequality $\Gamma(x+1/2)/\Gamma(x)\leq\sqrt x$ for $x>0$ (Wendel's inequality). 
    The bound in \eq{ratio-bound} implies that the summation can be bounded by a geometric series with ratio strictly smaller than 1.
    Thus the entire series can be bounded by a constant multiple of the first term, $\beta_1$, which is at most an absolute constant: by the inequality $\Gamma(x)/\Gamma(x+1/2)\leq\sqrt{x+1/2}/x$ (also due to Wendel), $\beta_1\leq 2\sqrt{2\pi/e}\,(a-1)^{3/2}/(\sqrt a(a-2))\leq 4$ for $a\geq 7$. 
    For the second part, we have
    \begin{align}
        \Delta_2 &:=\int_{-\infty}^\infty \frac{\mathcal I_2(t)}{|t|} |\varphi_{\ln Y}(t)| \d t \leq \frac{3\eta}{2^\ell} \cdot \int_{-\infty}^\infty \frac{1}{|t|}\left|\binom{\i t}{\ell}\right| \cdot |\varphi_{\ln Y}(t)| \d t \\
        &\leq  \frac{3\eta}{2^\ell} \cdot \frac{a-1}{\ell}\frac{\sqrt{\pi}\Gamma((a-\ell-1)/2)}{\Gamma((a-\ell)/2)} = O\left(\frac{\eta}{2^\ell}\cdot \frac{\sqrt a}{\ell}\right) = O\left(\frac{1}{\sqrt\ell}\left(\frac{4\ell}{ea}\right)^{(\ell-1)/2}\right),
    \end{align}
    where we used \lem{binomial-integral} with $a-1$ in place of $a$ and the identity $\eta\sqrt a/(2^\ell\ell)=\frac{4}{\sqrt{e\ell}}(4\ell/(ea))^{(\ell-1)/2}$.
    Finally, by the triangle inequality (the muting bound \eq{muted-bound} applied to $W=Z$ and to $W=Z_k$, and \lem{berry-esseen} for the muted variables, where we drop the factor $1/\pi$), we have
    \begin{align}
        \sup_{x\in\R}|F_{X_kY}(x) - F_{XY}(x)|
        &\leq 2\eta + \Delta_1 + \Delta_2 
        \leq O\left(\eta+\sqrt{\eta}+\epsilon\right) + O\left(\frac{1}{\sqrt\ell}\left(\frac{4\ell}{ea}\right)^{(\ell-1)/2}\right).
    \end{align}
    For $a\geq 2\pi k=2\pi\ell$, $\frac{16\ell}{ea}\leq \frac{8}{\pi e}< 0.94$ and $\frac{4\ell}{ea}<0.24$; therefore $\sqrt{\eta}<\sqrt 2\cdot 0.94^{\ell/4}=2^{-\Omega(k)}$, and the last term is $2^{-\Omega(k)}$ for $\ell\geq 2$ (for $\ell=k=1$ the claim is trivial since the Kolmogorov distance is at most $1$).
\end{proof}

\subsection{Lift to closeness in total variation}\label{sec:locally-invariant-tv}

Our approach to prove closeness in distribution makes it possible to lift the convergence in distribution to convergence in total variation. 

We first prove an analogue of \lem{binomial-integral} that will be useful for proving \thm{convergence-total-variation}.
\begin{lemma}\label{lem:integral-bound}
        For integers $j\geq 1$ and $a\geq 2j+3$,
        \begin{align}
            \int_{-\infty}^{\infty} \left|\binom{\i t}{j}\right| \cdot \left|\binom{\i t+a-1}{a-1}\right|^{-1} \d t < 16.
        \end{align}
\end{lemma}
    \begin{proof}
        The integral as a function of $a$ is decreasing, so it suffices to consider the case where $a=2j+3$.
        By direct calculation, we have
        \begin{align}
            &\int_{-\infty}^{\infty} \left|\binom{\i t}{j}\right| \cdot \left|\binom{\i t+a-1}{a-1}\right|^{-1} \d t
            = \frac{2}{j} \int_0^\infty t \prod_{r=j}^{a-1}\left(1 + \frac{t^2}{r^2}\right)^{-1/2} \d t \\
            &\qquad< \frac{2}{j} \int_0^\infty t \left(1 + \frac{t^2}{(a-1)^2}\right)^{-(a-j)/2} \d t 
            \leq \frac{(a-1)^2}{j} \int_1^\infty \frac{\d u}{u^{(a-j)/2}} \\
            &\qquad\leq \frac{2(a-1)^2}{j(a-j-2)} 
            = \frac{2(2j+2)^2}{j(j+1)} = 8\left(1+\frac{1}{j}\right) \leq 16.
        \end{align}
    \end{proof}

\begin{theorem}\label{thm:convergence-total-variation}
    Under the assumptions of \thm{convergence-distribution} (in particular $\epsilon\in[0,1]$),
    \begin{align}
        d_{\mathrm{TV}}(X_kY, XY) = O(\epsilon\log(2/\epsilon)) + 2^{-\Omega(k)},
    \end{align}
    where at $\epsilon=0$ the first term is interpreted as zero.
\end{theorem}
\begin{proof}
    We borrow the notations introduced in the proof of \thm{convergence-distribution}.
    We first show that the infinity norm of the difference between the density functions of $\ln((\tilde Z_k+1)Y)$ and $\ln((\tilde Z+1)Y)$ is small (these random variables have densities since $\tilde Z,\tilde Z_k\geq-1/2$ and $Y$ has a density), and then convert this into an $L^1$ bound by a truncation argument. 

    Since the density function of a random variable is the inverse Fourier transform of its characteristic function (the characteristic functions are integrable here because $|\varphi_{\ln Y}(t)|=O(|t|^{-(a-1)})$; as in \sec{berry-esseen}, we omit the factor $1/(2\pi)$), we can bound the infinity norm of the difference as
    \begin{align}\label{eq:infinite-norm-bound}
        \sup_{x\in\R} |f_{\ln((\tilde Z_k+1)Y)}(x) - f_{\ln((\tilde Z+1)Y)}(x)| 
        &\leq \int_{-\infty}^\infty (\mathcal I_1(t) + \mathcal I_2(t)) \cdot |\varphi_{\ln Y}(t)| \d t,
    \end{align}
    where $\mathcal I_1(t)$ and $\mathcal I_2(t)$ are defined as in \eq{truncation-bound} and \eq{remainder-bound-3} respectively. 
    Define $\eta$ as in \eq{eta}. 

    By \lem{integral-bound} (applicable since $a\geq 2\pi\ell\geq 2\ell+3$), we have
    \begin{align}
        \int_{-\infty}^\infty \mathcal I_1(t) \cdot |\varphi_{\ln Y}(t)| \d t 
        &\leq (\epsilon+3\sqrt\eta) \cdot \sum_{j=1}^{\ell-1} \sqrt{\Exp[Z^{2j}]} \cdot \int_{-\infty}^\infty \left|\binom{\i t}{j}\right| \cdot \left|\binom{\i t+a-1}{a-1}\right|^{-1} \d t \\
        &\leq 16 \cdot (\epsilon+3\sqrt\eta) \cdot \sum_{j=1}^{\ell-1} \left(\frac{4j}{ea}\right)^{j/2}.
    \end{align}
    Letting $r_j = \left(\frac{4j}{ea}\right)^{j/2}$, we use the ratio test: for $1\leq j<\ell$, 
    \begin{align}
        \frac{r_{j+1}}{r_j} = \frac{\left(\frac{4(j+1)}{ea}\right)^{(j+1)/2}}{\left(\frac{4j}{ea}\right)^{j/2}} = \left(\frac{j+1}{j}\right)^{j/2} \cdot \left(\frac{4(j+1)}{ea}\right)^{1/2} \leq 2\cdot \sqrt{\frac{\ell}{a}} < \sqrt{\frac{2}{\pi}} < 0.8,
    \end{align}
    and therefore the summation is bounded by $r_1/(1-0.8)\leq 5$, since $r_1=(4/(ea))^{1/2}<1$. 

    For the second term in \eq{infinite-norm-bound}, again by \lem{integral-bound}, we have 
    \begin{align}
        \int_{-\infty}^\infty \mathcal I_2(t) \cdot |\varphi_{\ln Y}(t)| \d t
        &\leq 6 \cdot \left(\frac{4\ell}{ea}\right)^{\ell/2} \cdot \int_{-\infty}^\infty \left|\binom{\i t}{\ell}\right| \cdot |\varphi_{\ln Y}(t)| \d t \\
        &\leq 96 \cdot \left(\frac{4\ell}{ea}\right)^{\ell/2}  %
        = O\left(\frac{\eta}{2^{\ell}}\right).
    \end{align}
    Therefore, \eq{infinite-norm-bound} is at most 
    \begin{align}
        \gamma_k := 80(\epsilon+3\sqrt\eta)+48\eta 2^{-\ell} = 2^{-\Omega(k)}+O(\epsilon)
    \end{align}
    for $a\geq 2\pi k$ and $\ell=k$. We may assume $\gamma_k\leq 1$: otherwise $\epsilon\geq 1/240$ or $\eta\geq 2^{-18}$ (if $\epsilon<1/240$ and $\eta<2^{-18}$, then $\gamma_k<\frac13+240\cdot 2^{-9}+48\cdot 2^{-19}<1$), so the right-hand side of the claimed bound is $\Omega(1)$ (for a suitable choice of the constants) and the claim is trivial because the total variation distance is at most $1$.

    To obtain the closeness in total variation, we first show that $\ln((\tilde Z_k+1)Y)$ and $\ln((\tilde Z+1)Y)$ have densities that are close in $L^1$ norm.
    By \thm{convergence-distribution}, since the Kolmogorov distance between $XY$ and $X_kY$ is $O(\epsilon) + 2^{-\Omega(k)}$, the distance between the muted random variables $\tilde X_kY$ and $\tilde XY$, where $\tilde X=a(\tilde Z+1)/N$ and $\tilde X_k=a(\tilde Z_k+1)/N$, is also small:
    \begin{align}
        \delta_k &:= \sup_{x\in\mathbb R}|F_{\tilde X_k Y}(x) - F_{\tilde X Y}(x)| \\
        &\leq \sup_{x\in\mathbb R}|F_{X_k Y}(x) - F_{XY}(x)| + \Pr[Z<-1/2] + \Pr[Z_k<-1/2] \\
        &= O(\epsilon) + 2^{-\Omega(k)}.
    \end{align}
    Since for $W\in\{Z,Z_k\}$, $\tilde W\geq-1/2$, we have $(\tilde W+1)\in[1/2,N/a]$ and $\ln((\tilde W+1)Y)\in(-\infty,\ln(N/a)]$. 
    We then choose a cutoff $R>0$ and optimize the bound by choosing $R$ appropriately:
    \begin{align}
        \|f_{\ln((\tilde Z_k+1)Y)} - f_{\ln((\tilde Z+1)Y)}\|_{1} 
        &\leq \int_{-R-\ln a}^{R-\ln a} \gamma_k \d x + \Pr[|\ln(a(\tilde Z+1)Y)| > R] + \Pr[|\ln(a(\tilde Z_k+1)Y)| > R] \\
        &\leq \int_{-R-\ln a}^{R-\ln a} \gamma_k \d x + 2\cdot \Pr[|\ln(a(\tilde Z+1)Y)| > R] + 2\delta_k \\\label{eq:tail2}
        &\leq \int_{-R-\ln a}^{R-\ln a} \gamma_k \d x + 2\cdot \Pr[|\ln(a(Z+1)Y)| > R] + 2\eta + 2\delta_k,
    \end{align}
    where the last inequality uses the fact 
    \begin{align}
        \Pr[\tilde Z\neq Z] = \Pr[Z<-1/2] \leq \eta.
    \end{align}
    Note that the random variable $\ln(a(Z+1)Y)$ in \eq{tail2} is the logarithm of the scaled $\Beta(1,N-1)$ random variable $NXY$, i.e., $\ln(a(Z+1)Y)=\ln(NXY)$ with $XY\sim\Beta(1,N-1)$ by the product rule \eq{beta-product}.
    Then by directly using the probability density function of $XY$, we have for $R>0$,
    \begin{align}
        \Pr[\ln(NXY)<-R] \leq e^{-R}, \qquad \Pr[\ln(NXY)>R] \leq e^{-e^R/2} \leq e^{-R}.
    \end{align}
    This implies that
    \begin{align}
        \|f_{\ln((\tilde Z_k+1)Y)} - f_{\ln((\tilde Z+1)Y)}\|_{1} 
        &\leq 2R\gamma_k + 4e^{-R} + 2\eta + 2\delta_k.
    \end{align}
    We optimize the bound by choosing $R=\ln(2/\gamma_k)>0$, which gives
    \begin{align}
        \|f_{\ln((\tilde Z_k+1)Y)} - f_{\ln((\tilde Z+1)Y)}\|_{1} \leq 2\gamma_k\ln(2/\gamma_k)+2\gamma_k+2\eta+2\delta_k.
    \end{align}
    Writing $\gamma_k=C\epsilon+\rho$ with $C=80$ and $\rho=240\sqrt\eta+48\eta2^{-\ell}=2^{-\Omega(k)}$, and using that $g(u)=u\ln(2/u)$ is concave on $(0,2]$ with $g(0^+)=0$, hence subadditive there, we get $\gamma_k\ln(2/\gamma_k)\leq g(C\epsilon)+g(\rho)\leq C\epsilon\ln(2/\epsilon)+2^{-\Omega(k)}$. Together with $\gamma_k+\eta+\delta_k=O(\epsilon)+2^{-\Omega(k)}$ and $\log(2/\epsilon)\geq 1$ for $\epsilon\leq 1$, the right-hand side is $O(\epsilon\log(2/\epsilon)) + 2^{-\Omega(k)}$.
    Since the total variation distance is invariant under the bijection $u\mapsto e^u$ and equals half the $L^1$ distance between densities \eq{tv-l1}, this bounds $d_{\mathrm{TV}}((\tilde Z_k+1)Y,(\tilde Z+1)Y)$. It remains to show that the total variation distance between $(W+1)Y$ and $(\tilde W+1)Y$ is small for $W\in\{Z,Z_k\}$. 
    To see this, we apply the coupling technique: $\tilde W$ and $W$ are defined on the same probability space as in \eq{muted} and differ only on the event $\{W<-1/2\}$, so by \eq{muted-bound}
    \begin{align}
        d_{\mathrm{TV}}\big((\tilde W+1)Y,(W+1)Y\big) \leq \Pr[W<-1/2] \leq 2 \cdot \left(\frac{16\ell}{ea}\right)^{\ell/2}=\eta.
    \end{align}
    For $a\geq 2\pi\ell$, this is $2^{-\Omega(k)}$.
    Now combining the bounds through the triangle inequality
    \begin{align}
        d_{\mathrm{TV}}(X_kY,XY)\leq d_{\mathrm{TV}}(X_kY,\tilde X_kY)+d_{\mathrm{TV}}(\tilde X_kY,\tilde XY)+d_{\mathrm{TV}}(\tilde XY,XY),
    \end{align}
    and recalling that scaling by $a/N$ does not change the total variation distance, we obtain
    \begin{align}
        d_{\mathrm{TV}}(X_kY,XY) = O(\epsilon\log(2/\epsilon)) + 2^{-\Omega(k)}.
    \end{align}
\end{proof}

Combining \lem{block-factorization} with \thm{convergence-distribution} and \thm{convergence-total-variation} yields the statements announced in \thm{ks-locally-invariant} and \thm{tv-locally-invariant}.

\begin{corollary}\label{cor:locally-invariant-entry}
    Let $k$ be a positive integer, let $a=2^m\geq 2\pi k$ and $N=2^n>2a$, and let $\epsilon\in[0,1]$.
    Let $\mu(k,\epsilon)$ be an $m$-locally invariant $\epsilon$-approximate unitary $k$-design on $n$ qubits, and for fixed $x,y\in\bit^n$ let $\nu(k,\epsilon)$ be the law of $|\bra{x}U\ket{y}|^2$ for $U\sim\mu(k,\epsilon)$, and let $\nu$ be the law of $|\bra{x}U\ket{y}|^2$ for a Haar random $U$ on $n$ qubits, i.e., $\nu=\Beta(1,N-1)$.
    Then
    \begin{align}
        d_{\mathrm K}\big(\nu(k,\epsilon),\nu\big)\leq 2^{-\Omega(k)}+O(\epsilon),
        \qquad
        d_{\mathrm{TV}}\big(\nu(k,\epsilon),\nu\big)\leq 2^{-\Omega(k)}+O(\epsilon\log(2/\epsilon)).
    \end{align}
\end{corollary}
\begin{proof}
    By \lem{block-factorization}, $\nu(k,\epsilon)$ is the law of $X_kY$, where $X_k$ is the probability of the block $B(x)$ in column $y$ for $U\sim\mu(k,\epsilon)$ and $Y\sim\Beta(1,a-1)$ is independent of $X_k$, while $\nu$ is the law of $XY$ for $X\sim\Beta(a,N-a)$ independent of $Y$.
    The claims are therefore \thm{convergence-distribution} and \thm{convergence-total-variation}.
\end{proof}

\subsection{Local invariance on a single qubit}\label{sec:one-local-invariance}

In this section, we show that the total variation distance is at most $O(k^{-1/2} + \epsilon)$ for $1$-locally invariant $\epsilon$-approximate unitary $k$-design.  

First, we show that we can bound the total variation distance through pointwise closeness between the distribution functions. 

\begin{lemma}\label{lem:total-variation-df}
    Let $\mu,\nu$ be two distributions on $[0,\infty)$.
    For a distribution $\rho$ on $[0,\infty)$, let $\mathcal M\rho$ be the distribution of $XY$ for independent $X\sim\rho$ and $Y\sim\Uniform[0,1]$ (the uniform distribution on $[0,1]$), and let $F_\rho$ be the distribution function of $\rho$.
    Then for every $b>0$,
    \begin{align}
        d_{\mathrm{TV}}(\M\mu,\M\nu) \leq F_\mu(b) + F_\nu(b) + \int_b^\infty \frac1t |F_\mu(t)-F_\nu(t)| \d t.
    \end{align}
\end{lemma}
\begin{proof}
    We do not assume that $\rho$ has a density; in particular, $\rho$ may have an atom at $0$, in which case $\M\rho$ has an atom of the same mass at $0$.
    Conditioned on $X=t>0$, the product $XY$ is uniformly distributed on $[0,t]$.
    Hence $\M\rho(\{0\})=\rho(\{0\})$, and on $(0,\infty)$ the distribution $\M\rho$ has the density
    \begin{align}\label{eq:m-rho}
        f_{\M\rho}(x) = \int_{(x,\infty)} \frac{1}{t}\d\rho(t), \qquad x>0.
    \end{align}
    Integrating by parts in the Stieltjes sense (the boundary term at infinity vanishes since $F_\rho(t)/t\to 0$), we can express this density through the distribution function: for every $x>0$,
    \begin{align}\label{eq:m-rho-df}
        f_{\M\rho}(x) = -\frac{F_\rho(x)}{x} + \int_x^\infty \frac{F_\rho(t)}{t^2}\d t.
    \end{align}
    Write $\Delta(t):=F_\mu(t)-F_\nu(t)$.
    By \eq{tv-l1}, applied with $\lambda$ the sum of the Dirac measure at $0$ and the Lebesgue measure on $(0,\infty)$, the total variation distance is
    \begin{align}
        d_{\mathrm{TV}}(\M\mu,\M\nu) = \frac12\big|\mu(\{0\})-\nu(\{0\})\big| + \frac12\int_0^\infty \big|f_{\M\mu}(x)-f_{\M\nu}(x)\big|\d x,
    \end{align}
    and we split the integral at the cutoff $b$.

    For $x\geq b$, \eq{m-rho-df} gives $|f_{\M\mu}(x)-f_{\M\nu}(x)|\leq |\Delta(x)|/x+\int_x^\infty |\Delta(t)|/t^2\d t$.
    Integrating over $x\in[b,\infty)$ and exchanging the order of integration in the second term,
    \begin{align}\label{eq:tv-df-above}
        \frac12\int_b^\infty \big|f_{\M\mu}(x)-f_{\M\nu}(x)\big|\d x
        &\leq \frac12\int_b^\infty \frac{|\Delta(x)|}{x}\d x + \frac12\int_b^\infty |\Delta(t)|\,\frac{t-b}{t^2}\d t \nonumber\\
        &= \int_b^\infty \frac{|\Delta(t)|}{t}\d t - \frac b2\int_b^\infty\frac{|\Delta(t)|}{t^2}\d t.
    \end{align}

    For $0<x<b$, by \eq{m-rho}, $f_{\M\mu}(x)-f_{\M\nu}(x) = \int_{(x,b]}\frac1t\d(\mu-\nu)(t) + \delta_b$, where $\delta_b:=\int_{(b,\infty)}\frac1t\d(\mu-\nu)(t)$ does not depend on $x$.
    By Fubini's theorem,
    \begin{align}
        \int_0^b\Big|\int_{(x,b]}\frac1t\d(\mu-\nu)(t)\Big|\d x
        &\leq \int_0^b\int_{(x,b]}\frac1t\d(\mu+\nu)(t)\d x \\
        &=\int_{(0,b]}\frac{t}{t}\d(\mu+\nu)(t)= \mu((0,b])+\nu((0,b]),
    \end{align}
    and by \eq{m-rho} and \eq{m-rho-df} at $x=b$, $b|\delta_b| = \big|{-\Delta(b)}+b\int_b^\infty \Delta(t)/t^2\d t\big|\leq |\Delta(b)|+b\int_b^\infty|\Delta(t)|/t^2\d t$.
    Hence
    \begin{align}\label{eq:tv-df-below}
        \frac12\int_0^b \big|f_{\M\mu}(x)-f_{\M\nu}(x)\big|\d x \leq \frac12\Big(\mu((0,b])+\nu((0,b])+|\Delta(b)|\Big)+\frac b2\int_b^\infty\frac{|\Delta(t)|}{t^2}\d t.
    \end{align}
    Adding \eq{tv-df-above}, \eq{tv-df-below}, and the atom term $\frac12|\mu(\{0\})-\nu(\{0\})|\leq\frac12(\mu(\{0\})+\nu(\{0\}))$, the two integrals of $|\Delta(t)|/t^2$ cancel, and using $|\Delta(b)|\leq F_\mu(b)+F_\nu(b)$ we obtain
    \begin{align}
        d_{\mathrm{TV}}(\M\mu,\M\nu)\leq \frac12\big(F_\mu(b)+F_\nu(b)\big)+\frac12\big(F_\mu(b)+F_\nu(b)\big)+\int_b^\infty\frac{|\Delta(t)|}{t}\d t,
    \end{align}
    which is the claim.
\end{proof}

For our purposes, $\mu$ and $\nu$ will be the distribution $\sigma$ of $NX$ for $X\sim\Beta(2,N-2)$ and the corresponding distribution $\sigma(k,\epsilon)$ of the scaled block probability of a design.
To see how \lem{total-variation-df} is used, suppose for the moment that we only knew a uniform bound $\eta$ on the Kolmogorov distance between $\sigma$ and $\sigma(k,\epsilon)$ (\lem{pointwise-df-to-christoffel} and \lem{pointwise-christoffel} below give $\eta=O(k^{-1/2}+\epsilon)$).
We have $F_\sigma(t)\leq t^2/2$, $\Exp_{T\sim\sigma}[T]=2$ and $\Exp_{T\sim\sigma}[T^2]=6N/(N+1)\leq 6$, and by the Cauchy--Schwarz inequality and \thm{polynomial-relative-error}, $\Exp_{T\sim\sigma(k,\epsilon)}[T]\leq\sqrt{(1+\epsilon)\Exp_{T\sim\sigma}[T^2]}\leq 4$.
Hence, by Markov's inequality, for $b\leq 1\leq R$,
\begin{align}
    d_{\mathrm{TV}}(\M\sigma,\M\sigma(k,\epsilon))
    &\leq b^2 + \eta + \int_b^R \frac{\eta}{x}\d x + \int_R^\infty \frac{\Pr_{T\sim\sigma}[T>x]+ \Pr_{T\sim\sigma(k,\epsilon)}[T>x]}{x}\d x \\
    &\leq b^2 + \eta + \eta\ln\frac{R}{b} + \int_R^\infty \frac{6}{x^2}\d x
    = b^2 + \eta + \eta\ln\frac{R}{b} + \frac{6}{R}.
\end{align}
Choosing $b=\sqrt\eta$ and $R=2/\eta$, we have that
\begin{align}
d_{\mathrm{TV}}(\M\sigma,\M\sigma(k,\epsilon)) =O(\eta\log(2/\eta))=O\left(\frac{\log k}{k^{1/2}} + \epsilon\log(2/\epsilon)\right).
\end{align}

To remove the $O(\log k)$ factor that appears in the bound on the total variation distance, we have to give an $x$-dependent bound on $|F_\sigma(x)-F_{\sigma(k,\epsilon)}(x)|$. 
To give such a bound, we observe that in our proof for the Kolmogorov convergence in \sec{convergence-upper-bound}, we consider the majorant $U_i$ and $L_i$ for the step function $1_{x\leq x_i}$ and quadrature nodes $x_i$ which are roots of the $(k+1)$st shifted Jacobi polynomial. 
This allows us to bound the Kolmogorov distance at the nodes by the weight $w_i$. 
Hence, applying the same argument to $\sigma$ (see \rmk{block-kolmogorov}), at each node $x=x_i$ of the $(k+1)$-node Gauss quadrature rule of $\sigma$ the difference $|F_\sigma(x)-F_{\sigma(k,\epsilon)}(x)|$ is bounded by $O(w_i+\epsilon)$, where $w_i=\lambda_{\sigma,k+1}(x_i)$ by \lem{christoffel-function-weight}; here $\lambda_{\sigma,k+1}$ is the Christoffel function of $\sigma$ in \eq{christoffel-mu-d} with $\ell=k+1$, whose admissible polynomials have degree at most $k$. 
For $x\notin\{x_1,\ldots,x_{k+1}\}$, we rely on the monotonicity of distribution functions to bound the distance by $O(\max_i w_i+\epsilon)$.

Therefore, to get an $x$-dependent bound, for every $x\in\mathbb R$, it is fair to attempt that one can bound the difference at $x$ by the Christoffel function. 
We then construct a quadrature with a node at $x$ and then the majorant $U_x$ and minorant $L_x$ for the step function $f(t)=1_{t\leq x}$. 
This is possible because the constructions of the majorant in \lem{majorant-indicator} and of the minorant in \lem{minorant-indicator} do not depend on the particular choice of the nodes. 

In \eq{christoffel-mu-d}, the Christoffel function can be written as the expectation of some non-negative function.
Since \thm{polynomial-relative-error} controls squared polynomials of degree at most $k$, the relevant Christoffel function throughout this section is $\lambda_{\sigma,k+1}$.
For any fixed $x$, let $p_x$ be the polynomial of degree at most $k$ that attains the variational form in \eq{christoffel-mu-d} for $\lambda_{\sigma,k+1}(x)$, i.e.,
\begin{align}\label{eq:christoffel-minimizer}
    \lambda_{\sigma,k+1}(x)=\Exp_{X\sim\sigma}[p_x(X)^2]
\end{align}
 with $p_x(x)=1$.
Then choosing the quadrature nodes to be the roots of $p_x$, as long as the roots of $p_x$ are all real and simple, we then use the constructions in \lem{majorant-indicator} and \lem{minorant-indicator} to give a majorant and minorant. 

The following lemma shows that the roots of $p_x$ are all real and simple. 
\begin{lemma}\label{lem:saturating-polynomial}
    Let $\sigma$ be a distribution whose support is an interval and which has a density that is positive on every open subinterval of its support, and let $p_x$ be the polynomial of degree at most $k$ that attains the variational form in \eq{christoffel-mu-d} for $\lambda_{\sigma,k+1}(x)$.
    Then all the roots of $p_x$ are real and simple.
\end{lemma}
\begin{proof}
    Since $p_x$ attains the minimum, for every polynomial $v$ of degree no more than $k$ satisfying that $v(x)=0$, we consider the degree-2 real polynomial in $\alpha$:
    \begin{align}
        g(\alpha) = \Exp_{X\sim\sigma}[(p_x(X)+\alpha v(X))^2].
    \end{align}
    Then $g(\alpha)$ is minimized at $\alpha=0$. 
    This implies that 
    \begin{align}\label{eq:vanish-inner-product}
        \Exp_{X\sim\sigma}[p_x(X)v(X)]=0, \qquad \deg(v)\leq k \text{ and } v(x)=0.
    \end{align}
    
    Now we argue by contradiction. Suppose that $p_x$ has a real root $s$ of multiplicity at least $2$, i.e., there is a polynomial $r$ of degree at most $k-2$ such that $p_x(t)=(t-s)^2 r(t)$ (note that $s\neq x$ since $p_x(x)=1$). 
    Then we show that there exists $v$ that violates \eq{vanish-inner-product}.
    Let $v(t)=(t-x)^2 r(t)$ be a polynomial of degree at most $k$ satisfying $v(x)=0$.
    Then we have 
    \begin{align}
        \Exp_{X\sim\sigma}[p_x(X)v(X)]
        = \Exp_{X\sim\sigma}[(X-s)^2(X-x)^2 r^2(X)] > 0
    \end{align}
    since the polynomial $(t-s)^2(t-x)^2 r(t)^2$ is non-negative, vanishes only at finitely many points, and $\sigma$ has a positive density on an interval. 

    It remains to show that all the roots of $p_x$ are real. 
    Suppose not, then there exists a pair of conjugate roots $s+\i u$ and $s-\i u$ for $s,u\in\mathbb R$. 
    Then following the same idea as for the first case closely, there exists a polynomial $r$ such that $p_x(t)=((t-s)^2+u^2)r(t)$.
    Letting $v(t)=(t-x)^2 r(t)$, for the same reason we conclude that $\Exp_{X\sim\sigma}[p_x(X)v(X)]>0$---a contradiction. 
\end{proof}

By \lem{saturating-polynomial}, all the roots $t_1<\cdots<t_d$ of $p_x$ are real and simple, where $d=\deg(p_x)\leq k$, and none of them equals $x$.
Taking the $d+1$ points $\{x,t_1,\ldots,t_d\}$ as nodes, \lem{majorant-indicator} and \lem{minorant-indicator} provide a majorant $U_x$ and a minorant $L_x$ of $1_{t\leq x}$, both polynomials of degree at most $2d\leq 2k$, such that $U_x(t)-L_x(t)=h(t)^2$, where $h$ is the Lagrange basis polynomial of the node $x$,
\begin{align}
    h(t) = \prod_{i=1}^{d} \frac{t-t_i}{x-t_i}.
\end{align}
(If $d=0$, i.e., $p_x\equiv 1$, we take $U_x\equiv 1$ and $L_x\equiv 0$.)
Since $h$ and $p_x$ are polynomials with the same simple roots $t_1,\ldots,t_d$ and no others, and both take the value $1$ at $x$, it must hold that $p_x=h$.
By \eq{christoffel-minimizer},
\begin{align}\label{eq:difference-majorant-minorant}
    \Exp_{X\sim\sigma}[U_x(X)-L_x(X)] = \Exp_{X\sim\sigma}[p_x(X)^2] = \lambda_{\sigma,k+1}(x).
\end{align}

Now we are ready to bound the difference between the distribution functions pointwise. 

\begin{lemma}\label{lem:pointwise-df-to-christoffel}
    For integer $k\geq 1$ and real number $0\leq\epsilon\leq 1$, let $\sigma$ be a distribution as in \lem{saturating-polynomial} and let $\sigma(k,\epsilon)$ be a distribution such that for every real polynomial $p$ of degree at most $k$, 
    \begin{align}\label{eq:pointwise-bound-assumption}
        (1-\epsilon)\Exp_{X\sim\sigma}[p^2(X)] \leq 
        \Exp_{X\sim\sigma(k,\epsilon)}[p^2(X)]
        \leq 
        (1+\epsilon)
        \Exp_{X\sim\sigma}[p^2(X)].
    \end{align}
    Then for every $x\in\mathbb R$, 
    \begin{align}
        |F_{\sigma(k,\epsilon)}(x) - F_\sigma(x)|
        \leq \epsilon\min\{F_\sigma(x),1-F_\sigma(x)\} + (1+\epsilon)\lambda_{\sigma,k+1}(x).
    \end{align}
\end{lemma}
\begin{proof}
    The proof idea is very similar to that for \thm{approximate-design-upper-bound}.
    Since $U_x$ is a non-negative polynomial of degree at most $2k$, it is a sum of two squares of polynomials of degree at most $k$ by \lem{sos-decomposition}, so \eq{pointwise-bound-assumption} applies to it (and likewise to $1-L_x$ and $p_x^2$ below). Hence 
    \begin{align}
        F_{\sigma(k,\epsilon)}(x) 
        = \Exp_{X\sim\sigma(k,\epsilon)}[1_{X\leq x}] 
        \leq \Exp_{X\sim\sigma(k,\epsilon)}[U_x(X)] 
        \leq (1+\epsilon)\Exp_{X\sim\sigma}[U_x(X)].
    \end{align}
    Then by \eq{christoffel-minimizer} and \eq{difference-majorant-minorant}, we have 
    \begin{align}
        F_{\sigma(k,\epsilon)}(x)&\geq \Exp_{X\sim\sigma(k,\epsilon)}[L_x(X)] 
        = \Exp_{X\sim\sigma(k,\epsilon)}[U_x(X) - p_x(X)^2] \\
        &\geq 
        (1-\epsilon)\Exp_{X\sim\sigma}[U_x(X)] - (1+\epsilon) \Exp_{X\sim\sigma}[p_x(X)^2] \\
        &=
        (1-\epsilon)\Exp_{X\sim\sigma}[U_x(X)] - (1+\epsilon) \lambda_{\sigma,k+1}(x).
    \end{align}
    Letting $A:=\Exp_{X\sim\sigma}[U_x(X)]$ and $w:=\lambda_{\sigma,k+1}(x)$, we have 
    \begin{align}\label{eq:pointwise-approximate}
        (1-\epsilon)A - (1+\epsilon) w \leq F_{\sigma(k,\epsilon)}(x) \leq (1+\epsilon) A
    \end{align}
    It remains to bound $A$ by $F_{\sigma}(x)$: since $U_x\geq 1_{t\leq x}\geq L_x$ and $U_x-L_x=p_x^2$, we have 
    \begin{align}\label{eq:pointwise-exact}
        F_{\sigma}(x) \leq A \leq F_{\sigma}(x) + w.
    \end{align}
    Combining \eq{pointwise-approximate} and \eq{pointwise-exact}, we have
    \begin{align}\label{eq:pointwise-bound-1}
        (1-\epsilon)F_\sigma(x) - (1+\epsilon) w \leq F_{\sigma(k,\epsilon)}(x) \leq (1+\epsilon)(F_\sigma(x) + w)
    \end{align} 
    Some reordering applied on the inequalities gives
    \begin{align}
        |F_{\sigma(k,\epsilon)}(x)-F_\sigma(x)| \leq \epsilon F_\sigma(x) + (1+\epsilon)\lambda_{\sigma,k+1}(x).
    \end{align}

    Similarly, for the complementary distribution functions, let $\bar F_\rho=1-F_\rho$ for distribution $\rho$. 
    Letting $B:=\Exp_{X\sim\sigma}[1-L_x(X)]$, since $1-L_x$ is non-negative, we have
    \begin{align}
        \bar F_{\sigma(k,\epsilon)}(x)
        \leq (1+\epsilon)\Exp_{X\sim\sigma}[1-L_x(X)] = (1+\epsilon)B,
    \end{align}
    and
    \begin{align}
        \bar F_{\sigma(k,\epsilon)}(x) \geq \Exp_{X\sim\sigma(k,\epsilon)}[1-U_x(X)] &\geq (1-\epsilon)\Exp_{X\sim\sigma}[1-L_x(X)]-(1+\epsilon)\Exp_{X\sim\sigma}[p_x(X)^2] \\
        &= (1-\epsilon)B - (1+\epsilon)w.
    \end{align}
    On the other hand,
    $\bar F_\sigma(x) \leq B\leq \bar F_\sigma(x)+w$, which gives
    \begin{align}
        (1-\epsilon)\bar F_\sigma(x)-(1+\epsilon)w
        \leq
        \bar F_{\sigma(k,\epsilon)}(x) \leq (1+\epsilon)(\bar F_\sigma(x) + w).
    \end{align}
    The inequalities look exactly the same as \eq{pointwise-bound-1}, except that the distribution functions are replaced with their complements. 
    Thus we have 
    \begin{align}
        |\bar F_{\sigma(k,\epsilon)}(x)-\bar F_\sigma(x)| \leq \epsilon \bar F_\sigma(x) + (1+\epsilon)\lambda_{\sigma,k+1}(x).
    \end{align}
\end{proof}

\lem{pointwise-df-to-christoffel} reduces the problem of bounding the distribution functions pointwise to bounding the Christoffel function.
Before we prove a bound in \lem{pointwise-christoffel}, we present an auxiliary lemma that the polynomial $p_x(t)$ that attains $\lambda_{\rho_R,\ell}(x)$ is concentrated around $x$, for the arcsine measure $\rho_R$ on $[0,R]$. 

\begin{lemma}\label{lem:localization}
    For integers $\ell\geq 1$ and $R\geq 64\ell$, let $\rho_R$ be the arcsine measure on $[0,R]$, i.e., $\d\rho_R(t)=\frac{\d t}{\pi\sqrt{t(R-t)}}$, and for $x\in[0,R]$ let $p_x(t):=K_{\ell,R}(x,t)/K_{\ell,R}(x,x)$ be the polynomial of degree at most $\ell-1$ attaining $\lambda_{\rho_R,\ell}(x)$, as in \eq{test-polynomial} and \lem{christoffel-function-orthogonal-polynomial}.
    Then there exists an absolute constant $C$ such that for all $x\in[0,1]$ and $t\in[0,R]$,
    \begin{align}
        |p_x(t)| \leq \frac{C}{1+\ell|\theta-\phi|},
    \end{align}
    where $\theta=2\arcsin\sqrt{x/R}$ and $\phi=2\arcsin\sqrt{t/R}$.
\end{lemma}
\begin{proof}
    Applying the change of variables $x=R\sin^2(\theta/2)$ and $t=R\sin^2(\phi/2)$ with $\theta,\phi\in[0,\pi]$, we can write
    \begin{align}
        \d\rho_R(t) = \frac{R/2\cdot\sin\phi \d\phi}{\pi\cdot R/2\cdot \sin\phi} = \frac{\d\phi}{\pi}.
    \end{align}
    For $x\in[0,1]$, $\theta\leq 2\arcsin(\sqrt{1/R})\leq 2\arcsin(1/8)\leq 4\pi/45$, and therefore $|\theta-\phi|\leq\theta+\phi\leq 49\pi/45<5\pi/3$.

    By \eq{arcsine-basis}, the orthonormal polynomials for $\rho_R$ are $q_0=1$ and $q_j(t)=\sqrt{2}T_j(2t/R-1)=\sqrt 2 T_j(-\cos\phi)=\sqrt 2(-1)^j\cos(j\phi)$ for $j\geq 1$.
    Hence
    \begin{align}
        K_{\ell,R}(x,t)
        &= \sum_{j=0}^{\ell-1} q_j(x)q_j(t) = 1 + 2\sum_{j=1}^{\ell-1} \cos(j\theta)\cos(j\phi) \\
        &= 1 + \sum_{j=1}^{\ell-1} \cos(j(\theta+\phi)) + \sum_{j=1}^{\ell-1} \cos(j(\theta-\phi)).
    \end{align}
    To simplify, we define the Dirichlet kernel
    \begin{align}
        D_{m}(\alpha) &:= 1+ 2\sum_{j=1}^{m} \cos(j\alpha)
        = \sum_{j=-m}^{m} e^{\i j\alpha} = \frac{\sin((m+1/2)\alpha)}{\sin(\alpha/2)},
    \end{align}
    so that $2K_{\ell,R}(x,t) = D_{\ell-1}(\theta+\phi) + D_{\ell-1}(\theta-\phi)$.
    We claim that $|D_m(\alpha)|\leq C_1\min\{m+1,1/|\alpha|\}$ for $|\alpha|\leq 5\pi/3$, with $C_1=10\pi/3$.
    Indeed, $|D_m(\alpha)|\leq D_m(0)=2m+1\leq 2(m+1)$, and for $0<|\alpha|\leq 5\pi/3$ the concavity of $\sin$ on $[0,5\pi/6]$ gives $|\sin(\alpha/2)|\geq \frac{3|\alpha|}{10\pi}$, so that
    \begin{align}
        |D_m(\alpha)|\leq \frac{1}{|\sin(\alpha/2)|} \leq \frac{10\pi}{3|\alpha|}.
    \end{align}
    Since $|\theta+\phi|\geq|\theta-\phi|$, we obtain
    \begin{align}
        2|K_{\ell,R}(x,t)| &\leq |D_{\ell-1}(\theta+\phi)| + |D_{\ell-1}(\theta-\phi)|
        \leq 2C_1 \min\left\{\ell,\frac{1}{|\theta-\phi|}\right\}.
    \end{align}
    By \eq{kernel-lower}, $K_{\ell,R}(x,x)\geq \ell/4$, and therefore
    \begin{align}
        |p_x(t)| = \frac{|K_{\ell,R}(x,t)|}{K_{\ell,R}(x,x)} \leq \frac{4C_1}{\ell}\min\left\{\ell,\frac{1}{|\theta-\phi|}\right\}
        = 4C_1 \cdot \min\left\{1,\frac{1}{\ell|\theta-\phi|}\right\}
        \leq \frac{8C_1}{1+\ell|\theta-\phi|}.
    \end{align}
\end{proof}

Now given the concentration result \lem{localization}, we are ready to give a pointwise bound for the Christoffel function for the scaled beta measure. 

\begin{lemma}\label{lem:pointwise-christoffel}
    Let $\sigma$ be the distribution satisfying $NX\sim\sigma$ for $X\sim\Beta(2,N-2)$. 
    Then there exists an absolute constant $C>0$ such that for $N\geq 4$ and $k\geq 2$, 
    \begin{align}
        \lambda_{\sigma,k+1}(t) \leq 
        \begin{cases}
            Ck^{-1/2}(t+k^{-1}) & 0\leq t\leq 1 \\
            Ck^{-1/2}t^{-2} & t\geq 1.
        \end{cases}
    \end{align} 
\end{lemma}
\begin{proof}
    For $t\geq 1$, we rely on the fact that the tail is dominated by a constant multiple of the density of an exponential random variable.
    Indeed, recall that the probability density function of $\sigma$ is 
    \begin{align}\label{eq:scaled-beta-exponential}
        f_\sigma(t) = \frac{1}{N} f_{\Beta(2,N-2)}\left(\frac{t}N\right)= \frac{(N-1)(N-2)}{N^2} t\left(1-\frac{t}{N}\right)^{N-3}
        \leq te^{-t(1-3/N)} \leq te^{-t/4},
    \end{align}
    where the last inequality holds for $N\geq 4$. 
    By \eq{christoffel-mu-d}, $\lambda_{\sigma,k+1}(t)$ is bounded by $\Exp_\sigma[p^2]$ for any polynomial $p$ of degree at most $k$ with $p(t)=1$.
    Let $\omega$ be the mean-one exponential measure and $\lambda_{\omega,k}$ its Christoffel function, as in \lem{exponential} (its admissible polynomials have degree at most $k-1$), and let $r$ be a polynomial of degree at most $k-1$ with $r(t)=1$ such that $u\mapsto r(8u)$ attains $\lambda_{\omega,k}(t/8)$, i.e.,
    \begin{align}
        \lambda_{\omega,k}(t/8) = \int_0^\infty r(8u)^2 e^{-u}\d u = \frac18\int_0^\infty r(s)^2 e^{-s/8}\d s.
    \end{align}
    The polynomial $s\mapsto (s/t)\, r(s)$ has degree at most $k$ and takes the value $1$ at $s=t$.
    Therefore, by \eq{christoffel-mu-d} and \eq{scaled-beta-exponential},
    \begin{align}
        \lambda_{\sigma,k+1}(t)
        &\leq \int_0^\infty \frac{s^2}{t^2} r(s)^2 f_\sigma(s)\d s
        \leq \int_0^\infty \frac{s^2}{t^2} r(s)^2 s e^{-s/4}\d s \\
        &\leq \frac{C_1}{t^2}\int_0^\infty r(s)^2 e^{-s/8}\d s
        = \frac{8C_1}{t^2}\cdot \lambda_{\omega,k}(t/8),
    \end{align}
    where $C_1:=\max_{s\geq 0} s^3e^{-s/8}=(24/e)^3$.
    By \lem{exponential}, for $k\geq 2$ and some absolute constant $C_2$, $\lambda_{\sigma,k+1}(t)\leq \frac{C_2}{t^2\sqrt{k}}$.

    It remains to prove the bound for $0\leq t\leq 1$. 
    Recall that in the proof of \lem{exponential}, we compare the exponential measure with the arcsine measure to get a uniform bound on the Christoffel function of the exponential measure. 
    Recall from \eq{arcsine} that $\d\rho_R(s)=\frac{\d s}{\pi\sqrt{s(R-s)}}$ is the arcsine measure on $[0,R]$, for a parameter $R$ to be determined later.
    Following the same idea as in \eq{interior-domination}, for $s\in[0,R]$, 
    \begin{align}\label{eq:interior-domination-2}
        \frac{e^{-s/4}\d s}{\d\rho_R(s)} \leq \pi\sqrt R \sup_{s\geq 0} s^{1/2} e^{-s/4} \leq \pi\sqrt{\frac{2R}{e}}. 
    \end{align}
    Let $m=\lfloor k/2\rfloor\geq 1$ and let $p_t(s):=K_{m+1,R}(t,s)/K_{m+1,R}(t,t)$ be the polynomial attaining $\lambda_{\rho_R,m+1}(t)$, as in \eq{test-polynomial}.
    By definition, $\deg(p_t)\leq m$ and $p_t(t)=1$, so that $\deg(p_t^2)\leq 2m\leq k$.
    Since $p_t^2$ is admissible, by \eq{scaled-beta-exponential} and \eq{interior-domination-2}, the Christoffel function can be bounded as
    \begin{align}
        \lambda_{\sigma,k+1}(t)
        &\leq \int_0^\infty p_t(s)^4 f_\sigma(s)\d s
        \overset{\eq{scaled-beta-exponential}}{\leq} \int_0^\infty p_t(s)^4 se^{-s/4}\d s \\\label{eq:318}
        &\overset{\eq{interior-domination-2}}{\leq} \pi\sqrt{\frac{2R}{e}} \int_0^R s\cdot p_t(s)^4\d\rho_R(s) + \int_R^\infty p_t(s)^4 s e^{-s/4}\d s.
    \end{align}
    We bound each integral as follows.
    By \eq{polynomial-growth} (applied with $\ell=m+1$), $|p_t(s)|\leq 8(4s/R)^m$ for $s\geq R$, and hence, substituting $s=Ru$,
    \begin{align}
        \mathcal I_2(t) := \int_R^\infty p_t(s)^4 s e^{-s/4}\d s
        \leq
        8^4\int_R^\infty s\left(\frac{4s}{R}\right)^{4m} e^{-s/4}\d s
        \leq 2^{10}R^2\int_1^\infty (4u)^{4m+1} e^{-uR/4}\d u.
    \end{align}
    Using the fact that $\ln(x)\leq x-1$ for $x\geq 1$, we have, provided $R/4>4(4m+1)$,
    \begin{align}
        \int_1^\infty (4u)^{4m+1} e^{-uR/4}\d u \leq \int_1^\infty e^{-uR/4 + (4m+1)(4u-1)} \d u \leq \frac{e^{-4m-1}}{R/4-4(4m+1)}.
    \end{align}
    Choosing $R=64(m+1)$ as required in \lem{localization}, we have
    \begin{align}\label{eq:319}
        \mathcal I_2(t)\leq \frac{2^{10}R^2 e^{-4m-1}}{12}\leq 32R^2 e^{-4m}=e^{-\Omega(k)}.
    \end{align}

    It remains to bound the first integral in \eq{318}.
    Let $t=R\sin^2(\phi/2)$ and $s=R\sin^2(\theta/2)$ with $\theta,\phi\in[0,\pi]$, so that $\d\rho_R(s)=\d\theta/\pi$.
    By \lem{localization} (applied with $\ell=m+1$ and the roles of the two points as indicated there; recall $t\in[0,1]$), for some absolute constant $C_3$,
    \begin{align}
        |p_t(s)| \leq \frac{C_3}{1+(m+1)|\theta-\phi|}\leq \frac{C_3}{1+m|\theta-\phi|}.
    \end{align}
    We also need a bound on $s$ in terms of $t$ and $|\theta-\phi|$.
    Since $\sin^2(\theta/2)-\sin^2(\phi/2)=\sin\frac{\theta-\phi}{2}\sin\frac{\theta+\phi}{2}$ and $\sin\frac{\theta+\phi}{2}\leq\sin\frac{\phi}{2}+\frac{|\theta-\phi|}{2}$, using $R\sin(\phi/2)=\sqrt{tR}$ we obtain
    \begin{align}\label{eq:s-vs-t}
        s\leq t+\frac{\sqrt{tR}}{2}|\theta-\phi|+\frac{R}{4}|\theta-\phi|^2.
    \end{align}
    Changing variables to $u=\theta-\phi$ and extending the range of integration to the whole real line,
    \begin{align}
        \mathcal I_1(t) := \int_0^R s\cdot p_t(s)^4\d\rho_R(s)
        &\leq \frac{C_3^4}{\pi}\int_{-\infty}^{\infty} \frac{t+\frac{\sqrt{tR}}{2}|u|+\frac R4 u^2}{(1+m|u|)^4} \d u.
    \end{align}
    Using $\int_{-\infty}^\infty \frac{\d u}{(1+m|u|)^4}=\frac{2}{3m}$, $\int_{-\infty}^\infty \frac{|u|\d u}{(1+m|u|)^4}=\frac{1}{3m^2}$, and $\int_{-\infty}^\infty \frac{u^2\d u}{(1+m|u|)^4}=\frac{2}{3m^3}$, we get
    \begin{align}\label{eq:328}
        \mathcal I_1(t) \leq \frac{C_3^4}{3\pi}\left(\frac{2t}{m} + \frac{\sqrt{tR}}{2m^2}+\frac{R}{2m^3}\right).
    \end{align}
    Finally, by \eq{319} and \eq{328}, recalling that $R=64(m+1)\leq 128m$ and $m=\lfloor k/2\rfloor\geq k/4$ for $k\geq 2$,
    \begin{align}
        \lambda_{\sigma,k+1}(t) \leq \pi\sqrt{\frac{2R}{e}} \cdot \mathcal I_1(t) + \mathcal I_2(t)
        = O\left(\frac{t}{\sqrt m}+\frac{\sqrt t}{m}+\frac{1}{m^{3/2}}\right)+e^{-\Omega(k)}
        = O\left(\frac{t}{k^{1/2}}+\frac{1}{k^{3/2}}\right),
    \end{align}
    where in the last step we used $\sqrt t/m\leq (t+1/m)/(2\sqrt m)$ by the inequality of arithmetic and geometric means.
\end{proof}

The $t$-dependence of the bound for $t\geq 1$ can be improved by using a Christoffel function of the exponential measure of lower degree, but we do not pursue the improvement here since the current form is sufficient for proving \thm{one-local-invariance}.

\begin{theorem}\label{thm:one-local-invariance}
    For integers $n\geq 2$ and $k\geq 2$ and $\epsilon\in[0,1]$, let $\mu(k,\epsilon)$ be the distribution of any $1$-locally invariant $\epsilon$-approximate $n$-qubit unitary $k$-design, and for fixed $x,y\in\bit^n$ let $\nu(k,\epsilon)$ be the distribution of $|\bra{x} U\ket{y}|^2$ for $U\sim\mu(k,\epsilon)$.
    Let $\nu$ be the distribution of $|\bra{x}U\ket{y}|^2$ for Haar random $U$ on $n$ qubits.
    Then there exists an absolute constant $C$ such that 
    \begin{align}
        d_{\mathrm{TV}}(\nu(k,\epsilon),\nu) \leq C\cdot (k^{-1/2}+\epsilon).
    \end{align}
\end{theorem}
\begin{proof}
    Let $\mu(k,\epsilon)$ be locally invariant on the $i$th qubit, and let $N=2^n\geq 4$.
    Let $\sigma(k,\epsilon)$ be the distribution of $N\cdot \sum_{z\in\{0,e_i\}}|\bra{x\oplus z}U\ket{y}|^2$ for $U\sim\mu(k,\epsilon)$, i.e., $N$ times the probability of the block $B(x)=\{x,x\oplus e_i\}$, and let $\sigma$ be the distribution of the same random variable for $U\sim\Haar(N)$, so that $Z/N\sim\Beta(2,N-2)$ for $Z\sim\sigma$ by \lem{block-factorization}.
    For $Y\sim\Uniform[0,1]=\Beta(1,1)$, let $\M\sigma$ be the distribution of $XY$ for $X\sim\sigma$ and $\M\sigma(k,\epsilon)$ be that of $XY$ for $X\sim\sigma(k,\epsilon)$, as in \lem{total-variation-df}.
    By \lem{block-factorization} with $m=1$, $Z/N\sim\nu(k,\epsilon)$ for $Z\sim\M\sigma(k,\epsilon)$ and $Z/N\sim\nu$ for $Z\sim\M\sigma$.
    Moreover, $\sigma$ has a positive density on $(0,N)$, and by \thm{polynomial-relative-error} applied to the block $B(x)$ (note that $p(NX)$ is a polynomial in $X$ of the same degree as $p$), the pair $\sigma,\sigma(k,\epsilon)$ satisfies \eq{pointwise-bound-assumption}; hence \lem{pointwise-df-to-christoffel} applies.
    
    By \lem{total-variation-df} and \lem{pointwise-df-to-christoffel}, together with $F_{\sigma(k,\epsilon)}(b)\leq(1+\epsilon)(F_\sigma(b)+\lambda_{\sigma,k+1}(b))$ from \eq{pointwise-bound-1} and $F_\sigma(b)\leq b^2/2$, for $b\in(0,1]$ we have
    \begin{align}
        &d_{\mathrm{TV}}(\M\sigma(k,\epsilon),\M\sigma)
        \leq
        F_{\sigma(k,\epsilon)}(b) + F_{\sigma}(b) + \int_b^\infty \frac1t|F_{\sigma(k,\epsilon)}(t)-F_\sigma(t)|\d t \\\label{eq:315}
        &\qquad\leq (1+\epsilon/2)b^2 + (1+\epsilon)\lambda_{\sigma,k+1}(b) + \epsilon\int_b^\infty \frac1t\min\{F_\sigma(t),\bar F_\sigma(t)\} \d t + (1+\epsilon)\int_b^\infty \frac{\lambda_{\sigma,k+1}(t)}{t}\d t.
    \end{align}
    We bound each term in \eq{315}.
    For the first integral, 
    \begin{align}
        \mathcal I_1(b) := \int_b^\infty \frac1t\min\{F_\sigma(t),\bar F_\sigma(t)\} \d t
        &\leq 
        \int_0^\infty \frac1t\min\{F_\sigma(t),\bar F_\sigma(t)\} \d t \\
        &\leq 
        \int_0^1 \frac1t F_\sigma(t) \d t 
        +  
        \int_1^\infty \frac1t \bar F_\sigma(t) \d t.
    \end{align}
    Using the inequality that $F_\sigma(t)\leq t^2/2$, we have 
    \begin{align}
        \int_0^1 \frac1t F_\sigma(t) \d t \leq \frac14.
    \end{align}
    For the second term, recalling that $\sigma=N\Beta(2,N-2)$, we have $\Exp_{Z\sim\sigma}[Z]=2$, and 
    \begin{align}
        \int_1^\infty \frac1t \bar F_\sigma(t) \d t
        \leq 
        \int_1^\infty \bar F_\sigma(t) \d t \leq \int_0^\infty \bar F_\sigma(t) \d t = \Exp_{Z\sim\sigma}[Z] = 2.
    \end{align}
    Here we use the identity that for non-negative random variable $X$ whose complementary distribution function is $\bar F$, $\Exp[X]=\int_0^\infty \bar F(t)\d t$.
    Combining these inequalities, we have $\mathcal I_1(b)\leq 9/4$.

    For the second integral, by \lem{pointwise-christoffel}, 
    \begin{align}
        \mathcal I_2(b) := \int_b^\infty \frac{\lambda_{\sigma,k+1}(t)}{t}\d t
        &\leq \frac{C}{\sqrt k}\left[\int_b^1 \left(1+\frac{1}{kt}\right)\d t
        +\int_1^\infty\frac1{t^3} \d t 
        \right] \\
        &\leq \frac{C}{\sqrt k}\left[(1-b)+\frac{\ln(1/b)}k
        +\frac12 
        \right].
    \end{align}
    Choosing $b=k^{-1}$, we have $\mathcal I_2(b)\leq \frac{C}{\sqrt k}\left(\frac32+\frac{\ln k}{k}\right)\leq\frac{2C}{\sqrt k}$ and, by \lem{pointwise-christoffel}, $\lambda_{\sigma,k+1}(b)\leq 2Ck^{-3/2}$.
    Therefore, using $\epsilon\leq 1$ and $k\geq 2$,
    \begin{align}
        d_{\mathrm{TV}}(\M\sigma(k,\epsilon),\M\sigma)
        &\leq (1+\epsilon/2)b^2 + (1+\epsilon)\lambda_{\sigma,k+1}(b) + \epsilon\mathcal I_1(b) + (1+\epsilon)\mathcal I_2(b) \\
        &\leq
        \frac{3}{2k^2} + \frac{4C}{k^{3/2}} + \frac{9\epsilon}{4} + \frac{4C}{\sqrt k}\leq \frac{6C+1}{\sqrt k} + \frac{9\epsilon}{4}.
    \end{align}
    We have bounded the total variation distance for the random variables multiplied by $N$.
    Since scaling does not change the distance, we conclude the theorem.
\end{proof}

We close this section by combining \thm{one-local-invariance} with a recent result on strong designs, which turns the theorem into a quantitative Porter--Thomas statement for one of perhaps the most commonly studied random circuit ensembles: one-dimensional brickwork circuits.
In this ensemble, independent Haar random two-qubit gates are applied in alternating layers to neighboring pairs of qubits arranged on a line, the gates of each layer acting on disjoint pairs.
Schuster, Ma, Lombardi, Brand\~ao, and Huang \cite[Lemma~3 and Appendix~A.4]{schusterStrongRandomUnitaries2025} show, by combining the spectral gap bound of Chen et al.~\cite{chenIncompressibilitySpectralGaps2025} with their translation from additive to relative error for strong designs, that brickwork circuits of depth $\Omega(\log^7(p+q)\,(n(p+q)+\log(1/\epsilon)))$ on $n$ qubits form strong $\epsilon$-approximate $(p,q)$-designs with relative error whenever $\epsilon\geq 2(p+q)^2/2^n$.

Such circuits are moreover locally invariant without any modification: the final layer applies independent Haar random gates to disjoint pairs of qubits, and the left invariance of the Haar measure on $\mathrm U(4)$ makes the circuit $1$-locally invariant on every qubit acted on by the final layer (in fact $2$-locally invariant on every such pair).
We therefore obtain the following corollary.

\begin{corollary}[Porter--Thomas statistics of brickwork circuits]\label{cor:brickwork}
    For $n\geq 2$ and $d\geq 1$, let $U$ be a one-dimensional brickwork circuit of depth $d$ on $n$ qubits with independent Haar random two-qubit gates, and for fixed $x,y\in\bit^n$ let $\nu_d$ be the distribution of $|\bra{x}U\ket{y}|^2$ and $\nu=\Beta(1,2^n-1)$ the Porter--Thomas distribution.
    There exist absolute constants $c_1,c_2>0$ such that the following hold.
    \begin{enumerate}
        \item For every integer $k\geq 2$ and every $\epsilon\in[8k^2/2^n,1]$, if $d\geq c_1\log^7(2k)\,(nk+\log(1/\epsilon))$, then $d_{\mathrm{TV}}(\nu_d,\nu)\leq C(k^{-1/2}+\epsilon)$, where $C$ is the constant of \thm{one-local-invariance}.
        \item For every $\gamma\in[c_2 2^{-n/5},1]$, if $d\geq c_2\,n\gamma^{-2}\log^7(2/\gamma)$, then $d_{\mathrm{TV}}(\nu_d,\nu)\leq\gamma$.
    \end{enumerate}
\end{corollary}
In particular, a brickwork circuit of linear depth $O(n)$ reproduces the Porter--Thomas distribution of each output probability up to any fixed total variation distance, and for fixed $c$, depth $O(n^{2c+1}\log^7 n)$ suffices for total variation distance $n^{-c}$.
\begin{proof}
    Fix a gate of the final layer, say acting on the neighboring qubits $i$ and $i+1$; it is distributed as $G\sim\Haar(4)$ independently of the remaining gates, so we may write $U=(G\otimes I)U'$ with $U'$ consisting of the remaining gates.
    For any fixed $V\in\mathrm U(2)$ acting on qubit $i$, the left invariance of the Haar measure on $\mathrm U(4)$ gives $(V\otimes I_2)G\sim G$, and hence $(V\otimes I)U=((V\otimes I_2)G\otimes I)U'\sim U$ by independence.
    Thus the distribution of $U$ is $1$-locally invariant with invariant qubit set $\{i\}$.

    For the first statement, we apply \cite[Lemma~3]{schusterStrongRandomUnitaries2025} with $p=q=k$, so that $p+q=2k$ and the condition there reads $\epsilon\geq 2(2k)^2/2^n=8k^2/2^n$: there is an absolute constant $c_1$ such that for $d\geq c_1\log^7(2k)(nk+\log(1/\epsilon))$ the distribution of $U$ is a strong $\epsilon$-approximate $(k,k)$-design with relative error.
    Since this distribution is $1$-locally invariant, \thm{one-local-invariance} yields $d_{\mathrm{TV}}(\nu_d,\nu)\leq C(k^{-1/2}+\epsilon)$.

    For the second statement, we may assume that $C\geq 1$.
    Given $\gamma\in(0,1]$, let $k=\lceil 4C^2/\gamma^2\rceil$ and $\epsilon=\gamma/(2C)$, so that $k\geq 4$, $\epsilon\leq 1/2$, and $C(k^{-1/2}+\epsilon)\leq\gamma$.
    Since $4C^2/\gamma^2\geq 1$, we have $k\leq 8C^2/\gamma^2$, so the condition $\epsilon\geq 8k^2/2^n$ is implied by $\gamma/(2C)\geq 512C^4\gamma^{-4}2^{-n}$, i.e., by $\gamma\geq 4C\cdot 2^{-n/5}$.
    Moreover, $\log(2k)\leq\log(16C^2)+2\log(1/\gamma)\leq c_0\log(2/\gamma)$ with $c_0=\log(16C^2)\geq 4$, where the last inequality uses $\log(1/\gamma)=\log(2/\gamma)-1$ and $\log(2/\gamma)\geq 1$; and $nk+\log(1/\epsilon)\leq 8C^2n\gamma^{-2}+\log(2C/\gamma)\leq (8C^2+2C)\,n\gamma^{-2}$, where we used $\log(2C/\gamma)\leq 2C/\gamma\leq 2Cn/\gamma^2$.
    Hence the depth requirement of the first statement is implied by $d\geq c_2 n\gamma^{-2}\log^7(2/\gamma)$ for $c_2=\max\{c_1c_0^7(8C^2+2C),4C\}$, and the first statement gives $d_{\mathrm{TV}}(\nu_d,\nu)\leq\gamma$.
\end{proof}

The linear dependence on $n$ in \cor{brickwork} is inherent to this route: circuits of independent Haar random two-qubit gates require depth $\Omega(n)$ in any geometry to form strong $2$-designs with relative error \cite[Proposition~3]{schusterStrongRandomUnitaries2025}, and one-dimensional circuits require depth $\Omega(n)$ to form strong designs with error $\epsilon<1/4$ under any of the error notions \cite[Proposition~1]{schusterStrongRandomUnitaries2025}.
Whether circuits of sublinear depth reproduce the Porter--Thomas distribution of a single output probability in total variation is an interesting open question.

\section{From marginal closeness to extractable randomness}\label{sec:random-number}

Now we have seen that the distribution of $|U_{ij}|^2$ for $U$ sampled from the approximate design is close to that of the Haar measure.
As an application of this result, in this section, we show that approximate designs reproduce the expected Shannon entropy of a Haar random unitary, and that the Shannon entropy of a single sample concentrates around this value. 
In the following analysis, we call the quantity $S=\sum_{x\in\bit^n} P_x \log(1/P_x)$, where $P_x=|\bra{x}U\ket{0}|^2$ and $0\log(1/0)=0$, the Shannon entropy of the output distribution of $U$.
We compare the expected Shannon entropy when $U$ is sampled from a Haar random unitary and when it is sampled from a design. 

Our main contribution in this section is to show that the expected Shannon entropy of the output distribution of an approximate design matches its Haar value up to an additive error that is independent of the dimension.
Previously, such a guarantee was established for Haar random unitaries \cite{aaronsonCertifiedRandomnessQuantum2023}. 
For certifying random number generation, if the adversary can coordinate different samples arbitrarily, then quantifying the difference of the expected Shannon entropy is insufficient to conclude entropy accumulation without introducing a verification that robustly certifies the absence of correlation.
Here, we do not provide a solution to the certification problem; instead, we show that when there is a quantum computer that independently and perfectly samples from a unitary design, then the expected entropy is high.
Since such a device produces independent rounds, the asymptotic equipartition property \cite{tomamichelFullyQuantumAsymptotic2009}, the special case of entropy accumulation \cite{metgerGeneralisedEntropyAccumulation2024} for independent rounds, shows that repeated sampling accumulates entropy at the rate of the expected Shannon entropy per round, which can then be extracted (see \sec{overview-entropy}). 

In \sec{expectation-shannon}, we prove that the mean of the Shannon entropy is close to the ideal mean in absolute difference $O(\eta^{1/2})$, where $\eta$ is the Kolmogorov distance we obtain in \sec{convergence-upper-bound} (for a generic design) and in \sec{local-translation-invariant} (for a locally invariant design). 
The asymptotic equipartition property only requires that the expected Shannon entropy per round is high enough, so this is already sufficient for an honest device to generate randomness. 

We then prove a stronger statement: not only is the expectation close, but the Shannon entropy of a single sample from the design is also concentrated around its Haar value.
In \sec{variance-shannon}, we perform a second-moment analysis and show that the standard deviation of the Shannon entropy is within $O(\delta^{1/4})$ of its Haar value, where $\delta$ is the Kolmogorov distance for linear combinations of two output probabilities (\dfn{pair-distance}), which we bound by $O(k^{-1/2}+\epsilon)$ for every $\epsilon$-approximate $k$-design.
Combining the two estimates with Chebyshev's inequality in \sec{concentration-shannon}, we show that the Shannon entropy deviates from its Haar value $n-\Theta(1)$ by more than $t+O(\eta^{1/2})$ with probability at most $O\big((2^{-n}+\delta^{1/2})/t^2\big)$.

\subsection{Expected entropy from marginal closeness}\label{sec:expectation-shannon}

In this section, we prove the mean of the Shannon entropy is close to the ideal mean in absolute difference $O(\eta^{1/2})$, when the Kolmogorov distance between the laws of each $P_x$ under the design and under the Haar measure is at most $\eta$. 

To show this, we consider a global cutoff $B/N$ for a cutoff parameter $B$ to be optimized later and bound the difference for the cases $P_x\leq B/N$ and $P_x>B/N$ for every $x\in\bit^n$. 
In the former case, since each truncated summand is bounded by $O(B/N)$ and has interval superlevel sets, we use the closeness in Kolmogorov distance to show that the two expectations differ by $O(\eta B)$. 
In the latter, since the second moments of each $P_x$ under both distributions---either from a Haar random unitary or from the design---are bounded by $O(1/N^2)$, we can conclude this part is insignificant to the difference when $B$ is large enough. 
Concretely, we get a bound that is $O(1/B)$. 
Choosing $B=\Theta(\eta^{-1/2})$ to balance the contributions from these two cases, we get an $O(\eta^{1/2})$ bound. 

In the previous sections, we have bounded closeness of the distributions both in Kolmogorov distance and in total variation distance.
The latter is strictly stronger: any two distributions that are $\eta$-close in total variation distance are also $\eta$-close in Kolmogorov distance, but the converse does not hold in general. 
Thus here we prove the following lemma that uses only the weaker notion of closeness to bound the closeness of the expectation values, and the same bound holds if closeness in total variation distance can be justified. 

\begin{lemma}\label{lem:expectation-ks}
    Let $\mu,\nu$ be the distributions of two random variables in $[0,1]$ whose Kolmogorov distance is at most $\eta$. 
    Then for $M\geq 0$ and measurable function $f:[0,1]\to[0,M]$ such that every superlevel set $\{x\in[0,1]:f(x)>t\}$, $t\geq 0$, is an interval, it holds that 
    \begin{align}
        \left|\Exp_{X\sim\mu}[f(X)] - \Exp_{X\sim\nu}[f(X)]\right| \leq M\cdot \min\{2\eta, 1\}.
    \end{align}
\end{lemma}
\begin{proof}
    Since the Kolmogorov distance is at most $\eta$, we have for every interval $I\subseteq[0,1]$, 
    \begin{align}
        \left|\Pr_{X\sim\mu}[X\in I]-\Pr_{X\sim\nu}[X\in I]\right| 
        &\leq 2\eta,
    \end{align}
    regardless of whether $I$ is open or closed at either end, since the bound $|\Pr_\mu[X\leq s]-\Pr_\nu[X\leq s]|\leq\eta$ also holds for the left limits $\Pr[X<s]=\lim_{s'\uparrow s}\Pr[X\leq s']$.
    Since $f$ is measurable, we have 
    \begin{align}\label{eq:expectation-df}
        \Exp_{X\sim\mu}[f(X)] 
        = \int_0^M \d t \Pr_{X\sim\mu}[f(X)>t]
        = \int_0^M \d t \Pr_{X\sim\mu}[X\in I_t],
    \end{align}
    where $I_t:=\{x\in[0,1]:f(x)>t\}$.
    Since the equality in \eq{expectation-df} holds for every distribution, it holds for $\nu$ in place of $\mu$.
    Therefore, 
    \begin{align}
        \left|\Exp_{X\sim\mu}[f(X)] - \Exp_{X\sim\nu}[f(X)]\right|
        &\leq \int_0^M \left| \Pr_{X\sim\mu}[f(X)>t] - \Pr_{X\sim\nu}[f(X)>t] \right| \d t \\
        &\leq \min\{2\eta,1\} \int_0^M \d t
        = M\cdot \min\{2\eta, 1\}.
    \end{align}
\end{proof}

To use \lem{expectation-ks} to prove closeness of the expected Shannon entropy, we note that the function $f(x)=x\log(1/x)$ on $[0,1]$ (with $f(0)=0$) has a single maximum at $x_*=1/e$, is increasing on $(0,x_*]$ and decreasing on $[x_*,1]$, and satisfies $f(0)=f(1)=0$; hence every superlevel set of $f$ is an interval. The same holds for the truncated and rescaled functions $\phi_B^+$ defined in the proof below, which are of the form $x\mapsto cf(x/c)\mathbf 1[x\leq c]$ with $c=B/N$. 

\begin{theorem}\label{thm:shannon-entropy-mean}
    For $\epsilon\in[0,1]$ and $k\geq 1$, let $\mu(k,\epsilon)$ be the distribution of any $n$-qubit $\epsilon$-approximate $k$-design, and $\mu$ be the Haar measure.
    Let $S=S(U)$ be the Shannon entropy of the output distribution of $U$, and let $\eta=\eta(k,\epsilon)$ be the maximum over $x\in\{0,1\}^n$ of the Kolmogorov distance between the law of $P_x$ for $U\sim\mu(k,\epsilon)$ and that for $U\sim\mu$.
    Then
    \begin{align}
        \left| \Exp_{U\sim\mu(k,\epsilon)}[S] - \Exp_{U\sim\mu}[S]\right| \leq \frac{4\sqrt{3\eta}}{e\ln 2}.
    \end{align}
\end{theorem}
\begin{proof}
    By definition, we have $S=\sum_{x\in\bit^n}-P_x\log(P_x)$, where $P_x=|\bra{x}U\ket{0}|^2$. 
    If $\eta=0$, the law of each $P_x$ is the same under the two measures, hence so is $\Exp[S]=\sum_x\Exp[-P_x\log P_x]$; we therefore assume $\eta>0$.
    Now define 
    \begin{align}\label{eq:shannon-entropy-decomposition}
        \phi_B^+(x) := 
        \begin{cases}
        x\log\left(\frac{B}{Nx}\right) & \text{if } B/N\geq x \\
        0 & \text{otherwise,} \\
        \end{cases}
        \qquad
        \phi_B^-(x) := 
        \begin{cases}
        -x\log\left(\frac{B}{Nx}\right) & \text{if } B/N\leq x \\
        0 & \text{otherwise.} \\
        \end{cases}
    \end{align} 
    Then by definition,
    \begin{align}\label{eq:entropy-deficit}
        S(P) = \log(N/B) + \sum_{x\in\bit^n} \phi_B^+(P_x) - \sum_{x\in\bit^n} \phi_B^-(P_x).
    \end{align}
    Since the first term is distribution-independent, it suffices to bound the closeness of the remaining two terms in expectation. 
    Marginal closeness in Kolmogorov distance gives 
    \begin{align}
        &\left|\Exp_{U\sim\mu(k,\epsilon)}\left[\sum_{x\in\bit^n} \phi_B^+(P_x)\right]- \Exp_{U\sim\mu}\left[\sum_{x\in\bit^n} \phi_B^+(P_x)\right]\right| \\
        &\qquad\leq 
        \sum_{x\in\bit^n}\left|\Exp_{U\sim\mu(k,\epsilon)}\left[\phi_B^+(P_x)\right]- \Exp_{U\sim\mu}\left[\phi_B^+(P_x)\right]\right| \leq 
        N \cdot 2\|\phi_B^+\|_\infty \cdot \eta 
        \leq \frac{2\eta B}{e\ln 2},
    \end{align}
    where the last inequality follows from \lem{expectation-ks} and the fact that $\|\phi_B^+\|_\infty \leq \frac{B}{N}\max_{y\in(0,1]} y\log(1/y) = \frac{B}{Ne\ln 2}$.

    For the second part, since $\ln x\leq x/e$ for $x\geq 1$, we have 
    $\phi_B^-(x)\leq \frac{Nx^2}{eB\ln 2}$, and therefore, for $\tau\in\{\mu(k,\epsilon),\mu\}$, 
    \begin{align}
        0\leq \sum_{x\in\bit^n} \Exp_{U\sim\tau}[\phi_B^{-}(P_x)] 
        &\leq \frac{N}{eB\ln 2}\Exp_{U\sim\tau}\left[\sum_{x\in\bit^n} P_x^2\right] 
        \leq \frac{2(1+\epsilon)}{eB\ln 2},
    \end{align}
    where the last inequality holds since $\Exp_{U\sim\mu}[P_x^2]=\frac{2}{N(N+1)}$ for every $x$, and for $\mu(k,\epsilon)$ the second moment is larger by a multiplicative factor of at most $1+\epsilon$ by \thm{polynomial-relative-error} (applied to the entry $(x,0)$ with $p(u)=u$, which requires $k\geq 1$). 

    Therefore we have 
    \begin{align}
        \left| \Exp_{U\sim\mu(k,\epsilon)}[S] - \Exp_{U\sim\mu}[S]\right| \leq
        \frac{2}{e\ln 2} \left(\eta B + \frac{2+\epsilon}{B}\right) \leq \frac{4\sqrt{(2+\epsilon)\eta}}{e\ln 2},
    \end{align}
    by choosing $B=\sqrt{(2+\epsilon)/\eta}$, which is at most $\frac{4\sqrt{3\eta}}{e\ln 2}$ for $\epsilon\leq 1$.
\end{proof}

\thm{shannon-entropy-mean} implies that the mean of the Shannon entropy of a design approximates that of the Haar random unitary well.

We also note that in \thm{shannon-entropy-mean}, we bound closeness in the expected Shannon entropy using elementwise closeness in Kolmogorov distance, without introducing a multiplicative loss proportional to the dimension $N=2^n$. 
A direct coordinatewise estimate would incur such a loss in general. 
Here we can avoid the loss by exploiting the identity in \eq{entropy-deficit}, which expresses the entropy deficit as an average of functions of the normalized coordinates. 
Marginal closeness in Kolmogorov distance bounds the difference of the expected truncated entropy contributions by $O(\eta B/N)$ per output probability. Summing over the $N$ outputs therefore gives an $O(\eta B)$ bound.
For the contribution of the tail, we then rely on the moment closeness inherent in the unitary design.

\subsection{Entropy fluctuations from pairwise closeness}\label{sec:variance-shannon}

In this section, we show that the standard deviation of the Shannon entropy of an $\epsilon$-approximate unitary $k$-design is close to that of a Haar random unitary acting on the same number of qubits.
The proof of \thm{shannon-entropy-mean} used only the law of each single output probability.
The variance of $S$ involves pairs of output probabilities, and the quantity that replaces the Kolmogorov distance $\eta$ is the following.

\begin{definition}\label{dfn:pair-distance}
    For an $n$-qubit distribution $\mu(k,\epsilon)$, let $\delta=\delta(k,\epsilon)$ be the supremum, over $x,x'\in\bit^n$ (not necessarily distinct) and $s,t\geq 0$ not both zero, of the Kolmogorov distance between the law of $sP_x+tP_{x'}$ for $U\sim\mu(k,\epsilon)$ and the law of the same random variable for $U\sim\mu$, where $\mu$ denotes the Haar measure on $n$ qubits.
\end{definition}

Since the Kolmogorov distance is invariant under multiplication of both random variables by the same positive constant, $\delta$ is also the supremum over the normalized projections $(sP_x+tP_{x'})/(s+t)$, which take values in $[0,1]$.
Taking $t=0$ shows that $\delta\geq\eta$, where $\eta$ is the single-entry Kolmogorov distance of \thm{shannon-entropy-mean}.
We first show that $\delta=O(k^{-1/2}+\epsilon)$ for every $\epsilon$-approximate $k$-design, by combining the pointwise comparison of \lem{pointwise-df-to-christoffel} with uniform bounds on the Christoffel functions of $\Beta(1,N-1)$ and $\Beta(2,N-2)$, which we extract from \thm{jacobi-weight-upper-bound} and \lem{pointwise-christoffel}.

\begin{lemma}[Uniform Christoffel bounds]\label{lem:uniform-christoffel}
    There is an absolute constant $C_0$ such that the following hold, where $\mu_{a,D}$ denotes the law of $\Beta(a,D)$ as in \sec{jacobi-gauss}.
    \begin{enumerate}
        \item For all integers $D\geq 1$ and $\ell\geq 1$ and all $x\in[0,1]$, $\lambda_{\mu_{1,D},\ell}(x)\leq C_0\ell^{-1/2}$.
        \item For all integers $D\geq 2$ and $\ell\geq 3$ and all $x\in[0,1]$, $\lambda_{\mu_{2,D},\ell}(x)\leq C_0\ell^{-1/2}$.
    \end{enumerate}
\end{lemma}
\begin{proof}
    For the first statement, we use Step~2 of the proof of \thm{jacobi-weight-upper-bound}, which applies to every $x\in[0,1]$ and not only to the quadrature nodes: the measure domination $\mu_{1,D}\leq\e\,\omega_D$ and the variational definition \eq{christoffel-mu-d} give $\lambda_{\mu_{1,D},\ell}(x)\leq\e\,\lambda_{\omega_D,\ell}(x)=\e\,\lambda_{\omega,\ell}(Dx)$, where the equality is the scaling identity \eq{scaling}, and \lem{exponential} bounds the right-hand side by $\e C\ell^{-1/2}$.

    For the second statement, let $N:=D+2\geq 4$, and let $\sigma$ be the law of $NX$ for $X\sim\Beta(2,D)=\Beta(2,N-2)$, as in \lem{pointwise-christoffel}.
    The map $p\mapsto p(N\,\cdot)$ is a bijection between the polynomials of degree at most $\ell-1$ with value $1$ at $Nx$ and those with value $1$ at $x$, and $\Exp_{Z\sim\sigma}[p(Z)^2]=\Exp_{X\sim\Beta(2,D)}[p(NX)^2]$; hence $\lambda_{\mu_{2,D},\ell}(x)=\lambda_{\sigma,\ell}(Nx)$.
    By \lem{pointwise-christoffel} applied with $k=\ell-1\geq 2$, and using $\ell/(\ell-1)\leq 3/2$, we have $\lambda_{\sigma,\ell}(t)\leq C(\ell-1)^{-1/2}(t+(\ell-1)^{-1})\leq\frac32C(\ell-1)^{-1/2}\leq (\frac32)^{3/2}C\ell^{-1/2}$ for $t\in[0,1]$ and $\lambda_{\sigma,\ell}(t)\leq C(\ell-1)^{-1/2}t^{-2}\leq (\frac32)^{1/2}C\ell^{-1/2}$ for $t\geq 1$.
\end{proof}

\begin{lemma}[Pair projections]\label{lem:pair-projection}
    For integers $n\geq 2$ and $k\geq 2$ and $\epsilon\in[0,1]$, let $\mu(k,\epsilon)$ be any $n$-qubit $\epsilon$-approximate $k$-design.
    Then $\delta(k,\epsilon)\leq C_1(k^{-1/2}+\epsilon)$ for an absolute constant $C_1$.
\end{lemma}
\begin{proof}
    Fix $x,x'\in\bit^n$ and $s,t\geq 0$ not both zero.
    By symmetry we may assume $s\geq t$, so that $s>0$, and by scale invariance we may replace $sP_x+tP_{x'}$ by $W:=P_x+\theta P_{x'}$ with $\theta:=t/s\in[0,1]$.
    If $x=x'$, then $W=(1+\theta)P_x$, and by scale invariance again we may take $\theta=0$.
    Let $\sigma$ be the law of $W$ for $U\sim\mu$ and $\sigma(k,\epsilon)$ its law for $U\sim\mu(k,\epsilon)$; both are supported in $[0,1]$.

    For every real polynomial $p$ of degree at most $k$, $p(u+\theta v)$ is a polynomial of total degree at most $k$ in the two output probabilities $P_x,P_{x'}$ (in the single probability $P_x$ if $x=x'$), so \rmk{several-blocks} (or \thm{polynomial-relative-error} if $x=x'$) gives
    \begin{align}
        (1-\epsilon)\Exp_{W\sim\sigma}[p(W)^2]\leq\Exp_{W\sim\sigma(k,\epsilon)}[p(W)^2]\leq(1+\epsilon)\Exp_{W\sim\sigma}[p(W)^2],
    \end{align}
    which is the hypothesis \eq{pointwise-bound-assumption} of \lem{pointwise-df-to-christoffel}.

    If $x=x'$, then $\theta=0$ and $\sigma=\mu_{1,N-1}$.
    Let $x\neq x'$.
    For a Haar random unitary, $(P_x,P_{x'},1-P_x-P_{x'})\sim\Dir(1,1,N-2)$, and $(P_x,P_{x'})$ has the density $g(u,v)=(N-1)(N-2)(1-u-v)^{N-3}$ on the triangle $\{u,v\geq 0,\ u+v\leq 1\}$ (\sec{preliminaries-beta}).
    For $\theta\in(0,1)$ and $w\in(0,1)$, substituting $u=w-\theta v$, the density of $W$ is
    \begin{align}\label{eq:pair-density}
        f_\theta(w)=(N-1)(N-2)\int_0^{v_{\max}(w)}\big(1-w-(1-\theta)v\big)^{N-3}\d v,
        \qquad v_{\max}(w):=\min\Big\{\frac{w}{\theta},\frac{1-w}{1-\theta}\Big\},
    \end{align}
    which is positive because the integrand is positive on $[0,v_{\max}(w))$.
    Since the integrand is at most $(1-w)^{N-3}$ and $v_{\max}(w)\leq w/\theta$,
    \begin{align}
        f_\theta(w)
        &\leq (N-1)(N-2)\int_0^{w/\theta} (1-w)^{N-3}\d v \\\label{eq:pair-domination-2}
        &=\frac{(N-1)(N-2)}{\theta}\,w(1-w)^{N-3}=\frac1\theta f_{\mu_{2,N-2}}(w),
    \end{align}
    where $f_{\mu_{2,N-2}}$ is the density of $\Beta(2,N-2)$; on the other hand, extending the range of integration to $[0,(1-w)/(1-\theta)]$ and integrating,
    \begin{align}\label{eq:pair-domination-1}
        f_\theta(w)
        &\leq(N-1)(N-2)\int_0^{\frac{1-w}{1-\theta}}\big(1-w-(1-\theta)v\big)^{N-3}\d v \\
        &=\frac{(N-1)(1-w)^{N-2}}{1-\theta}=\frac{1}{1-\theta}f_{\mu_{1,N-1}}(w).
    \end{align}
    For $\theta=1$, $W=P_x+P_{x'}$ is the probability of a block of two basis states, so $\sigma=\mu_{2,N-2}$ by the aggregation property (\sec{preliminaries-beta}) and \eq{pair-domination-2} holds with equality; for $\theta=0$, $\sigma=\mu_{1,N-1}$ and \eq{pair-domination-1} holds with equality.
    In all cases, $\sigma$ has support $[0,1]$ and a positive density on $(0,1)$, so \lem{pointwise-df-to-christoffel} applies to the pair $\sigma,\sigma(k,\epsilon)$, and we have the domination of measures $\sigma\leq 2\mu_{2,N-2}$ if $\theta\geq 1/2$ and $\sigma\leq 2\mu_{1,N-1}$ if $\theta\leq 1/2$.

    By the variational definition \eq{christoffel-mu-d}, for $\rho\in\{\mu_{1,N-1},\mu_{2,N-2}\}$, a domination $\sigma\leq 2\rho$ implies $\lambda_{\sigma,k+1}(z)\leq 2\lambda_{\rho,k+1}(z)$ for every $z$, so \lem{uniform-christoffel} gives $\lambda_{\sigma,k+1}(z)\leq 2C_0k^{-1/2}$ for all $z\in[0,1]$ (for both cases, we use $(k+1)^{-1/2}\leq k^{-1/2}$).
    Both $\sigma$ and $\sigma(k,\epsilon)$ are supported in $[0,1]$, so their distribution functions coincide outside $[0,1]$, and for $z\in[0,1]$ \lem{pointwise-df-to-christoffel} gives
    \begin{align}
        |F_{\sigma(k,\epsilon)}(z)-F_\sigma(z)|
        &\leq\epsilon\min\{F_\sigma(z),1-F_\sigma(z)\}+(1+\epsilon)\lambda_{\sigma,k+1}(z) 
        \leq\frac\epsilon2+4C_0k^{-1/2}.
    \end{align}
    Hence the Kolmogorov distance between $\sigma(k,\epsilon)$ and $\sigma$ is at most $4C_0k^{-1/2}+\epsilon/2$, and the lemma follows with $C_1=\max\{4C_0,1/2\}$.
\end{proof}

We now carry out the second-moment analysis.

\begin{theorem}\label{thm:shannon-entropy-variance}
    For integers $n\geq 2$ and $k\geq 2$ and $\epsilon\in[0,1]$, let $\mu(k,\epsilon)$ be the distribution of any $n$-qubit $\epsilon$-approximate $k$-design and $\mu$ the Haar measure, let $S$ be the Shannon entropy of the output distribution of $U$, and let $\delta=\delta(k,\epsilon)$ be as in \dfn{pair-distance}.
    Then there is an absolute constant $C$ such that 
    \begin{align}
        \left|\Sd_{U\sim\mu(k,\epsilon)}[S]-\Sd_{U\sim\mu}[S]\right|\leq 17\,\delta^{1/4}\leq C\,(k^{-1/2}+\epsilon)^{1/4},
    \end{align}
    where $\Sd_{U\sim\tau}[S]$ denotes the standard deviation of $S$ for $U\sim\tau$.
\end{theorem}
\begin{proof}
    The second inequality follows from the first and \lem{pair-projection}; we prove the first.
    Throughout, $c:=1/\ln 2$, the symbol $\tau\in\{\mu(k,\epsilon),\mu\}$ denotes either of the two measures, and $B>0$ is a parameter to be chosen at the end.
    In place of the truncated functions in \eq{shannon-entropy-decomposition}, we use the smooth functions
    \begin{align}
        \phi_B^+(x):=x\ln\Big(1+\frac{B}{Nx}\Big),\qquad \phi_B^-(x):=x\ln\Big(1+\frac{Nx}{B}\Big),\qquad x\in(0,1],
    \end{align}
    with $\phi_B^\pm(0):=0$.
    For $x\in(0,1]$ we have $c\big(\phi_B^+(x)-\phi_B^-(x)\big)=cx\ln\frac{B}{Nx}=x\log\frac1x-x\log\frac NB$, and both sides vanish at $x=0$; summing over $x\in\bit^n$ and using $\sum_xP_x=1$,
    \begin{align}\label{eq:entropy-smooth-decomposition}
        S=\log\frac NB+G_B-R_B,\qquad G_B:=c\sum_{x\in\bit^n}\phi_B^+(P_x),\qquad R_B:=c\sum_{x\in\bit^n}\phi_B^-(P_x).
    \end{align}
    Since the first term is a constant, the triangle inequality for the $L^2$ norm applied to the centered random variables gives
    \begin{align}\label{eq:sd-sandwich}
        \Sd_\tau[G_B]-\Sd_\tau[R_B]\leq\Sd_\tau[S]\leq\Sd_\tau[G_B]+\Sd_\tau[R_B],
    \end{align}
    and therefore
    \begin{align}\label{eq:sd-difference}
        \left|\Sd_{\mu(k,\epsilon)}[S]-\Sd_{\mu}[S]\right|\leq\left|\Sd_{\mu(k,\epsilon)}[G_B]-\Sd_{\mu}[G_B]\right|+\Sd_{\mu(k,\epsilon)}[R_B]+\Sd_{\mu}[R_B].
    \end{align}

    \emph{The tail part $R_B$.}
    Since $\ln(1+u)\leq u$ for $u\geq 0$, we have $0\leq\phi_B^-(x)\leq Nx^2/B$.
    By the Cauchy--Schwarz inequality, $\big(\sum_x\phi_B^-(P_x)\big)^2\leq N\sum_x\phi_B^-(P_x)^2\leq\frac{N^3}{B^2}\sum_xP_x^4$, and hence
    \begin{align}
        \Sd_\tau[R_B]^2\leq\Exp_\tau[R_B^2]\leq\frac{c^2N^3}{B^2}\sum_{x\in\bit^n}\Exp_\tau[P_x^4].
    \end{align}
    For the Haar measure, $P_x\sim\Beta(1,N-1)$ and $\Exp[P_x^4]=\frac{24}{N(N+1)(N+2)(N+3)}\leq\frac{24}{N^4}$ (\sec{preliminaries-beta}); for $\mu(k,\epsilon)$, \thm{polynomial-relative-error} applied to the entry $(x,0)$ with $p(u)=u^2$ (which requires $k\geq 2$) gives $\Exp_{\mu(k,\epsilon)}[P_x^4]\leq(1+\epsilon)\Exp_{\mu}[P_x^4]$.
    Hence, for both measures and $\epsilon\leq 1$,
    \begin{align}\label{eq:sd-rb}
        \Sd_\tau[R_B]\leq\sqrt{\frac{24(1+\epsilon)c^2}{B^2}}\leq\frac{7c}{B}.
    \end{align}

    \emph{A Laplace representation of $G_B$.}
    For $x\in(0,1]$, since $\frac{\d}{\d u}\,x\ln(u+Nx)=\frac{x}{u+Nx}$,
    \begin{align}
        \phi_B^+(x)=\int_0^B\frac{x}{u+Nx}\d u=\frac BN-\frac1N\int_0^B\frac{u}{u+Nx}\d u
        =\frac BN-\frac1N\int_0^\infty w_B(s)\,e^{-sNx}\d s,
    \end{align}
    where $w_B(s):=\int_0^Bu\,e^{-su}\d u$; here we wrote $\frac{1}{u+Nx}=\int_0^\infty e^{-s(u+Nx)}\d s$ and exchanged the order of integration, which is permitted since all integrands are non-negative.
    The last expression also equals $\phi_B^+(0)=0$ at $x=0$, because $\int_0^\infty w_B(s)\d s=\int_0^Bu\int_0^\infty e^{-su}\d s\,\d u=B$.
    Consequently, with
    \begin{align}
        \Lambda:=\sum_{x\in\bit^n}\int_0^\infty w_B(s)\,e^{-sNP_x}\d s\in[0,NB],
    \end{align}
    we have $G_B=c\,(B-\Lambda/N)$ and $\Var_\tau[G_B]=\frac{c^2}{N^2}\Var_\tau[\Lambda]$, and by Tonelli's theorem
    \begin{align}\label{eq:lambda-moments}
        \Exp_\tau[\Lambda]&=\sum_{x\in\bit^n}\int_0^\infty w_B(s)\,\Exp_\tau\big[e^{-NsP_x}\big]\d s, \\
        \Exp_\tau[\Lambda^2]&=\sum_{x,x'\in\bit^n}\int_0^\infty\!\!\int_0^\infty w_B(s)w_B(t)\,\Exp_\tau\big[e^{-N(sP_x+tP_{x'})}\big]\d s\,\d t.
    \end{align}

    \emph{Closeness of the variances of $G_B$.}
    Fix $x,x'\in\bit^n$ and $s,t\geq 0$ not both zero, and let $V:=(sP_x+tP_{x'})/(s+t)\in[0,1]$.
    By \dfn{pair-distance}, the laws of $V$ under $\mu(k,\epsilon)$ and under $\mu$ are within Kolmogorov distance $\delta$, and $f(v):=e^{-N(s+t)v}$ is a decreasing function from $[0,1]$ to $[0,1]$, so every superlevel set of $f$ is an interval; by \lem{expectation-ks} with $M=1$,
    \begin{align}\label{eq:laplace-closeness}
        \Big|\Exp_{\mu(k,\epsilon)}\big[e^{-N(sP_x+tP_{x'})}\big]-\Exp_{\mu}\big[e^{-N(sP_x+tP_{x'})}\big]\Big|\leq 2\delta,
    \end{align}
    and for $s=t=0$ both expectations equal $1$.
    Inserting \eq{laplace-closeness} into \eq{lambda-moments} and using $\int_0^\infty w_B(s)\d s=B$, we obtain $\big|\Exp_{\mu(k,\epsilon)}[\Lambda^2]-\Exp_{\mu}[\Lambda^2]\big|\leq 2\delta N^2B^2$ and $\big|\Exp_{\mu(k,\epsilon)}[\Lambda]-\Exp_{\mu}[\Lambda]\big|\leq 2\delta NB$; since $0\leq\Exp_\tau[\Lambda]\leq NB$, the latter gives $\big|\Exp_{\mu(k,\epsilon)}[\Lambda]^2-\Exp_{\mu}[\Lambda]^2\big|\leq 2\delta NB\cdot 2NB=4\delta N^2B^2$.
    Therefore $|\Var_{\mu(k,\epsilon)}[\Lambda]-\Var_{\mu}[\Lambda]|\leq 6\delta N^2B^2$, i.e.,
    \begin{align}
        \left|\Var_{\mu(k,\epsilon)}[G_B]-\Var_{\mu}[G_B]\right|\leq 6\delta c^2B^2,
    \end{align}
    and since $|\sqrt a-\sqrt b|\leq\sqrt{|a-b|}$ for $a,b\geq 0$,
    \begin{align}\label{eq:sd-gb}
        \left|\Sd_{\mu(k,\epsilon)}[G_B]-\Sd_{\mu}[G_B]\right|\leq\sqrt{6\delta}\,cB.
    \end{align}

    \emph{Conclusion.}
    By \eq{sd-difference}, \eq{sd-rb}, and \eq{sd-gb}, for every $B>0$,
    \begin{align}
        \left|\Sd_{\mu(k,\epsilon)}[S]-\Sd_{\mu}[S]\right|\leq\sqrt{6\delta}\,cB+\frac{14c}{B}.
    \end{align}
    If $\delta=0$, letting $B\to\infty$ shows that the two standard deviations coincide.
    Otherwise we choose $B=\sqrt{14}\,(6\delta)^{-1/4}$, for which both terms equal $c\sqrt{14}\,(6\delta)^{1/4}$, and obtain the bound $2\sqrt{14}\cdot 6^{1/4}\,c\,\delta^{1/4}<17\,\delta^{1/4}$.
\end{proof}

\subsection{Concentration of the Shannon entropy}\label{sec:concentration-shannon}

Since sampling from an $n$-qubit Haar random unitary (i.e., dimension $N=2^n$) yields a random vector that follows $\mathrm{Dir}(1^N)$, the mean and the variance were previously known (see, e.g., \cite{wolpertEstimatingFunctionsProbability1995} for Dirichlet distributions of any shape in general; the closed forms below are the case $\beta=1$ of \cite[Eqs.~(8)--(9)]{nemenmanEntropyInferenceRevisited2002}, which restate the formulas of \cite{wolpertEstimatingFunctionsProbability1995} with small corrections).

\begin{theorem}[{\cite[Theorems~7 and~8]{wolpertEstimatingFunctionsProbability1995}; see also \cite[Eqs.~(8)--(9)]{nemenmanEntropyInferenceRevisited2002}}]\label{thm:wolpert-wolf}
    For a unitary $U$, let $S$ be the Shannon entropy of the output distribution of $U$.
    Then for a Haar random unitary $U\sim\mu$ on $n$ qubits, the mean and the variance of $S$ are
    \begin{align}
        \Exp_{U\sim\mu}[S] = \frac{1}{\ln 2}\left(\psi(N+1) - \psi(2)\right) = n - \Theta(1),
    \end{align}
    and
    \begin{align}
        \Var_{U\sim\mu}[S] =
        \frac{1}{\ln^2 2} \left(\frac{2}{N+1}\left(\frac{\pi^2}{6}-1\right) - \sum_{j=N+1}^\infty \frac{1}{j^2}\right) = \Theta\left(\frac{1}{N}\right).
    \end{align}
\end{theorem}

Combining \thm{shannon-entropy-mean}, \thm{shannon-entropy-variance}, and Chebyshev's inequality, we obtain the following concentration bound.

\begin{theorem}\label{thm:shannon-entropy-concentration}
    Under the assumptions of \thm{shannon-entropy-variance}, let $\eta=\eta(k,\epsilon)$ be the single-entry Kolmogorov distance of \thm{shannon-entropy-mean}, let $\delta=\delta(k,\epsilon)$ be as in \dfn{pair-distance}, and let $S_{\Haar}:=\Exp_{U\sim\mu}[S]$ be the expected Shannon entropy of a Haar random unitary.
    Then for every $t>0$,
    \begin{align}
        \Pr_{U\sim\mu(k,\epsilon)}\left[|S-S_{\Haar}|>t+\frac{4\sqrt{3\eta}}{e\ln 2}\right]
        \leq\frac{2\Var_{U\sim\mu}[S]+578\,\delta^{1/2}}{t^2}
        =O\left(\frac{1}{Nt^2}+\frac{\delta^{1/2}}{t^2}\right).
    \end{align}
\end{theorem}
\begin{proof}
    By \thm{shannon-entropy-mean}, $|\Exp_{\mu(k,\epsilon)}[S]-S_{\Haar}|\leq\frac{4\sqrt{3\eta}}{e\ln 2}$, so the event on the left-hand side is contained in $\{|S-\Exp_{\mu(k,\epsilon)}[S]|>t\}$, whose probability is at most $\Var_{\mu(k,\epsilon)}[S]/t^2$ by Chebyshev's inequality.
    By \thm{shannon-entropy-variance} and $(a+b)^2\leq 2a^2+2b^2$,
    \begin{align}
        \Var_{\mu(k,\epsilon)}[S]\leq\big(\Sd_{\mu}[S]+17\delta^{1/4}\big)^2\leq 2\Var_{\mu}[S]+578\,\delta^{1/2},
    \end{align}
    and $\Var_{\mu}[S]=\Theta(1/N)$ by \thm{wolpert-wolf}.
\end{proof}

For an $\epsilon$-approximate $k$-design without any further structure, \thm{approximate-design-upper-bound} and \lem{pair-projection} give $\eta=O(k^{-1/2}+\epsilon)$ and $\delta=O(k^{-1/2}+\epsilon)$, and by \thm{approximate-design-lower-bound} the former cannot be improved in general for $k\leq N$.
Hence the Shannon entropy of a single sample from the design lies within $t+O\big((k^{-1/2}+\epsilon)^{1/2}\big)$ of its Haar value $S_{\Haar}=n-\Theta(1)$ with probability at least $1-O\big((1/N+k^{-1/4}+\epsilon^{1/2})/t^2\big)$, which tends to one whenever $t\gg k^{-1/8}+\epsilon^{1/4}$ and $t\gg N^{-1/2}$.

Local invariance improves the first of these two parameters but not the second.
For an $m$-locally invariant design with $2^m\geq 2\pi k$ and $n\geq m+2$, \cor{locally-invariant-entry} gives $\eta=2^{-\Omega(k)}+O(\epsilon)$, so the shift $O(\sqrt\eta)$ of the center in \thm{shannon-entropy-concentration} becomes $2^{-\Omega(k)}+O(\epsilon^{1/2})$.
The pair-projection distance $\delta$, however, enters \thm{shannon-entropy-variance} through all $N^2$ pairs of outputs, and for a pair $x,x'$ that differ outside the invariant qubits, we do not know how to exploit local invariance beyond the moment condition, so \lem{pair-projection} remains the only bound on $\delta$ available to us.

\section*{AI statement}

The results of this paper were partly developed through conversations with Claude~Fable~5 and ChatGPT~5.6-Sol.
The manuscript was written by the author, drawing on information from AI-generated notes.
The technical contributions were checked and the manuscript proofread using Claude~Fable~5, Claude~Fable~5.1, and ChatGPT~5.6-Sol, and the paper was revised based on the issues found by these tools.
The author takes sole responsibility for the correctness of the paper.

\section*{Acknowledgments}

The author thanks Scott Aaronson for introducing the idea of adding a small Haar random unitary to make the sampling statistically close.
This generous sharing motivated the author to study the convergence of unitary designs and write this paper.
The author also thanks Nick Hunter-Jones for helpful discussions on unitary designs. 
This work is supported by the National Science and Technology Council in Taiwan, Grant~No.~114-2222-E-002-002-MY3, and by National Taiwan University, Grant~No.~NTU-CC-114L895503.

\bibliographystyle{plainurl}
\bibliography{ref}

\end{document}